\pdfoutput=1
\documentclass[11pt]{amsart} 

\newcommand{\todo}[1]{}
\newcommand{\vote}[1]{}
\usepackage[pdftex,
   pdftitle={Dynamic Planar Convex Hull},
   pdfauthor={Gerth Brodal, Riko Jacob},
   pdfsubject={Journal Draft $HeadURL$},
   pdfkeywords={computational geometry, convex hull, finger search
   trees, lower bounds},
   pdfpagemode=None,
   bookmarksopen=true,
   pagebackref]{hyperref}

\usepackage{color}
\definecolor{rltred}{rgb}{0.75,0,0}
\definecolor{rltgreen}{rgb}{0,0.5,0}
\definecolor{rltblue}{rgb}{0,0,0.75}

\usepackage{placeins}

\newtheorem{theorem}{Theorem}
\newtheorem{lemma}[theorem]{Lemma}
\newtheorem{corollary}[theorem]{Corollary}
\newtheorem{definition}[theorem]{Definition}

\def\paragraph#1{\textbf{* #1}\\}

\def\figfontsize#1#2{\tmp{9}{9}}
\newcommand{\inputfig}[1]{
    \let\tmp\fontsize%
    \let\fontsize\figfontsize%
    \input{#1}%
    \let\fontsize\tmp}

\usepackage{palatino}

\let\defword\emph

\usepackage{graphicx}
\usepackage{color}

\usepackage{amssymb}
\usepackage{latexsym}

\newcommand{\conclusion}{\\[1ex]\nopagebreak}

\newcommand{\UV}{\textup{\textrm{UV}}}
\newcommand{\UC}{\textup{\textrm{UC}}}
\newcommand{\UCo}{\textup{\textrm{UC$^{\circ}$}}}
\newcommand{\UH}{\textup{\textrm{UH}}}

\newcommand{\LE}{\textup{\textrm{LE}}}

\newcommand{\SC}[1]{\bar{#1}}

\newcommand{\IT}{\mathbb{T}}

\newcommand{\rX}[1]{r\!_{\SC{#1}}}
\newcommand{\rA}{\rX{A}}
\newcommand{\rB}{\rX{B}}

\newcommand{\CSP}{P}
\newcommand{\argmax}{\textrm{argmax}}

\newcommand{\cal}{}

\usepackage{pict2e} 

\newcommand{\Lray}[1]{\overleftarrow{#1}}
\newcommand{\Rray}[1]{\overrightarrow{#1}}
\newcommand{\LrayI}[1]{\smash{\overset{{%
\setlength{\unitlength}{0,18ex}%
\begin{picture}(10,6)(0,4.5)
\put(3,5){\line(1,0){6}}%
\put(3,5){\circle*{4}}%
\end{picture}}}{#1}}}
\newcommand{\RrayI}[1]{\smash{\overset{{%
\setlength{\unitlength}{0,18ex}%
\begin{picture}(10,1)(0,4.5)
\put(0,5){\line(1,0){6}}%
\put(8,5){\circle*{4}}%
\end{picture}}}{#1}}}

\newcommand{\LrayP}[1]{\smash{\overset{{%
\setlength{\unitlength}{0,18ex}%
\begin{picture}(10,6)(0,4.5)
\put(4.5,5){\line(1,0){5}}%
\put(2.5,5){\circle{4}}%
\end{picture}}}{#1}}}
\newcommand{\RrayP}[1]{\smash{\overset{{%
\setlength{\unitlength}{0,18ex}%
\begin{picture}(10,1)(0,4.5)
\put(0,5){\line(1,0){6}}%
\put(8,5){\circle{4}}%
\end{picture}}}{#1}}}

\newcommand{\DiSeP}{\textup{\textsc{DisjointSet$_{n,k}$}}}

\newcommand{\floor}[1]{{\lfloor{#1}\rfloor}}

\newcommand{\Rset}{\mathbb{R}}

\usepackage[ruled,
algonl,
vlined,algo2e]{algorithm2e}
 
\SetKw{KwGoto}{Goto}
\SetKw{KwLet}{Let}
\SetKw{KwSet}{Set}
\SetKw{KwAssume}{Assume}
\SetKw{KwAnd}{and}
\SetKw{KwOr}{or}

\SetKwBlock{Call}{call}{}

\SetCommentSty{textsl}
\def\com#1{\tcp{ #1}}
\DontPrintSemicolon
\SetAlgoLined

\newcommand{\pc}[1]{\textcolor{rltgreen}{\sc\lowercase{#1}}} 
\newcommand{\PC}[1]{\textcolor{rltgreen}{\sc\lowercase{#1}}} 

\makeatletter
\gdef\inwtype#1{\hfill\hbox{#1}&}
\gdef\inftype#1{#1}
\gdef\inarg#1{#1}
\gdef\innarg{\hfil\cr}%
\gdef\inlinearg{,}%
\gdef\inwret#1{#1~}
\gdef\infret#1{$#1$}
\gdef\away#1{}

\gdef\Tsituation#1{\textbf{Situation:} #1\\}
\gdef\Msituation#1{\textbf{Situation:}\;\Indp {\let\mynewline\; #1}\;\Indm}
\gdef\Tresult#1{\textbf{Result:} #1 \\}
\gdef\Mresult#1{\textbf{Result:}\;\Indp {\let\mynewline\; #1}\;\Indm}
\gdef\AlgSit#1{{\textbf{Pre:} {\def\mynewline{\;\quad} #1} } \; }
\gdef\AlgResult#1{{\textbf{Post:} {\def\mynewline{\;\quad} #1} }\; }

\long\gdef\defproc#1#2#3#4#5#6#7{
\expandafter\gdef\csname sit#1\endcsname{#5}%
\expandafter\gdef\csname res#1\endcsname{#6}%
\expandafter\gdef\csname pc#1\endcsname{{\PC{#2}}}%
\expandafter\gdef\csname pt#1\endcsname{%
  \let\wtype\inwtype\let\ftype\inftype\let\arg\inarg\let\narg\innarg%
  \let\wret\inwret\let\fret\infret%
  {#4}%
  {\PC{#2}%
  \(\left(\begin{matrix}#3\end{matrix}\right)\)}}%
\expandafter\gdef\csname st#1\endcsname{
  \let\wtype\away\let\ftype\away\let\arg\inarg\let\narg\inlinearg%
  \let\wret\away\let\fret\infret%
  {#4}%
  {\PC{#2}%
  \hbox{$\left(\begin{matrix}#3\end{matrix}\right)$}}}
\expandafter\gdef\csname stWith#1\endcsname##1##2{
  \let\wtype\away\let\ftype\away\let\arg\inarg\let\narg\inlinearg%
  \let\wret\away\let\fret\infret%
  {\(##1\)#4}%
  {\PC{#2}%
  \hbox{\(\left(\begin{matrix} #3 \end{matrix}\,=\, ##2\right) \)}}}
\expandafter\gdef\csname code#1\endcsname{
  \begin{algorithm2e}
  \def\\{\;\leftskip=0pt\skiptotal=0.5em}
  \def\>{\Indp}
  \caption{#2}
  \label{alg:#1}
  \BlankLine
  \hbox{\csname pt#1\endcsname}
  \BlankLine
  \let\sit\AlgSit\let\res\AlgResult\let\mynewline\;%
  #5 #6\BlankLine #7%
  \end{algorithm2e}
  \expandafter\gdef\csname code#1\endcsname{} 
}%
\immediate\write\allOut{\noexpand\code{#1}\noexpand\filbreak}}

\def\code#1{\ifcsname code#1\endcsname\csname code#1\endcsname \else \errmessage{no pseudocode for #1 provided}\fi}
\def\pt#1{\par\hbox{\ifcsname pt#1\endcsname\csname pt#1\endcsname\else \errmessage{no pseudocode for #1 provided}\fi}\par}
\def\pcwo#1{\hbox{\ifcsname pc#1\endcsname\csname pc#1\endcsname\else \errmessage{no pseudocode for #1 provided}\fi}}
\def\st#1{\hbox{\ifcsname st#1\endcsname\csname st#1\endcsname\else \errmessage{no pseudocode for #1 provided}\fi}}

\def\pcwoX#1#2{\hbox{\pcwo{#1}\((#2)\)}}

\def\FullCall#1#2#3{$\left[
    \begin{array}[c]{c}
      #2 = x(#3)\\
      \hbox{\ifcsname pt#1\endcsname\csname pt#1\endcsname\else \errmessage{no pseudocode for #1 provided}\fi}
    \end{array}
  \right]$}
\def\FullCallWith#1#2#3{%
{\let\sit\Tsituation\def\mynewline{; }\csname sit#1\endcsname}$ #2 = x(#3)$\\
\hbox{\ifcsname pt#1\endcsname\csname pt#1\endcsname\else \errmessage{no pseudocode for #1 provided}\fi}\\
{\let\res\Tresult\def\mynewline{; }\csname res#1\endcsname}\\
}
\def\FullCallNoArgs#1{%
{\let\sit\Tsituation\def\mynewline{; }\csname sit#1\endcsname}%
\hbox{\ifcsname pt#1\endcsname\csname pt#1\endcsname\else \errmessage{no pseudocode for #1 provided}\fi}\\
{\let\res\Tresult\def\mynewline{; }\csname res#1\endcsname}\\
}
\def\myCall#1#2#3{\Call({\hbox{#1 (Algorithm~\ref{alg:#3})}}){#2}}
\def\AlgFullCallWithOut#1#2#3#4{
  {\lnl{#4}\expandafter\myCall{\hbox{\ifcsname st#1\endcsname\csname st#1\endcsname\else \errmessage{no pseudocode for #1 provided}\fi}}{%
  {\let\sit\AlgSit\csname sit#1\endcsname}
  {\let\res\AlgResult\csname res#1\endcsname}
}{#1}}
}
\def\AlgFullCallWith#1#2#3#4{%
  {\lnl{#4}\expandafter\myCall{\hbox{\ifcsname st#1\endcsname\csname stWith#1\endcsname{#2}{#3}\else \errmessage{no pseudocode for #1 provided}\fi}}{%
  {\let\sit\AlgSit\csname sit#1\endcsname}
  {\let\res\AlgResult\csname res#1\endcsname}
}{#1}}
}
\def\ItemFullCall#1#2{%
\item Line~\ref{#2} %
  \st{#1}\\
  \def\mynewline{\\}%
  {\let\sit\Tsituation\csname sit#1\endcsname}
  {\let\res\Tresult\csname res#1\endcsname}
}

\def\XFullCall#1#2#3{\FullCall{#1}{}{#3}\textcolor{red}{$ #2 $}}
\def\XFullCall#1#2#3{$\hbox{\csname pt#1\endcsname}(#3)\textcolor{red}{#2}$}

\makeatother
\newwrite\allOut
\immediate\openout\allOut=allcode.tex
\input{pseudocode}
\immediate\closeout\allOut

\begin{document}


\title{Dynamic Planar Convex Hull}

\author{Riko Jacob%
}
\email{rikj@itu.dk}
\urladdr{http://www.itu.dk/people/rikj/}
\address{%
IT University of Copenhagen,
Rued Langgaards Vej 7, 
2300 Copenhagen S, Denmark
}
\thanks{Main part of the work done while staying at BRICS, Basic Research in Computer
  Science, Department of Computer Science, Aarhus University {\tt www.brics.dk}.   
  Partially supported by the
      Future and Emerging Technologies programme of the EU under
      contract number IST-1999-14186 (ALCOM-FT).
      Supported by the Carlsberg Foundation (contract number
      ANS-0257/20).}

\author{Gerth St{\o}lting Brodal%
}
\email{gerth@cs.au.dk}
\urladdr{www.cs.au.dk/\~\/gerth}

\address{MADALGO (Center for Massive Data Algorithmics - a Center of
  the Danish National Research Foundation), Department of Computer Science, Aarhus University,
  {\AA}bogade 34, 8200 Aarhus~N, Denmark.}


\vote{is ``by querying'' good?}
\begin{abstract}
  In this article, we determine the amortized computational complexity
  of the planar dynamic convex hull problem by querying.
  We present a data structure that maintains a set of~$n$ points in
  the plane under the insertion and deletion of points in
  amortized~$O(\log n)$ time per operation.
  The space usage of the data structure is~$O(n)$.
  The data structure supports extreme point queries in a given 
  direction, tangent queries through a given point, and queries for
  the neighboring points on the convex hull in $O(\log n)$ time.  
  The extreme point queries can be used to decide whether or not a given
  line intersects the convex hull, and the tangent queries to determine
  whether a given point is inside the convex hull.
  We give a lower bound on the amortized asymptotic time complexity
  that matches the performance of this data structure.

\textbf{Categories and Subject Descriptors:}
 E.1 [Data]: Data Structures; F.2.2 [Analysis of Algorithms and
 Problem Complexity]: Nonnumerical Algorithms and
 Problems---geometrical problems and computations 

\textbf{Additional Keywords:}
Planar computational geometry, dynamic convex hull, lower
bound, data structure, search trees, finger searches

\end{abstract}



\maketitle

\section{Introduction}
\label{sec:intro}

\footnotetext[1]{The results presented here have been published in the
  conference 
  version~\cite{FOCS02} of this article, and in the PhD-Thesis of Riko
  Jacob~\cite{RikoPhD}.}

The convex hull of a set of points in the plane is a
well studied object in computational geometry.
It is well known that
the convex hull of a static set of~$n$ points can be computed in 
optimal worst-case $O(n\log n)$ time, e.g., with Graham's
scan~\cite{graham72} or 
Andrew's vertical sweep line variant~\cite{andrew79} of it.
Optimal output sensitive convex hull algorithms are due to Kirkpatrick and
Seidel~\cite{kirkpatrick86} and also to Chan~\cite{chan96}, who
achieve $O(n\log k)$ time, where~$k$ denotes the number of vertices on
the convex hull.

In this article we consider the dynamic setting, where we have a set~$S$
of points in the plane that can be changed by the insertion and deletion
of points.
Observing that a single insertion or deletion can change the convex
hull of~$S$ by $|S|-2$ points, reporting the changes to the convex
hull is in many applications not desirable. Instead of reporting the
changes, one maintains a data structure that allows queries for points
on the convex hull, an approach taken in most of the previous work
(very clearly in~\cite{chan01}, but implicitly already
in~\cite{OvL81}). 
Typical examples of queries are the extreme point in a given
direction~$d$, the 
tangent line(s) to the hull that passes through a given point~$q$, whether or
not a point~$p$ is inside the convex hull, and the segments of the convex hull
intersected by a given line~$\ell$.
%
These queries are illustrated in Figure~\ref{fig:queries}.
Furthermore, we might want to report (some consecutive subsequence of)
the points on the convex hull or count their cardinality.

We restrict our attention to work with the upper hull of the set of points.
Maintaining the lower hull is completely symmetrical.
Together, the upper and lower hull clearly represent the convex hull of a set of points.

\begin{figure}[htb]
  \begin{center}
	\input{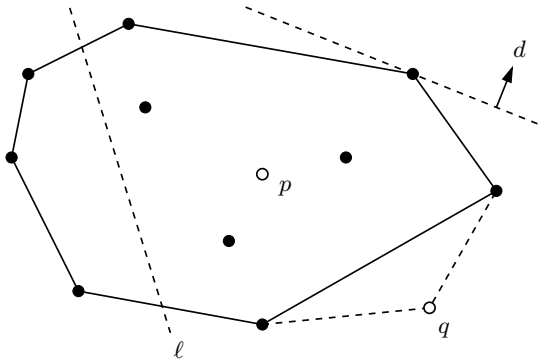}
    \caption[Different convex hull queries]
    {Different queries on the convex hull of a set of points.}
    \label{fig:queries}
  \end{center}
\end{figure}

Before summarizing the previous work on dynamic convex hull problems,
we state the main result of this article:
\begin{theorem}
\label{thm:fast}
  There exists a data structure for the fully dynamic planar convex
  hull problem supporting the insertion and deletion of points in
  amortized 
  $O(\log n)$ time, and extreme point queries, tangent queries
  and neighboring-point queries in 
  $O(\log n)$ time, where $n$ denotes the size of the stored point set
  before the operation.  The space usage is $O(n)$.
\end{theorem}

This upper bound is complemented by a matching lower bound:%
\begin{theorem}
\label{thm:CHLB}
  Assume there is a semidynamic insertion-only convex hull data
  structure on the real-RAM, that supports extreme point queries in
  amortized~$q(n)$ time, and insertions in amortized~$I(n)$, where~$n$
  is the size of the current set and~$q$ and~$I$ are non-decreasing functions.
  Then $q(n) = \Omega( \log n )$ and $I(n) = \Omega( \log(n/q(n)))$.
\end{theorem}

This lower bound implies that insertions have to take amortized $\Omega(\log n)$ time, as long
as queries take amortized time $O(n^{1-\varepsilon})$ for any $\varepsilon>0$.
Note that this result is stronger than just applying the well known $\Omega(n\log n)$
lower bound for the static convex hull computation
(as presented for example in the textbook by Preparata and
Shamos~\cite[Section~3.2]{preparata85}), since computing a static
convex hull using $n$ insertions and $n$ next-neighbor queries implies
in the dynamic setting only that the sum of the amortized running
times of insertions and next-neighbor queries is~$\Omega(\log n)$.

\subsection{Previous Work}

The dynamic convex hull problem was first studied by  
Overmars and van Leeuwen~\cite{OvL81}.
Their work gives a solution that uses~$O(\log^2 n)$ time per update
operation and maintains the vertices on the convex hull in
clockwise order in a leaf-linked balanced search tree.
Such a tree allows all of the above mentioned
queries in~$O(\log n)$ time.  The leaf-links allow to report~$k$
consecutive points on the convex hull (between two directions, tangent
lines or alike) in $O(\log n + k)$ time.

In the semidynamic variants of the problem, the data structure is not
required to support both insertions and deletions.
For the insertion-only problem Preparata~\cite{preparata79} gave
an~$O(\log n)$ worst-case time algorithm that maintains the vertices
of the convex hull in a search tree.
The deletion-only problem was solved by Hershberger and
Suri~\cite{hershberger92}, where initializing the data structure
(build) with~$n$ points and up to~$n$ deletions are accomplished in
overall~$O(n\log n)$ time.
Hershberger and Suri~\cite{hershberger96} also consider the off-line
variant of the problem, where both insertions and deletions are
allowed, but the times (and by this the order) of all insertions and
deletions are known a priori.
The algorithm processes a list of insertions and deletions in~$O(n\log
n)$ time and space, and produces a data structure that can answer
extreme point queries for any time using~$O(\log n)$ time (denoted queries in history in~\cite{hershberger96}).  Their data
structure does not provide an explicit representation of the convex
hull as a search tree.
%
The space usage can be reduced to~$O(n)$ if the queries are part of
the off-line information.

Chan~\cite{chan99,chan01} gave a 
construction for the fully dynamic problem
with~$O(\log^{1+\varepsilon}n)$ amortized time for updates (for any
constant $\varepsilon>0$), and $O(\log n)$ time for extreme point
queries.  His construction does not maintain an explicit
representation of the convex hull.  It is based on the general
dynamization technique of logarithmic rebuilding attributed to Bentley
and Saxe~\cite{bentley80}.  Using the semidynamic deletions-only
data structure of Hershberger and Suri~\cite{hershberger92}, and a
constant number of bootstrapping steps, the construction achieves
amortized update times of~$O(\log^{1+\varepsilon}n)$ for any constant~$\varepsilon>0$.
To support fast queries, the construction uses an augmented variant of
an interval tree storing the convex hulls of the semidynamic
deletion only data structures.
Brodal and Jacob~\cite{brodal00} and independently
Kaplan, Tarjan and Tsioutsiouliklis~\cite{tarjan01} improved the
amortized update time to~$O(\log n \log\log n)$.
The two results are based on very similar ideas.
The improved update time in~\cite{brodal00} is achieved by
reconsidering the framework of Chan~\cite{chan01}, and by constructing
a semidynamic deletion-only data structure that is adapted better to
the particular use.
More precisely, this semidynamic data structure supports build
in~$O(n)$ time under the assumption that the points are already
lexicographically sorted.
Deletions are supported in~$O(\log n\log\log n)$ amortized time.
%
All these dynamic data structures use~$O(n)$ space.

For arbitrary line queries, i.e., the problem of reporting the segments on the convex hull intersected by a query line $\ell$,
Chan~\cite{chan01} presented a 
solution with query and amortized update time $O(\log^{3/2} n)$ and
subsequently improved the bounds to  
$2^{O(\sqrt{\log\log n\log\log\log n})}\log n$~\cite{Chan12}.

Changing the model of computation towards points with integer coordinates, Demain and Patrascu determine the complexity of many dynamic planar convex hull related queries to be $\Theta(\log n /\log\log n)$~\cite{DemainePatrascu2007} if the updates take at most polylogarithmic time.

\subsection{Relation to Previous Work}
   
The data structure presented here continues the development of
\cite{chan01}, \cite{tarjan01}, and our preliminary construction
presented in~\cite{brodal00}.
It consists of an improved semidynamic, deletion-only convex hull data
structure 
that achieves amortized linear merge operations, and amortized
logarithmic deletions.
So far, merging two (or several) semidynamic sets was achieved by
building a new semidynamic data structure from scratch.
The new construction instead continues to use the already existing
data structures and adds a new top level data structure that
represents the merged semidynamic set.
This top level data structure, called \defword{merger}, operates on two sets of
points that are not necessarily horizontally disjoint.
It maintains the upper hull of the union of the two sets under
the deletion of points, assuming that the upper hulls of the two sets
are (recursively) maintained.

We augment the interval tree introduced in~\cite{chan01}, by
introducing the concept of \defword{lazy movements} of points between secondary structures
to reduce the number of expensive deletions from secondary structures.
A limited number of movements of points cannot be performed lazily, denoted \defword{forced moves}.
In the analysis we distinguish between two cases of forced moves, where we charge the expensive deletions in secondary structures to either the insertion or the deletion of points in the overall data structure. 

Important to our result and the result in~\cite{chan01} is the
repeated bootstrapping of the construction, where each bootstrapping step
improves the deletion time. In~\cite{chan01} the number of
bootstrapping steps  required increases when the $\varepsilon$ approaches zero in the update time~$O(\log^{1+\varepsilon}n)$, whereas in our construction we only need two bootstrapping steps.


\subsection{Structure of the Article}

Section~\ref{sec:ProblemSpec} discusses the precise setting
in which we consider the convex hull problem,
Sections~\ref{sec:overall}--\ref{sec:degenerate} describe the different parts of the data
structure, and Section~\ref{sec:LowerBounds} gives the complementary lower bound.
The sections explaining the data structures start in
Section~\ref{sec:overall} by showing the top-level decomposition
into subproblems, in particular the merger, a data structure capable
of merging two upper hulls.
Section~\ref{sec:KinHeap} gives the top-level construction for a
somewhat simpler case, namely the kinetic heap, partly as an
illustration of how to use the merger.
Section~\ref{sec:IT} and Section~\ref{sec:ext} describe the
 refined version of the interval tree and how it supports efficient queries.
The Sections~\ref{sec:SeparationCertificate}--\ref{sec:mainReestablish}
explain the details of the merger.
This discussion starts in Section~\ref{sec:SeparationCertificate} with
the geometric construction that we call a separation certificate, certifying that no intersection of the two hulls have been overlooked.
Its monotonicity property when points are deleted/replaced is explained in
Section~\ref{sec:monotonicity}.
These geometric considerations motivate a data structures that allows
a special kind of search, where every call to a search takes amortized
constant time.
Section~\ref{sec:helpers} introduces this data structure, and extends
it by the possibility to suspend the search if the geometric situation
is not yet conclusive.
Section~\ref{sec:Representation} discusses the representation of the merger, mainly the separation certificate, as a data structure.
Section~\ref{sec:algRepl} introduces the call structure and some of the algorithms, while leaving the core of the algorithmic task of reestablishing a certificate to Section~\ref{sec:mainReestablish}.
Section~\ref{sec:degenerate} considers how the data structure can be adapted to deal with degenerate input.
Section~\ref{sec:LowerBounds} discusses the matching lower bound an algebraic setting.
Finally, Section~\ref{sec:openProblems} concludes with some open problems.


%

\section{Problem Specification}
\label{sec:ProblemSpec}

\subsection{Fully Dynamic Convex Hull: Problem Definition}
\label{sec:ProblemSpecification}

\relax   

The focus of this article is a data structure to store a set~$S$
of points in the plane, that allows the following: 
\begin{description}
\item[Create] creates a data structure with~$S=\emptyset$.
\item[Insert$(p)$] Changes~$S \leftarrow S\cup \{p\}$, for a new point~$p$.
\item[Delete$(p)$] Changes~$S \leftarrow S\setminus \{p\}$, for an existing point~$p$..
\item[Query$(d)$] Return the extreme point of~$S$ in direction~$d$.
\end{description}

The query can be formalized by a scalar product, where we
consider $d$ and~$p$ as 2-dimensional vectors and hence return
$\argmax_{p\in S} d\cdot p$. 
Section~\ref{sec:ext} discusses how the data structure can support
tangent and arbitrary line queries.


%






%

\subsection{Computational Model}
\label{sec:CompModel}

In all our considerations,
we think of the geometric input as being given by precise points,
where the algorithms have access to some geometric primitives.
Namely, we assume that we can define new lines by two points, a new
point as the intersection of two lines, and that we can determine if a
point is above or below a line.
We also use a duality transformation (see Section~\ref{sec:duality}), but the algorithms actually do
not need to perform this transformation, they either work in the primal
world or in the dual world, never simultaneously in both.

The most commonly used model of computation in Computational Geometry
is the algebraic real-RAM, i.e., a RAM that has registers containing
real numbers that can be copied, added, multiplied, and compared at
unit cost.
The lower bound of Section~\ref{sec:LowerBounds} holds for this strong
model of computation.
In contrast, the presented algorithms and data structures do not use
the full power of the model.
They only use pointers to constant size records, i.e., no pointer
arithmetic or index calculations in arrays.
One could call this model a ``pointer machine over the reals''.
Additionally, our algorithms only evaluate constant degree polynomials
of input values.
Hence, our algorithms have the same asymptotic running time and space
usage if executed on an integer (word) RAM (allowing constant time addition and multiplication) that represents
rational values.

Throughout this article we will use the amortized analysis framework
of Tarjan~\cite{tarjan85} to analyse our data structures.

\subsection{Duality and Application to $k$-Levels}
\label{sec:duality}
   
There is a close connection between the upper hull of a set of points
and the lower envelope of the corresponding dual set of lines.
We define (as is standard, e.g., see~\cite[Section 8.2]{berg08}) the dual of a
point~$p=(a,b)\in\Rset^2$ to be the line $p^* \;:=\; \{(x,y)\mid
y=a\cdot x -b\}$, where~$a$ is called the slope.
For a set of points~$S$ the dual~$S^*$ consists of the lines dual to
the points in~$S$.
This concept is illustrated in Figure~\ref{fig:duality}.
%

\begin{figure}[tb]
  \begin{center}
    \input{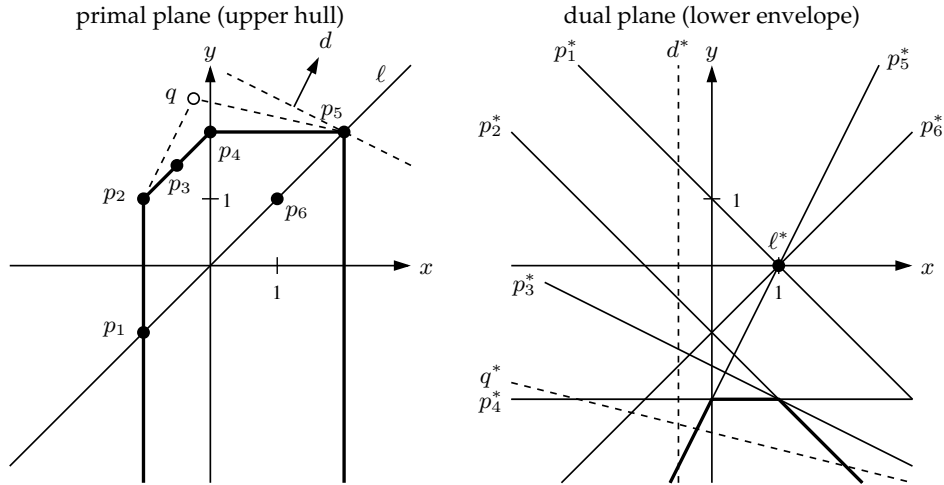}
    \caption
    {Duality of points, lines, and queries: Extreme point query $d$
      and its dual $d^*$; tangent query $q$ and its dual~$q^*$. 
      Both $p_1$ and~$p_3$ are not considered to be on the upper hull, and accordingly $p_1^*$ and~$p_3^*$ not as part of the lower envelope.}
    \label{fig:duality}
  \end{center}
\end{figure}

The essential properties of the duality transformation are captured in the lemma below, which follows from~\cite[Observation~8.3]{berg08} and is illustrated in Figure~\ref{fig:duality}.
Every non-vertical line in the plane is the graph of a linear function. 
For a finite set~$L$ of linear functions the point-wise minimum 
$m_L(t)=\min_{\ell\in L} \ell(t)$ is a piecewise linear function. The graph
of~$m_L$ is called the \defword{lower envelope} of~$L$, denoted $\LE(L)$.
A line~$\ell\in L$ is on the lower envelope of~$L$ if it defines one of
the linear segments of~$m_L$.
\begin{lemma}
\label{lem:dual}
  Let~$S$ be a set of points in the plane and $S^*$ the dual set of lines.
  A point~$p\in S$ is on the upper hull of~$S$ if and only if $p^*$ is
  on the lower envelope~$\LE(S^*)$ of~$S^*$. 
  The left-to-right order of points on the upper hull of~$S$ is the same as the
  right-to-left order of the corresponding segments of the lower
  envelope of~$S^*$.

  The extreme-point query on the upper hull of~$S$ in
  direction~$d=(-\alpha,1)$ 
  (the answer tangent line has slope~$\alpha$) is equivalent to
  evaluating $m_{S^*}(\alpha)$, the vertical line query~$d^*$ to the lower
  envelope of the lines in~$S^*$.
  A tangent-query point~$q$ is above the upper hull and has tangent
  points $a$ and $b$ on the upper hull, $a$ to the left of $b$,
  if and only if $q^*$ intersects the lines $b^*$ and $a^*$ on the
  lower envelope, and the intersection with $b^*$ is to the left of
  the intersection with $a^*$. If the query point $q$ is to the left (right) 
  of all points then $q$ only has one tangent point $a$,  $q^*$ will 
  only intersect $a^*$ on the lower envelope and $q^*$ is below the lower envelope for $x\rightarrow +\infty$ ($x\rightarrow -\infty$).
\end{lemma}

A dynamic planar lower envelope data structure is frequently understood as a
\defword{parametric heap}~\cite{basch99,tarjan01}, a generalization of a priority queue.
We think of the lines in~$S^*$ as linear functions that describe the
linear change of a value over time.
Instead of the \pc{find-min} operation of the priority queue, a parametric heap allows queries that evaluate~$m_{S^*}(t)$, i.e., 
vertical line queries reporting the stored linear function that attains the smallest value at time~$t$.
The update operations amount to the insertion and deletion of lines.
%
%
By Lemma~\ref{lem:dual}, 
the data structure we summarize in Theorem~\ref{thm:fast} allows
updates and queries in amortized~$O(\log n)$ time.

A \defword{kinetic heap}~\cite{basch97}
 is a parametric heap with the restriction
that the argument (time~$t$) of kinetic queries may not decrease between
two queries.
This naturally leads to the notion of a \defword{current time} for
queries.  In Section~\ref{sec:KinHeap} we describe a data structure
supporting kinetic queries in amortized $O(1)$ time and updates
in amortized $O(\log n)$ time.  
The previous best bounds were amortized $O(\log n)$ insert and $O(\log
n\cdot\log\log n)$ delete with $O(1)$ query time~\cite{tarjan01}.

Several geometric algorithms use a parametric or kinetic heap to store
lines.
In some cases the function-calls to this data structure dominate the
overall execution time.
Our improved data structure immediately improves such algorithms.
One such example is the algorithm by Edelsbrunner and
Welzl~\cite{welzl86} solving the \defword{$k$-level problem} in the
plane.
The problem is in the dual setting and is given by a set~$S$ of~$n$
non-vertical lines in the plane.  For every vertical line we are
interested in the $k$-th lowest intersection with a line of~$S$. 
The answer is given by a collection of line-segments from lines
of~$S$.
This generalizes the notion of a lower envelope~($k=1$) and an upper
envelope~($k=n$).  The situation is exemplified in
Figure~\ref{fig:2lev}. 
\vote{complexity of it, Davenport-Schinzel sequence?}

\begin{figure}[htb]
  \begin{center}
    \input{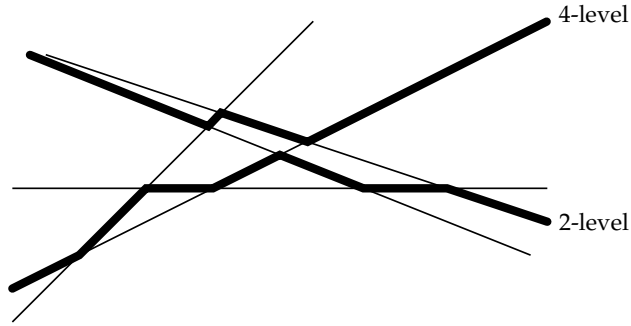}
    \caption[The 2-level and 4-level of 5 lines in the plane]
    {The 2-level and 4-level of 5 lines in the plane.  Note that the 2-level
      consists of 7 segments, two of the lines define two separate
      segments, whereas the 4-level only contains 4 segments, and one
      line is not contributing to the 4-level.}
    \label{fig:2lev}
  \end{center}
\end{figure}

As discussed by Chan~\cite{chan99level} we can use two
kinetic heaps to produce the $k$-level of a set of~$n$ lines
by a left-to-right sweepline algorithm maintaining the segments above and below the
$k$-level respectively, 
provided that the kinetic heap can provide the next time~$t$ when
the minimum changes.
If we have~$m$ segments on the $k$-level (the output size), the
algorithm using the data structure of the present article completes
(supporting arbitrary line intersection queries with the lower envelope)
in~$O((n+m)\log n)$ time.
This improves over the fastest deterministic algorithm by Edelsbrunner
and Welzl~\cite{welzl86}, that achieves $O(n\log n + m\log^{1+\varepsilon} n)$ time when using Chan's data structure~\cite{chan99level}. 
It is also faster than 
the randomized algorithm of Har-Peled and Sharir~\cite{sariel99} with
expected running time~$O((n+m)\alpha(n)\log{n})$,
where~$\alpha(n)$ is the slowly growing inverse of Ackerman's function.

\section{Overall Data Structure}
\label{sec:overall}

This section describes the overall structure of our solution, 
including the performance of the different components. 
In Section~\ref{sec:merger-interface} we describe the interface of 
the main novel contribution of this article, 
the \emph{merger} data structure, and state its performance as Theorem~\ref{thm:merger} below.
Sections~\ref{sec:SeparationCertificate}--\ref{sec:mainReestablish} give
the details of how to achieve Theorem~\ref{thm:merger}.
Section~\ref{sec:SemidynamicInterface} describes the application of our merger to maintain the upper hulls for a collection of sets of points, the so called \emph{Join-Delete} data structure, and Sections~\ref{sec:LogMethod}--\ref{sec:ext} describe how to combine this with the logarithmic method to answer kinetic heap and dynamic convex hull queries.



For a finite set~$S$ of points in the plane, the convex hull of~$S$
naturally decomposes into an upper and a lower part.  In
the remaining of the article we work with the upper hull only, the lower
hull is completely symmetrical.  To represent the convex hull we store
every point of~$S$ in two data structures, one that maintains access
to the upper hull, and one for the lower hull.  This allows us to
answer the mentioned queries on the convex hull.
To simplify the exposition, we extend the upper hull of~$S$ to also
include two vertical half-lines as segments, one extending from the
leftmost point of the upper hull of~$S$ vertically downward, and another one extending
downward from the rightmost point, see Figure~\ref{fig:duality}(left).
The vertices of the upper hull are denoted 
as $\UV(S)\subseteq S$, the set of points that 
are on segments and vertices of the upper hull
of~$S$ as~$\UH(S)$, the interior of the upper hull as~$\UCo(S)$,
and the upper hull together with its interior as~$\UC(S)=\UH(S)\cup \UCo(S)$---the upper
closure of $S$.


\subsection{Merger Data Structure}
\label{sec:merger-interface}

A central contribution of this article is a data structure denoted a
\defword{merger} that maintains the upper hull of the union of two
sets under point deletions, where (the changes to) the upper hulls of
the two participating sets is assumed to be known (because it is already computed recursively). 

In the following $A$ and $B$ denote the upper hulls of two point
sets. The deletion of a point can only be applied to a point
$r\in\UV(A\cup B)$.
Figure~\ref{fig:merger-example} gives an example of the deletion of a
point~$r$ in~$B$.
To describe the types of the arguments passed to the methods of a
merger, we adopt the ``locator'' terminology used by Goodrich and
Tamassia~\cite[Section~2.4.4]{GoodrichTamassia2002}). When passing a
``point'' we assume to pass a pointer to the \defword{base record} storing the $(x,y)$
coordinates of a point in $\Rset^2$ (the record can possibly also
store other information), whereas a ``locator'' is the abstraction of
the place of a point stored within the merger data structure.
A locator can only be given as an argument to the data structure and
be dereferenced to return the (pointer to the record of the) stored point.


\begin{figure}
  \begin{center}
    \input{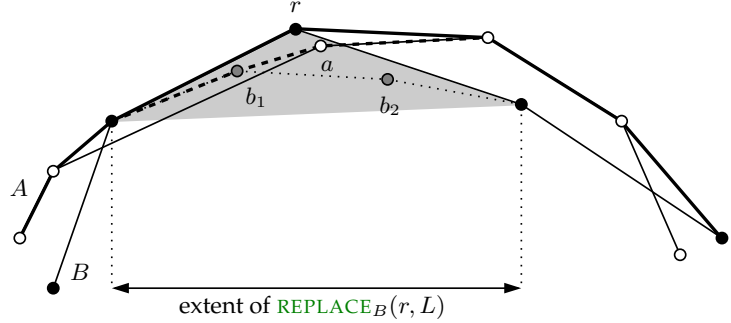}
    \caption{\pc{Replace}$_B(r,L)$: In $B$ (black points),
       the point $r\in B\cap \UV(A\cup B)$ is replaced with the list
       of points $L=(b_1,b_2)$. The merger reports that in
       $\UV(A\cup B)$, $r$ is replaced by the two points $(b_1,a)$ where
       the point $a$ is from $A$ (white points).}
    \label{fig:merger-example}
  \end{center}
\end{figure}

Our merger supports the following operations:

\begin{itemize}

\item \pc{Create}$(A,B)$ - Given two lists of points $A$ and $B$ in upper
  convex position, returns a list of locators to the points in
  $\UV(A\cup B)$.

\item \pc{Replace}$_B(r,L)$ - Given the locator to a point $r\in
  B\cap\UV(A\cup B)$ and a list $L=(b_1,b_2,\ldots)$ of points
  replacing $r$ in $B$, returns a list $L'$ of locators to points
  replacing $r$ in $\UV(A\cup B)$, i.e., the points 
  $\UV(A\cup (B\setminus\{r\}) \cup L) \setminus \UV(A\cup B)$.
\end{itemize}

\pc{Replace}$_A$ is defined analogously.  In the above it is required that
the lists of points (or locators to points) passed to/from the
operations is sorted w.r.t.\ the $x$-coordinates of the points. For
Replace$_B(r,L)$ it is required that replacing $r$ by $L$ in $B$,
$B$ remains left-to-right sorted, in upper convex position, and
$UC(B)$ shrinks (i.e., that $L$ is contained within the shaded
triangle in Figure~\ref{fig:merger-example}). 
The (vertical) slab between the two neighboring points of $r$ in $B$
is denoted the \defword{extent} of the \pc{Replace} operation.
Note that we first
return the locator to a point $p$ when $p\in \UV(A\cup B)$.
The following Theorem is proved in Sections~\ref{sec:SeparationCertificate}--\ref{sec:mainReestablish}.

\begin{theorem}
  \label{thm:merger}
  There exists a data structure supporting \pc{Create} in amortized time
  $O(|A|+|B|)$, and \pc{Replace}$_A$ and \pc{Replace}$_B$ in
  amortized time $O(1+|L|+|L'|)$. The data structure uses $O(|A|+|B|)$
  space.
\end{theorem}

\subsection{Join-Delete Data Structure}
\label{sec:SemidynamicInterface} 

In this section we show how to use the merger of the previous section to
construct a data structure maintaining the upper hulls for a
collection of sets of points under the following operations:

\begin{itemize}

\item \pc{MakeSet}$(p)$ - Given a point $p$, creates a new set $S=\{p\}$,
	and returns a reference to the
  set~$S$ and a locator to the point~$p$.

\item \pc{Join}$(S_1,S_2)$ - Given the references to two nonempty
  sets $S_1$ and $S_2$, creates the new merged set $S=S_1\cup
  S_2$. Returns a reference to the new set~$S$ and the list $L=\UV(S)$
  of  locators to the points on the upper hull of $S$ sorted
  left-to-right. The references to the two sets $S_1$ and $S_2$ are
  not valid any longer.

\item \pc{Delete}$(r)$ - Takes the locator to a point~$r$ in some
  set~$S$, and deletes $r$ from $S$.  
  If $r\in\UV(S)$ returns a left-to-right sorted list~$L$ of
  locators to the points replacing $p$ on the upper hull of $S$ ($L$
  is possibly an empty list). If $r\not\in\UV(S)$ returns
  that $r$ was not on the upper hull of~$S$.

\end{itemize}

In the following we describe how to compose mergers in binary trees 
to support the above operations within the following complexities.

\begin{theorem}
\label{thm:mergeDS}
  There exists a data structure storing a collection of point
  sets, supporting \pc{MakeSet} in time $O(1)$, and
  \pc{Join} and \pc{Delete} in amortized time
  $O(|S_1|+|S_2|)$ and $O(1+|L|)$, respectively. The data structure
  uses space linear in the number of \pc{MakeSet} operations.
\end{theorem}

\subsubsection*{Recursive construction}

Each set $S$ is represented by a binary tree, and a reference to a
set~$S$ is a pointer to the root of the tree representing~$S$. Each
leaf corresponds to either a \pc{MakeSet} operation, denoted a
\defword{\pc{MakeSet}-leaf}, or a list of points in upper convex
position sorted left-to-right, denoted an \defword{extracted
  set}.  For each \pc{MakeSet} operation performed we have
exactly one \pc{MakeSet}-leaf (also after the point has been
deleted) and for each \pc{Join} operation we create one
extracted set and two internal nodes, where each internal node stores a merger
(see Figure~\ref{fig:join-extractor}).  The sole role of extracted sets
is to achieve linear space.

Since mergers only allow the deletion of points that are present on
the merged upper hull, it is essential to our solution to
postpone the actual deletion of a point in a Join-Delete data structure until it shows up on the upper hull
of all the mergers it is stored in. For a point
set~$S$ we store both the points in $S$ and possibly some deleted
points~$D$ that correspond to \pc{MakeSet} leaves in the binary tree
for~$S$.  Each point $p\in S\cup D$ is stored in exactly one leaf:
either the leaf corresponding to \pc{MakeSet}$(p)$ or in some
extracted set~$X$.
Note that~$X$ is a child of an ancestor of the leaf \pc{MakeSet}$(p)$.

We let $L_{\ell}$ denote the list of locators to the
points stored at the leaf~$\ell$. Each \pc{MakeSet}-leaf stores at most
one point and each extracted set stores a possibly empty list of
points in upper convex position.
For each internal node~$v$ with  with two children $u_1$ and $u_2$ we let $S_v$ denote the set of points
stored at the leaves of the subtree rooted at~$v$. At the node~$v$ we
store the list $L_v=\UV(S_v)$ of locators to the points on the upper hull of $S_v$,
sorted left-to-right.
We maintain $L_v$
using the merger of Theorem~\ref{thm:merger} on the two upper hulls
$L_{u_1}$ and $L_{u_2}$, using the observation that
$$L_v
=\UV(S_v)
=\UV(S_{u_1}\cup S_{u_2})
=\UV(\UV(S_{u_1}) \cup \UV(S_{u_2}))
=\UV(L_{u_1}\cup L_{u_2})\;.$$

For each point~$p$ created with a \pc{MakeSet}$(p)$ operation we store
a record together with the information below. A pointer to this
\defword{base record} is the locator for the point~$p$ in the Join-Delete
data structure.

\begin{itemize}
\item the point~$p$,
\item a pointer to the leaf $\ell$ storing the point~$p$,
\item a bit indicating if the point has been deleted,
\item a bit indicating if the point is on the upper hull of the set
  it is stored in, and 
\item a list of locators to the point~$p$ in the mergers at the
  ancestors of the leaf $\ell$ storing $p$ (note that by the interface
  of a merger we only have locators for the nodes~$v$ where
  $p\in\UV(S_v)$).
\end{itemize}

\subsubsection*{Operations}

The three operations \pc{MakeSet}, \pc{Join}, and \pc{Delete} are implemented as follows.

A \pc{MakeSet}$(p)$ operation creates a new tree consisting of a
single \pc{MakeSet}-leaf $\ell$ with $L_\ell=(p)$.
It creates a base record for $p$ with a pointer to~$\ell$, marks~$p$ as not being deleted, marks~$p$ as being on the
upper hull of its set, and stores a locator to $p$ in $L_\ell$.
We return a reference to the leaf~$\ell$ as the reference for the new set, and a
reference to the new base record as the locator for $p$ in the Join-Delete
data structure.


\begin{figure}
  \begin{center}
    \input{join-extractor}
    \caption{\pc{Join}$(S_1,S_2)$.}
    \label{fig:join-extractor}
  \end{center}
\end{figure}

For a \pc{Join}$(S_1,S_2)$ operation we get pointers to the two
roots $u_1$ and $u_2$ of the two binary trees representing~$S_1$
and~$S_2$, respectively.
Now there is an easy and natural $O(n\log n)$ space solution where we use directly one additional merger to create the data structure for $S=S_1\cup S_2$.
The call to \pc{Join}$(S_1,S_2)$ returns the locators created by the merger of Theorem~\ref{thm:merger} and for points on the current upper hull of~$S$, adds them to the base record of the point.

To achieve linear space, we create a new leaf with an initially empty extracted set $X$ and two nodes $v$ and $v'$, each storing a merger,
such that $u_1$ and $u_2$ are the children of $v'$, and $v'$ and $X$
are the children of the new root~$v$ (see
Figure~\ref{fig:join-extractor}). We mark all points in $L_{u_1}$ and
$L_{u_2}$ as not being on the upper convex hull of their set,
construct $L_{v'}$ and $L_v$ using two new mergers (where $X$ is empty),
and mark every point~$p$ in $L_v$ as being on the upper convex hull
and add to the record of $p$ the locators of $p$ in the
mergers~$L_{v'}$ and $L_v$. To guarantee linear space we call
\pc{CheckExtractor}$(v)$ (described below). Finally we return a
reference to $v$ (representing the new set) and the list of
Join-Delete locators of the points in~$L_v$.
\vote{is this really how we should present it? empty X, call checkextractor?}

We implement \pc{Delete}$(r)$ by first marking $r$ as being deleted.
If $r$ is not marked as being on the upper convex hull of its
set we are done, since we have to postpone the deletion of~$r$.
Otherwise, we delete~$r$ from the structure as follows. The point $r$
is stored in the mergers on a path of nodes $v_0,\ldots,v_h\,$, where
$v_0$ is the leaf~$\ell$ storing $r$ and $v_h$ is the root of the
tree.  We first delete $r$ from the list $L_{v_0}$ at the leaf $v_0$.
For each of the mergers at the internal nodes along the path
$v_1,\ldots,v_h$ we update $L_{v_i}$ by calling
\pc{Replace}$_A(r_i,\bar{L}_{i-1})$ or
\pc{Replace}$_B(r_i,\bar{L}_{i-1})$, depending on if $v_{i-1}$ is the
left or right child of $v_i$. Here $r_i$ is the locator of $r$ in the
merger for $L_{v_i}$ (the locators $r_0,\ldots,r_h$ are stored in the
base record of~$r$) and $\bar{L}_{i-1}$ is the list of non-deleted points
determined to replace $r$ in $L_{v_i-1}$ ($\bar{L}_0=\emptyset$).  
We compute $\bar{L}_i$ from 
the list $\hat{L}_i$ of points returned by the call to \pc{Replace}
at~$v_i$ as follows. If none
of the points in $\hat{L}_i$ are marked as deleted then
$\bar{L}_i=\hat{L}_i$.  Otherwise, each point $p$ in $\hat{L}_i$ being
marked as deleted is now on the upper hull in all the mergers $p$ is
stored in, and can now be deleted. We recursively delete $p$ from all
the mergers storing $p$ in the subtree rooted at $v_i$, and let $\bar{L}^p$ denote the
list of non-deleted points replacing $p$ in the merger at~$v_i$
(computing this can again trigger the recursive deletion of points marked
deleted and calls to \pc{CheckExtractor}). In $\hat{L}_i$ we replace
$p$ by $\bar{L}^p$.  We repeat this recursive deletion of points
marked deleted in $\hat{L}_i$ until all points in $\hat{L}_i$ are
marked non-deleted.  We then let $\bar{L}_i=\hat{L}_i\subseteq L_{v_i}$.  For
each point $p\in\bar{L}_i$, the new locator of $p$ in $L_{v_i}$ is
appended to the base record of $p$.  If $v_i$ has even depth (i.e., the
node corresponded to a root after a \pc{Join} operation) we call
\pc{CheckExtractor}$(v_i)$.  Finally at the end of \pc{Delete}$(r)$
all points $p\in \bar{L}_h$ are marked as being on the upper
convex hull of their point set and we return the Join-Delete locators
of the points in $\bar{L}_h$, i.e., to all newly surfacing non-deleted points on the upper
convex hull of the set.
\vote{(hard to follow the recursive computation of $L_i$) try to improve?}


\subsubsection*{Maintaining Extractors}

To guarantee linear space we require that each internal node~$v$
with even depth, i.e., nodes where the right child is an extracted
set, satisfies the following two invariants. Let $X$ be the extracted
set of the right child of $v$ and $v'$ be the left child of $v$ (see
Figure~\ref{fig:join-extractor}).
\begin{eqnarray}
  X & \subseteq & L_v\;, \label{inv:extracted} \\
  2|L_v| & \leq & |L_{v'}| + 3|X|\;. \label{inv:propagated}
\end{eqnarray}
By the first invariant (\ref{inv:extracted}) we simply require that
$X\subseteq UV(S_v)$. The second invariant~(\ref{inv:propagated}) will
ensure total linear space as we prove below. The idea is that, if
(\ref{inv:propagated}) is not satisfied then we make $X=L_v$ 
(and therefore reestablish~(\ref{inv:propagated})) by
moving the set of points $X'=L_v\setminus X$ from $L_{v'}$ to $X$ by
recursively deleting the points from $S_{v'}$.

The operation \pc{CheckExtractor}$(v)$ checks if
(\ref{inv:propagated}) is violated for node~$v$. If the invariant is
not true we reestablish the invariant as follows.  First we identify
the set of points $X'=L_{v'} \cap L_v$ to be moved to $X$. We then
recursively delete each point $p\in X'$ from $S_{v'}$ as described
above (this can cause calls to \pc{CheckExtractor}$(w)$ for
descendants $w$ of $v'$). Note that we do not delete $p$ from the
mergers at the ancestors of $v'$. 
We then merge the lists $X$ and $X'$ and make the merged result a new list $X_{\textrm{\scriptsize new}}$. 
For the points from $X'$ we create new locators into $X_{\textrm{\scriptsize new}}$ and insert this locator as a replacement for the deleted
locators. Finally we create a new merger for $v$ by calling
\pc{Create}$(L_{v'},X_{\textrm{\scriptsize new}})$.  Note that now
$L_v=X_{\textrm{\scriptsize new}}$, but $L_v$ is still exactly the
same set of points, 
since
$S_v$ remains unchanged. We update the locators for all points in
$L_v$ to the new locators in~$L_v$.

\subsubsection*{Analysis}

\newcommand{\Lprop}{L_{v'}^{\textrm{\scriptsize prop}}}
\newcommand{\Lhide}{L_{v'}^{\textrm{\scriptsize hide}}}
\newcommand{\newLprop}{\underline{L}_{v'}^{\textrm{\scriptsize prop}}}
\newcommand{\newLhide}{\underline{L}_{v'}^{\textrm{\scriptsize hide}}}
\newcommand{\newX}{\underline{X}}
\newcommand{\newL}{\underline{L}}

In the following we prove that the performance of the Join-Delete data structure is as stated in Theorem~\ref{thm:mergeDS}. The correctness follows from that all points in a set are stored at exactly one leaf of the tree,  we bottom up maintain the upper hulls for each subset defined by a node of the tree, and all points at the root are not marked as deleted.

For the space usage we first note that if $n$ \pc{MakeSet} operations have been performed, 
then at most $n-1$ \pc{Join} operations can have been performed creating at most $2(n-1)$ internal nodes.
Since each \pc{MakeSet} and \pc{Join} operation creates a new leaf, 
the total number of leaves created is at most $2n-1$, i.e., the total number of nodes is at most $4n-3$.

To bound the space for the points in the mergers, 
consider an internal node $v$ at an even level, with left child $v'$ and right child an extracted set $X$. 
We partition the set $L_{v'}$ into two sets 
$\Lprop$ and $\Lhide$, where $\Lprop=L_{v'} \cap L_v$ are the points that are \emph{propagated} at the parent node $v$, and $\Lhide=L_{v'} \setminus L_v$ are the points that are \emph{hidden} at~$v$.  
Since $L_{v'}=\Lhide\cup\Lprop\subseteq L_v$ we have
$$
|L_v|\leq 2|L_v|-|\Lprop| \leq 3|X|+|L_{v'}|-|\Lprop| = 3|X|+|\Lhide|\;,
$$
where the second inequality follows from invariant~(\ref{inv:propagated}).
Since each point is contained in at most one $X$ and one $\Lhide$ set, we get a linear bound on the total number of points in all $L_v$ sets at internal nodes at even levels, implying a total $O(n)$ bound on the number of points stored in all mergers.

To analyze the amortized time for \pc{Join} and 
\pc{Delete} we define a potential $\Phi_p$ for each point $p$ stored in a tree. 
Let $d(p)$ be the depth of the leaf storing~$p$,
and $D(p)$ the depth of the topmost node where $p$ is on the upper hull of a merger.
We let $\Phi_p=d(p)+D(p)+1$ if $p$ is not marked deleted, and $\Phi_p=d(p)-D(p)+1$ otherwise.
We first analyze the cost of the operations without the calls to 
\pc{CheckExtractor}. 

The worst-case cost of \pc{Join}$(S_1,S_2)$ is $O(|S_1|+|S_2|)$ to construct two new nodes $v$ and $v'$ and the associated mergers. 
The depth of all existing nodes in the two trees increases by two, i.e., the potential of all points not marked deleted increases by four. 
The potential of points marked deleted remains unchanged.
Since the total increase in potential is  $4|S_1|+4|S_2|$, the amortized cost of \pc{Join} becomes $O(|S_1|+|S_2|)$ (without the cost for \pc{Check\-Extractor}). 

For \pc{Delete}$(r)$, marking $r$ deleted
decreases the potential of $r$ by $2D(r)$.
If $r$ is not on the upper hull of its set, we are done having used amortized $O(1)$ time.
Otherwise $r$ will be deleted from $d(r)$ mergers, which can recursively trigger other points to appear on the upper hull in a merger and being inserted in a sequence of parent mergers (with \pc{Replace}) and other points already marked as deleted to be deleted from all the mergers they are stored in. The cost for handling each point will be charged to each point's decrease in potential.
A point $p$ marked deleted that now gets removed from the data structure is deleted from the $d(p)-D(p)+1$ mergers storing $p$ (with \pc{Replace}). By Theorem~\ref{thm:merger} it is sufficient to charge each such deletion $O(1)$ work, i.e., total work $O(d(p)-D(p)+1)$. 
Since $p$ releases potential $d(p)-D(p)+1$, the work is covered by the decrease in potential. 
A non-deleted point 
$p$ that newly appears on the upper hull of the merger at a node~$v$ (output of \pc{Replace} at $v$) is inserted into the merger at the parent of $v$ (input to  \pc{Replace} at the parent), except if $v$ is the root. If $p$ in total newly appears on the upper hulls of $k$ mergers it is inserted in at most $k$ new mergers, i.e., by Theorem~\ref{thm:merger} it is sufficient to charge $O(k)$ work to $p$. Since $D(p)$ decreases by $k$, the point $p$ releases potential $O(k)$ to cover this work also. In summary all work by a \pc{Delete} operation (except the cost for \pc{Check\-Extractor} operations) is paid by the decrease in potential.

Finally we consider the work of \pc{CheckExtractor}$(v)$,  when 
invariant (\ref{inv:propagated}) is violated, i.e., $2|L_v|>|L_{v'}| + 3|X|$.
Again we will charge this work to the decrease in potential.  Let $v'$ be the left child of $v$ and define $\Lprop$ and $\Lhide$ as above. In the following we let underlined variables denote the variables after the \pc{CheckExtractor} operation, and non-underlined variables the values before the operation. E.g., we have $\newX=X\cup\Lprop$ and $\newLprop=\emptyset$. 
The primary work consists of extracting the points $\Lprop$ from $L_{v'}$ and all descendants of $v'$, and building the new $\newX$ and the merger at~$v$. 
The later costs worst-case $O(|\newX|+|\newL_{v'}|)$. 
Assume the depth of $v'$ is $\Delta$. 
Each point $p\in\Lprop$ is deleted from $d(p)-\Delta$ mergers. Since $p$ is moved to $\newX$, we have $\underline{d}(p)=\Delta$, and $p$ releases $d(p)-\Delta$ potential for deleting $p$ from all structures. 
The remaining work for all points that during the recursive deletion of a point $p$ are propagated or are deleted because they were marked deleted is charged to these points' change in potential as described above for \pc{Delete}.
Finally we argue that the worst-case cost $O(|\newX|+|\newL_{v'}|)$ can be charged to the potential released from the points in $\Lprop$ and the points 
$q\in\newLhide\setminus\Lhide$ that all decreased their $D(q)$ value by being inserted into the merger of $v'$.

From invariant (\ref{inv:propagated}) being violated we 
have $2|L_v|>|L_{v'}|+3|X|$, and 
$$|X|+|\Lhide| \leq 3|X|+|L_{v'}|-2|X|<2|L_v|-2|X|=2|\Lprop|\;.$$
Since $\newL_{v'}=\newLhide$ we have
$$
\begin{array}{rcl}
|\newX|+|\newL_{v'}|
& = & |X|+|\Lprop|+|\newLhide| \\[1ex]
& = & |X|+|\Lprop|+|\Lhide|+|\newLhide\setminus\Lhide| \\[1ex]
& < & 3|\Lprop|+|\newLhide\setminus\Lhide|\;.
\end{array}
$$
It follows that the remaining $O(|\newX|+|\newL_{v'}|)$ work can be charged to the potential released by the points in $\Lprop$ and $\newLhide\setminus\Lhide$.

\subsection{Global Rebuilding and Logarithmic Method}
\label{sec:LogMethod}

We use several, by now standard, techniques for dynamic data
structures that are explained in\linebreak[5] depth in~\cite{OvermarsPhD}.
We use the assumption that we know in advance the number~$n$ of points
stored in the data structure up to a constant factor. 
This is justified by the amortized version of the ``global
rebuilding'' technique from \cite[Theorem~5.2.1.1]{OvermarsPhD}.

We use the ``logarithmic method'' of Bentley and
Saxe~\cite{bentley80} to turn a semidynamic (deletion only) data
structure into a fully dynamic data structure.
We use this technique with a degree of~$\log n$. 
More precisely, we decompose the set~$S$ into subsets, called
\defword{blocks}~\cite{OvermarsPhD}.
Each block has a \emph{rank} and is stored in a semidynamic data structure
that maintains the vertices of the upper hull of the block explicitly
in a linked list. To delete a point, we delete it from its block.
For an insertion of~$p$, we create a new block of rank~0 containing only~$p$.
If we have $\log n$ blocks of the same rank~$r$, we join all of them
into one new block of rank~$r+1$.
In this way, there exist at most $\log n / \log\log n$ different
ranks, which is also an upper bound on the number of times a point can
participate in the joining of blocks.
Additionally, the total number of blocks is bounded by $\log^2 n /
\log\log n$.
These two bounds make it feasible to achieve fast queries by maintaining an interval tree as explained in detail in Section~\ref{sec:OutlineIT} and Section~\ref{sec:IT}.
For the (simpler) kinetic case discussed in Section~\ref{sec:KinHeap} the bounds allow the use of a single secondary structure.

\subsection{Kinetic Heap}
\label{sec:KinHeap}

This section describes a data structure implementing a kinetic heap
using the logarithmic method.
Recall that a kinetic heap stores a set~$S$ of linear functions and
has a current time~$t_c$.
It supports the following operations:

\begin{description}
\item[\pc{Insert}$(\ell)$] Inserts the line~$\ell$ into~$S$.
\item[\pc{Delete}$(\ell)$] Removes the line~$\ell$ from~$S$.
\item[\pc{kinetic-find-min}$(t)$] Evaluates $\min_{\ell\in S} \ell(t)$.
  Requires $t\geq t_c$ and sets $t_c:=t$.
\end{description}

In this section the Join-Delete data structure described in
Section~\ref{sec:SemidynamicInterface} 
is used in its dual form.
It stores lines and maintains the lower envelope of the lines. 
Vertical line queries return the line on the lower 
envelope intersected by a vertical line (see Figure~\ref{fig:duality}).

\begin{theorem}
\label{thm:KinHeap}
  There exists a kinetic heap implementation supporting \pc{Insert} and
  \pc{Delete} in amortized $O(\log n)$ time, and \pc{kinetic-find-min}
  in amortized $O(1)$ time.  The space usage of the data structure
  is~$O(n)$. 
  The parameter $n$ denotes the size of the stored set before the
  operation.
\end{theorem}

We achieve Theorem~\ref{thm:KinHeap} by bootstrapping twice 
with the following construction:
\begin{lemma}[Bootstrapping kinetic queries]\label{lem:KinBoost}
  Let~$D$ be a nondecreasing positive function.  Assume there exists a
  kinetic heap data structure
  supporting \pc{Insert} in amortized $O(\log n)$ time, \pc{Delete} in
  amortized~$O(D(n))$ time, and \pc{kinetic-find-min} in $O(1)$
  amortized time, where~$n$ is the total number of lines inserted.
  Assume the space usage of this data structure is~$O(n)$.\conclusion
  Then there exists a kinetic heap data structure
  supporting \pc{Insert} in amortized $O(\log n)$ time and \pc{Delete}
  in amortized $O( D(2(\log^2 n)/\log\log n)+\log n )$ time, and \pc{kinetic-Find-min} in amortized $O(1)$ time, where~$n$ is the total number of
  lines inserted.  The space usage of this data structure is~$O(n)$.
\end{lemma}

\begin{proof} (of Theorem~\ref{thm:KinHeap})
  Preparata's semidynamic insertion only data
  structure~\cite{preparata79} supports insertions in $O(\log n)$ time
  and deletions in $O(n\log n)$ time.  It maintains the segments of
  the lower envelope explicitly in a sorted list.  By charging the
  deletions $O(n)$ additional work to pay for advancing the kinetic
  search over all segments, we achieve amortized~$O(1)$ kinetic
  queries.  Using this data structure in Lemma~\ref{lem:KinBoost}
  (bootstrapping) we get~$O(\log n)$ amortized insertions, $O(1)$
  amortized queries and amortized deletion cost
  $$O\left(\frac{2\log^2 n}{\log\log n}\cdot\log\frac{2\log^2 n}{\log\log
    n} +\log n\right)=O(\log^2 n) \,.$$ 
  Bootstrapping one more time reduces the amortized deletion cost
  further to $O(\log^2(2(\log^2 n)/\log\log n)+\log n)=O(\log n)$. 
  Here $n$ is the total number of lines inserted. 
  To achieve the space bound of Theorem~\ref{thm:KinHeap} we globally rebuild the data
  structure whenever half of the inserted elements have been deleted 
  by repeated inserting the non-deleted lines into a new empty data structure. 
%
\end{proof}

\begin{proof}  (of Lemma~\ref{lem:KinBoost})
  We use the logarithmic method as explained in
  Section~\ref{sec:LogMethod} with degree~${\log n}$.
  This partitions the set~$S$ of input lines into at most $(\log^2
  n)/\log\log n$ 
  blocks, each stored in a Join-Delete data structure.
  Every line~$\ell$ that defines a segment of the lower envelope of its
  block gives rise to an \defword{activity interval}~$I_\ell$ by the projection
  of the segment to the time axis.
  This activity interval is implied by the neighboring lines on the envelope of the block and hence accessible by the Join-Delete data structure.
  The activity interval $I_\ell$ is the interval of query-values for which~$\ell$ is the correct
  \pc{kinetic-find-min} answer from its block.
  
  The Join-Delete data structure of Theorem~\ref{thm:mergeDS} 
  storing a block supports \pc{kinetic-find-min}
  queries in amortized $O(1)$ time:
  We remember the last answer, and for a new query we scan along the
  lower envelope of the block until we find the segment containing the new query
  time.
%
A line is \emph{scanned over} by a \pc{kinetic-find-min} query advancing time from $t$ to $t'$ if the activity interval $I_\ell\subset ]t,t'[$. 
Since a line can be scanned over at most once during all \pc{kinetic-find-min} queries (see Figure~\ref{fig:activity-interval}) we can charge this scanning to the \pc{join} operation creating the block.
\begin{figure}[htb]
  \begin{center}
    \input{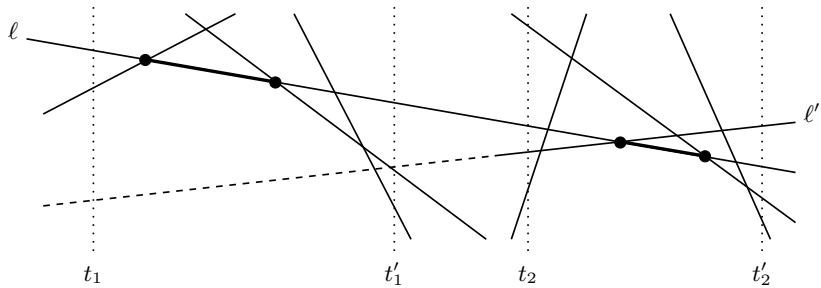}
    \caption{A line $\ell$ cannot be scanned over (bold line segments)
	       during two kinetic \pc{kinetic-find-min}
      operations on a block advancing the
      current time from $t_1$ to $t_1'$ and $t_2$ to~$t_2'$,
      respectively. The line $\ell'$ would hide the line $\ell$ during the
      time interval $[t_1,t_1']$ (dashed), since each block is a
      semi-dynamic deletion-only data structure.}
    \label{fig:activity-interval}
  \end{center}
\end{figure}
  The right endpoint~$e_\ell$ of the activity interval~$I_\ell$ of the
  current answer line~$\ell$ for a block gives the time until which~$\ell$ stays the
  correct answer.
  We refer to~$e_\ell$ as the \defword{(expiration) event} of~$\ell$.
  For the current time~$t$ the set~$L_t$ contains all lines that are
  answers for the \pc{find-min} queries on the up to~$(\log^2 n)/\log\log
  n$ Join-Delete data structures.
  The set of expiration events is stored in a priority queue~$Q$,
  implemented e.g. as a (2,4)-tree.
  This gives access to the earliest expiration event and the
  corresponding element in~$O(1)$ time, and allows for amortized
  $O(\log |L_t|)$ insertions, and amortized~$O(1)$ deletions
  of events (we only delete events that are stored).
  
  The set~$L_t$ is stored in a kinetic heap data structure~$H$
  with performance according to the assumption in the lemma.
  We call~$H$ the \defword{secondary structure}.
  For each block we store in~$H$ the line whose 
  activity interval contains the current time~$t$.
  In~$H$ we perform two types of deletions:
  If a line~$\ell$ gets deleted from the set~$S$ while stored in~$H$, we
  delete it from~$H$.
  This situation is called a \defword{forced deletion}.
  In contrast, if a line no longer is the current answer line for a
  block, but the line is still in $S$,
  we perform a \defword{lazy deletion}.
  Leaving lazily deleted lines in~$H$ does not influence the \pc{find-min}
  queries for the (new) current time, since the lines are still in~$S$.
  Using the idea of a periodic global rebuild, we can perform lazy
  deletions with the amortized cost of an insertion.
  To do this, a lazy deletion in~$H$ merely marks the lazily deleted line as
  deleted, but leaves it in $H$.
  Only if the number of deleted lines in~$H$ exceeds the number of not deleted
  lines, we perform a global rebuild, i.e., we create a new secondary
  structure storing only the current answer lines of the blocks.
  This build operation is paid by the lazy deletions.
  Additionally, the size of the secondary structure is at most twice
  the number of not-deleted lines in it, i.e. two times the number of blocks.
  
  For a \pc{kinetic-find-min} query for time~$t'>t$ do the following:
  Check if $t'$ is smaller than the earliest event in the priority
  queue~$Q$.
  If so, conclude $L_t=L_{t'}$, and there is no further change to the
  data structure.
  Otherwise, perform delete-min operations on $Q$ until
  the new minimum is larger than~$t'$.
  This identifies the set~$X$ of lines that should be lazily deleted
  from~$H$.
  For every block storing a line of~$X$, perform a
  \pc{kinetic-find-min} query for time~$t'$, leading to~$L_{t'}$.
  Insert the resulting line into~$H$ and the corresponding 
  expiration events into~$Q$.
  After these updates, $L_{t'}$ is represented in the secondary
  structure~$H$, and the query amounts to a \pc{kinetic-find-min} query for
  time~$t'$ in~$H$.
  
  The joining of $\log n$ blocks by the logarithmic method is
  implemented by $\log n-1$ \pc{join} operations arranged in a
  balanced binary tree of height $O(\log\log n)$.  The affected expiration
  events
  are removed from~$Q$, and lazy deletions are
  performed on the lines in~$H$.  The newly created Join-Delete data structure
  is queried with the current time and the answer-line is inserted
  in~$H$, and the expiration event into $Q$.
  
  A \pc{Delete}$(\ell)$ operation deletes~$\ell$ from the
  Join-Delete data structure it is stored in.
  Additionally, if~$\ell$ is currently stored in the secondary
  structure~$H$, 
  perform a forced deletion of~$\ell$ from $H$, and remove the corresponding 
  expiration event from~$Q$ if $\ell$ was the current minimum of a block.
  Perform a \pc{kinetic-find-min} query for the current time on the
  Join-Deleted data structure storing the block that changed.
  Insert the resulting line into~$H$.
  If $\ell$ was defining an expiration event in the priority queue
  (i.e., $\ell$ was the current minimum of a block or the line on the
  lower envelope immediately to the right of the current minimum), we
  update~$Q$ by inserting the expiration event of the current
  line for a \pc{kinetic-find-min} query in the block of $\ell$.

  Note that the priority queue $Q$ and the secondary structure~$H$ store
  up to~$\log^2 n/\log\log n$ and $2\log^2 n/\log\log n$ elements, 
  respectively, and that updates of $Q$
  and insertions into~$H$ take amortized~$O(\log\log
  n)$ time.  Every line pays at each of the at most $\log n / \log\log
  n$ levels of the logarithmic method one insertion into~$Q$, 
  one insertion into~$H$, and for participating in
  $O(\log\log n)$
  \pc{join} operations (Theorem~\ref{thm:mergeDS}).  
  This includes deletions from $Q$ 
  and lazy deletions from $H$.  It totals to~$O(\log n)$
  amortized time, charged to the insertion of the line.  The kinetic
  queries on the Join-Delete data structure are paid for by the cost of
  the \pc{join} operations, and totals to amortized~$O(\log n)$ per line.  A
  deletion pays for the forced deletion of one line in $H$
  and for querying and reinserting one event into~$Q$
  and a line into~$H$.  Hence, the
  data structure achieves the amortized time bounds claimed.
  The space bound follows from Theorem~\ref{thm:mergeDS}.  
  This concludes the proof of Lemma~\ref{lem:KinBoost}
\end{proof}

\subsection{Bootstrapping Extreme Point / Vertical Line Queries}
\label{sec:OutlineIT}   

In the following we consider how to use the logarithmic method to
support extreme point queries, i.e., in the dual vertical line queries.
Every block of the logarithmic method provides the vertices of the
upper hull of a block, i.e., in the dual the segments of a lower envelope,
in the following abbreviated as envelope
(see Figure~\ref{fig:dual-delete}).
The interval tree (as explained in Section~\ref{sec:IT}) 
together with the logarithmic method and the Join-Delete data structures allow the following bootstrapping construction, for which
Section~\ref{sec:IT} gives the proof:

\begin{figure}[htb]
  \begin{center}
    \input{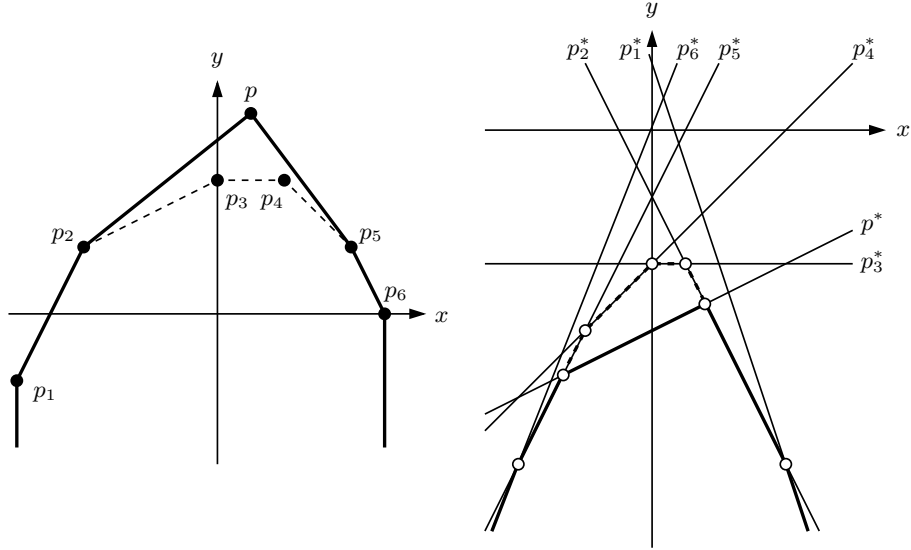}
    \caption[Deleting a point $p$ in the primal and the dual line
    deletion] {Deleting a point $p$ in the primal (left) and the
    corresponding deletion of the line $p^*$ in the dual (right). The
    two points $p_3$ and $p_4$ become visible on the upper hull, and
    correspondingly $p_3^*$ and $p_4^*$ become visible on the lower
    envelope. Note that on the lower envelope the visible part of the
    adjacent lines $p_2^*$ and $p_5^*$ is expanded.}
    \label{fig:dual-delete}
  \end{center}
\end{figure}

\begin{theorem}[Bootstrapping vertical line queries]  
\label{thm:boost}
  Let~$D$ be a nondecreasing positive function.
  Assume there exists a dynamic lower envelope data structure
  supporting \pc{Insert} in amortized~$O(\log n)$ time, \pc{Delete} in
  amortized~$O(D(n))$ time, and \pc{Vertical Line Query} in $O(\log
  n)$ time, with~$O(n)$ space usage, where~$n$ is the total number of
  lines inserted.
  \conclusion Then there exists a dynamic lower envelope data
  structure  supporting \pc{Insert} in amortized $O(\log n)$
  time, \pc{Delete} in amortized $O( D(\log^4 n)^2 + \log n )$ time,
  and \pc{Vertical Line Query} in $O(\log n)$ time, where~$n$ is the
  total number of lines inserted.
  The space usage of this data structure is~$O(n)$.
\end{theorem}

\begin{figure}[htb]
  \begin{center}
    \input{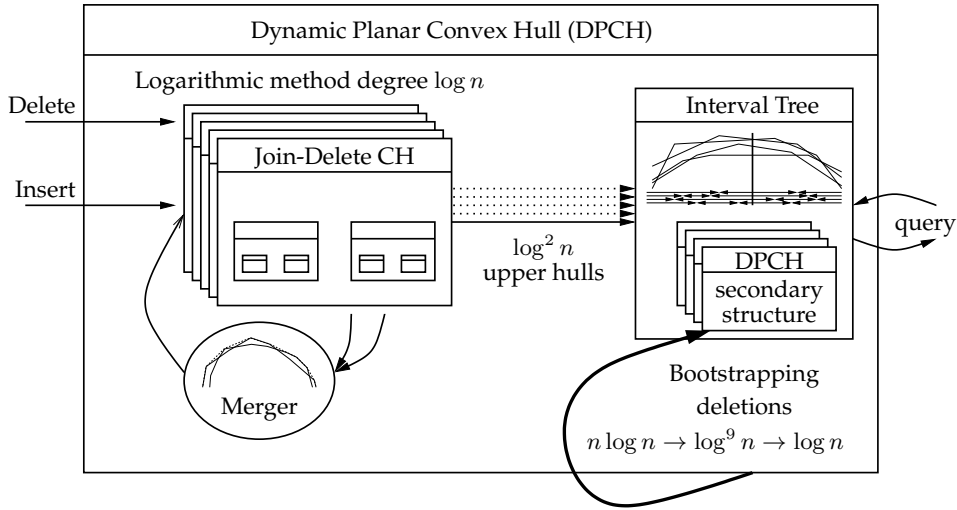}
    \caption
    {Illustrating the overall structure of the Dynamic Planar Convex Hull data structure.}
    \label{fig:hierarchy}
  \end{center}
\end{figure}
  

\begin{proof}  (of Theorem~\ref{thm:fast})
We use the bootstrapping of Theorem~\ref{thm:boost}
twice to prove our main
Theorem~\ref{thm:fast}.  The overall construction is summarized in
Figure~\ref{fig:hierarchy}.  For a first bootstrapping step we use
Preparata's data structure~\cite{preparata79} with $O(\log n)$ 
insertion and query time, and $D(n)=O(n\log n)$ deletion time, where deletions are supported by rebuilding the structure using insertions.  
This yields a data
structure with deletion time~
$$D(n)=O((\log^4 n\cdot \log(\log^4 n))^2+\log n)=O(\log^9 n)\;.$$
In a second
bootstrapping step we get a data structure with
$$D(n)=O((\log^9 (\log^4 n))^2 + \log n)=O(\log n)\;.$$

In the above, $n$ is the total number of lines inserted. 
To achieve the statement of Theorem~\ref{thm:fast}, where $n$ is the current number of lines, we globally rebuild the data structure when half of the inserted lines have been deleted by repeatedly (re-)inserting the current lines into a new empty data structure. 
\end{proof}


\subsection{Interval Tree: Extreme Point / Vertical Line Queries}
\label{sec:IT}

In this section we give a proof of Theorem~\ref{thm:boost}.
Similar to the construction in Section~\ref{sec:KinHeap} that uses the
Join-Delete data structure to implement a fast kinetic
heap we in this section use the
Join-Delete data structure to achieve a
fast parametric heap, or dually (Figure~\ref{fig:duality})
a fast fully dynamic convex hull data
structure supporting extreme point queries 
as stated in Theorem~\ref{thm:fast}.

Most of the ideas in Section~\ref{sec:KinHeap} for the
kinetic heap are used here again:
We use 
the logarithmic method with degree~$O(\log n)$,
ensuring that we at any time have $O(\log^2 n/\log\log n)$ 
lower envelopes,
and the concept of bootstrapping. 
What is naturally more complicated here is the ``on the fly
merging'' of the lower envelopes stemming from the semidynamic sets.
Where in the kinetic setting it is sufficient to store the lines
whose activity interval 
intersects the current time, now we
must
allow queries for an arbitrary
time~$t$.
Conceptually, we want to find the lines whose activity interval
intersects~$t$ and determine their lowest intersection with the
vertical line defined by~$t$.
This naturally suggests the use of an \defword{interval tree} to
organize several secondary structures.
This construction is made precise in Section~\ref{sec:ITDetails}, and
it achieves that each secondary structure stores only~$O(\log^4 n)$
lines.
Starting from a secondary structure that already has the aimed-at
performance of~$O(\log n)$ for insertions and queries, we get the speed
up for deletions
stated in 
Theorem~\ref{thm:boost}.



\subsubsection{Details of the Interval Tree}
\label{sec:ITDetails}

A traditional interval tree is a data structure that stores a set
of intervals in $\Rset$, and a query reports all intervals 
containing a given $x\in\Rset$.  This data structure is due to
Edelsbrunner~\cite{edelsb80} and McCreight~\cite{McCreight80} and is
described in detail for example in the textbook~\cite[Section
  10.1]{berg08}.  The central idea is to store the intervals at the
nodes of a balanced binary search tree having as leaves the
endpoints of the intervals.  An interval~$I$ is stored at exactly one
node of the tree, namely at the lowest common ancestor node of the two
leaves storing the endpoints of $I$. This guarantees that all
intervals containing a query value~$x$ are stored along the search path
to~$x$ in the tree.  

We take a variant~$\IT$ of an interval tree.  The underlying tree
structure is an insertion-only B-tree~\cite{actainf72bm} with degree
parameter~${\log n}$, i.e., the height of~$\IT$ is~$O(\log n/\log\log
n)$.  A line with a nonempty activity interval $I$ is stored exactly
at one common ancestor of the leaves of the endpoints of $I$ (not
necessary the lowest common ancestor). Lines with empty activity
intervals are allowed to be stored up to once in the tree at an
arbitrary node, these are denoted \defword{superfluous} lines.  For
every node of~$\IT$ we have as a secondary structure a fully dynamic
parametric heap for all the lines stored at the node.  This allows to
correctly answer \pc{Vertical Line Query}$(x)$ queries by collecting the answers
from the $O(\log n/\log\log n)$ secondary structures along the search path in~$\IT$ for~$x$ and taking
the minimum.

To bound the sizes of the secondary structures we move intervals among
the secondary structures of the 
nodes of the interval tree. We do this efficiently by adopting the
following concepts.  Every lower envelope is partitioned into
\defword{chunks} of $\Theta(\log n /\log\log n)$ consecutive activity
intervals.  The lower bound on the number of lines in a chunk might be
violated if
the lower envelope contains $O(\log n/\log\log\ n)$ intervals and all
intervals are stored in a single chunk.
For every
chunk~$c$ we determine its activity interval~$I_c$ as the union of the
activity intervals of the lines in~$c$.  The leaves of the interval
tree~$\IT$ store the endpoints of the activity intervals of the
chunks. For every activity interval~$I\subseteq I_c$ in a chunk~$c$
we define the \defword{canonical node} of $I$
in~$\IT$ to be the lowest common ancestor of the leaves storing the
endpoints of~$I_c$.  Every chunk has a pointer to its canonical node,
and every node has a linked list of the lines for which it is the
canonical node.  If a line is stored in a node above or below its
canonical node, we say that it is a \defword{lazily moved} line. Every
node has counters for the number of superfluous and lazily moved lines
stored at the node.

We maintain the invariant that in each secondary structures at most
half of the lines are superfluous or lazily moved, or
equivalently, at least half of the lines have this
node as their canonical node. This guarantees that the size of a
secondary structure is bounded by
$$O\left(2 \cdot \log n \cdot \frac{\log
  n}{\log\log n} \cdot \frac{\log^2 n}{\log\log
  n}\right)=O(\log^4 n)\;,$$ 
where the terms stem respectively
from the  superfluous and
lazily moved lines,
the degree of~$\IT$ (which bounds the number of chunks from one
semidynamic lower envelope with the same canonical node), the chunk size,
and the number of semidynamic sets (lower envelopes).

%
%

\subsubsection{Insert}
      
We handle the insertion of a new line as given by the logarithmic method. 
We first create a Join-Delete data structure storing the line as a
singleton set.  The line is a chunk by itself with
activity interval~$\Rset$ and the root of $\IT$ as its canonical node.
We insert the line in the root's secondary structure.
If no merging is required by the logarithmic method then we are done.

Otherwise, several Join-Delete data structures storing several  
sets of lines  are joined into one set by the
logarithmic method, and we replace several lower envelopes by one new lower
envelope as computed by the Join-Delete data structure. 
The new lower envelope is partitioned into chunks.  Before
updating the activity intervals of the lines, we insert the endpoints
of the activity intervals of the new chunks into the leaves of the
B-tree~$\IT$.  As in a standard B-tree, this can make it necessary to
split nodes of~$\IT$.  We split a node~$u$ by creating a new
sibling~$v$ of~$u$, and making it a new child of $u$'s parent~$w$,
possibly succeedingly requiring $w$ to be recursively split. The keys
and (for non-leaf nodes) the child pointers of $u$ are split
among $v$ and $u$. For all lines that have $u$ as their canonical
node, we recompute from the activity intervals of their (old) chunks
their new canonical node, which is either $u$, $v$ or~$w$. Finally we
destroy the associated secondary structure of $u$ and re-insert all
the lines from it at their canonical secondary structures, except for
superfluous lines which are not re-inserted.

After having updated $\IT$ we  identify the  new canonical nodes for
all new chunks and lines
participating in the joining.  Since
for each of the lines participating in the join the activity interval
can only shrink, the lines' current positions in~$\IT$ remain valid,
but the lines can become superfluous or lazily moved.  For each node
$u$ where the majority of lines now are superfluous or lazily moved we
perform a \defword{clean} operation: The secondary structure of $u$ is
destroyed and we re-insert all the lines from it in the secondary
structures of their canonical nodes, except for superfluous lines
which are not re-inserted.

\subsubsection{Delete}

To delete a line~$\ell\in S$ we first delete it from the
Join-Delete data structure, and
check if $\ell$ is stored
in a secondary structure. If it is not stored in a secondary structure we
are done. Otherwise we delete it from the secondary structure were it
is stored, and we are again done if $\ell$ was superfluous, i.e., it
had an empty activity interval because it was not on the lower
envelope of the semidynamic set containing~$\ell$. 

The interesting case is if $\ell$ was on the lower envelope of its
set (see Figure~\ref{fig:deletion-lower-enevelope}).
From the Join-Delete data structure 
we get a
piece~$L$ of the lower envelope that replaces~$\ell$.
Let $C$ be $L$ together with the chunks~$c_l$ and~$c_r$ containing
$\ell$'s left and right neighbors $\ell_l$ and $\ell_r$ on the lower
envelope. In the case $c_l$ and~$c_r$ are the same chunk and $C$
is below the minimal chunk size we include a neighboring chunk into~$C$.
We divide $C$ into new chunks of size $\Theta(\log n/\log\log n)$. 
The endpoints of the chunks are inserted at the
leaves of~$\IT$, possibly causing nodes to be split as described for
the insert operation.

\begin{figure}[tb]
  \begin{center}
    \input{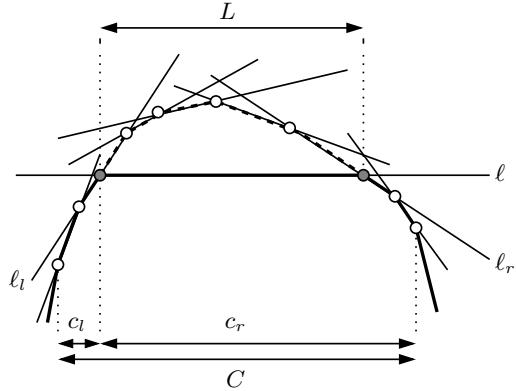}
    \caption[Deleting a line from a lower envelope]
    {Deleting a line~$\ell$ from a lower envelope.}
    \label{fig:deletion-lower-enevelope}
  \end{center}
\end{figure}

For every new chunk we determine the canonical node of~$\IT$.  For
every line~$h$ in~$C$ we assign it to its canonical node.  If $h$ is
already stored in~$\IT$ at a node spanning its activity interval, we
leave it in its secondary structure (if this is not its canonical node,
we are in the case that 
$h$ is lazily moved). If $h$ is not currently stored in a secondary
structure we insert it into the secondary structure of its canonical
node. The final case is when $h$ is stored at a node $u$ that is not
spanning its new activity interval. In this case~$h$ has to be stored
at a different node of~$\IT$.  We perform a \defword{forced move} of
$h$ by deleting $h$ from the  secondary
structure at~$u$ and inserting $h$ in the secondary structure of its
canonical node.

\subsubsection{Analysis}

We first analyze the cost of maintaining the Join-Delete data
structures. During $n$ insertions each line participates in at
most $O(\log n)$ Join operations, since we join sets  in a
perfect balanced binary tree 
(the joining of $O(\log n)$ sets in the logarithmic
method is performed by joining in a balanced binary tree manner). By
Theorem~\ref{thm:mergeDS}, the total cost for the Join operations 
becomes amortized $O(n\log n)$. For deletions
we by Theorem~\ref{thm:mergeDS} need to bound the total number of lines 
surfacing on the lower envelopes during all deletions.
Each line can surface on a
lower envelope at most once for each of the $O(\log n/\log\log n)$ levels
in the logarithmic method. 
By Theorem~\ref{thm:mergeDS} the total work for deletions in the Join-Delete data structures is amortized $O(n\log n/\log\log n)$. 
It follows that the total work for maintaining the Join-Delete data is $O(n\log n)$. 

To bound the size of $\IT$ we first consider the number of chunks
created by all insertions and deletions. Each chunk created is either
\textit{i}) the result of the joining in the logarithmic method,
\textit{ii}) a sequence of lines surfacing because of a deletion, or
\textit{iii}) one of the $O(1)$ chunks created at the ends of a
deletion region (can contain both newly surfacing lines and lines
already on the lower envelope of the set before the deletion). Since
each line can surface at most once on each of the $O(\log n/\log\log
n)$ levels of the logarithmic method, the total number of
chunks created in the first two cases is $O(n(\log n/\log\log n)/(\log
n/\log\log n))=O(n)$. There are at most $n$ deletions, i.e., the
number of chunks created by the last case is also $O(n)$. It follows
that the total number of chunks created is $O(n)$, which is a
bound on the total number of leaf insertions into~$\IT$. Since the nodes of
$\IT$ have degree $\Theta(\log n)$ it follows that $\IT$ has $O(n/\log
n)$ internal nodes, which is also a bound on the number of node splits.  
Since a node can be split in $O(\log n)$
time, if we ignore the cost of maintaining the secondary structures, the
cost of rebalancing $\IT$ becomes at most $O(n)$. 

For each of the $O(n)$ chunks created, we use time $O(\log n)$ to
identify the canonical node of the chunk in the interval tree, i.e.,
we spend $O(n\log n)$ time to identify the canonical nodes of all
chunks and lines to be inserted.

To analyze the cost of using the secondary structures, we need
to consider the number of insertions and deletions in the secondary
structures.  Since all secondary structures have size $O(\log^4 n)$,
our bootstrapping assumption in Theorem~\ref{thm:boost} implies that
the time for each insertion into and deletion from a secondary structure
is $O(\log(\log^4 n))=O(\log\log n)$ and $O(D(\log^4 n))$,
respectively. 

Line insertions into secondary structures are caused by the creation
of new chunks, clean operations, and node splittings. 
Since a total of $O(n)$ chunks are created, at most $O(n\log
n/\log\log n)$ lines are inserted by the creation of new chunks.  
For
the clean operations we bound the number of reinsertions at the
canonical nodes by the number of lazily moved lines and superfluous
lines.  These lines cause the clean operation and are removed from the
node during the clean operation and inserted at their canonical nodes.
Since each line can at most become superfluous once for each of the 
$O(\log n/\log\log n)$ levels of
the logarithmic method, the total number of superfluous lines
introduced is $O(n\log n/\log\log n)$.  Lazily moved lines are a
subset of the lines in the chunks created, i.e., we can also bound
these by $O(n\log n/\log\log n)$.

To bound the number of lines reinserted at a node due to a node split,
we first observe that the lowest three levels of $\IT$ have
$O(n/\log^i n)$ nodes, for $i=1..3$. On the remaining levels there is
a total of $O(n/\log^4 n)$ nodes. This is simultaneously a bound on
the number of node splits in $\IT$ on the different levels. We have
already argued that for any node at most $O(\log^4 n)$ lines can have
the node as their canonical node. For the nodes on the three lowest
levels we have that a node on level~$i$ has $O(\log^i n)$ chunk
endpoints stored in its subtree, and therefore it can only be the
canonical node of $O(\log^i n)$ chunks, i.e., the canonical node of
$O(\log^{i+1} n/\log\log n)$ lines.  The total number of segments
being reinserted at their canonical node due to node splits becomes
$$O\left(
\frac{n}{\log^4 n}\log^4 n
+ \sum_{i=1}^3 \frac{n}{\log^i n}\cdot \frac{\log^{i+1} n}{\log\log n}\right)
= O\left(n\frac{\log n}{\log\log n}\right)\;.$$
It follows that the total number of insertions into secondary
structures is $O(n\log n/\log\log n)$, taking total time $O(n\log n)$.

The deletion of a line from a secondary structure is either due to one
of the $O(n)$ explicit deletions of a line from the fully dynamic
parametric heap structure, or due to a forced move during the deletion
of a line.  Assume a superfluous or lazily moved line~$h$ is
forcefully moved due to the deletion of a line~$\ell$.  The line $h$
is stored at the current node~$v$, because on the lower envelope of
some block that contained~$h$, the activity interval of $h$ was
contained in the span of~$v$ and was defined by two lines $\ell_l$
and~$\ell_r$.  We say that $\ell_l$ and $\ell_r$ were the
\emph{witnesses} of~$h$.  Since $h$ is forcefully moved the activity
interval of $h$ must have increased, i.e., at least either $\ell_l$ or
$\ell_r$ must have been deleted at some point of time.

We bound the number of forced moves by the following accounting
argument. We assume that each forced move costs one \emph{coin}.  Each
line~$h\in \IT$ that is not on the lower envelope of its current set
either \textit{i}) has a coin assigned, or \textit{ii}) has two
witnesses assigned, i.e., lines $\ell_l$ and $\ell_r$ such that the
interval where $h=\min(\ell_l,h,\ell_r)$ is contained in the span of
the node where $h$ is currently stored. We let $W(\ell')$ denote the
set of lines for which $\ell'$ is a witness, and maintain the
following invariant: If $\ell'$ is stored in a block at level~$k$ of
the logarithmic method, then $|W(\ell')|\leq 2k$, and $\ell'$ is only
the witness for lines contained in the same block as~$\ell'$.

Whenever the joining in the logarithmic causes a line $h$ to be hidden
from the lower envelope, the two neighbors $\ell_l$ and $\ell_r$ of
$h$ \emph{just before the merge} become the witnesses of $\ell$. When
a node $\ell'$ advances one level in the logarithmic method, then
potentially $\ell'$ becomes furthermore the witness of its two
neighbors just before the merge. Whenever a line~$\ell'$ is deleted we
charge the deletion $|W(\ell')|\leq 2\log n/\log\log n$ coins, that
are distributed to the lines for which $\ell'$ was a witness. A line
with a pair of witnesses cannot require a forced move. Only nodes that
have a coin can potentially be forcefully moved. After the forced move
of a line, the line is now on the lower envelope of its set and does
not require a coin, i.e., each forced move releases a coin to pay for
the forced move.

To limit the number of coins charged to a single deleted line we
introduce a parameter~$b(n)$, where $0\leq b(n)\leq \log n/\log\log
n$, that allows us to shift some of the cost for forced moves to the
insertions.  We introduce $b(n)$ equally distributed \defword{barrier
  levels} of the logarithmic method.  After every
$\frac{\log n/\log\log n}{b(n)+1}$ level the next level is a barrier
level (if $b(n)=0$ there are no barrier levels).  At a barrier level,
all lines in the created block get a coin to pay for a forced move,
allowing the witness sets of all lines in the created block to be made
empty. This way every inserted line is charged up to $b(n)$ coins.  We
now have the invariant that if a line~$\ell'$ is stored in a block
$k$~levels from the last barrier level then $|W(\ell')|\leq 2k$.  The
deletion of a line now only has to be charged $O(\frac{\log n/\log\log
  n}{b(n)+1})$ coins.

The total number of forced moves during a sequence of $n$ line
insertions and $d$ line deletions can be bounded by $O(n\cdot
b(n)+d\cdot\frac{\log n / \log\log n}{b(n)+1})$.  The total cost for the
forced deletions is $O(D(\log^4 n)\cdot(n\cdot b(n)+d\cdot\frac{\log n /
  \log\log n}{b(n)+1}))$, which is bounded by $O(n\log n+d\cdot
D(\log^4 n)^2)$ by setting $b(n)=\lfloor\log n/D(\log^4 n)\rfloor$.

Summarizing, $n$ line insertions and $d$ line deletions require total
time $O(n\log n+d\cdot D(\log^4 n)^2)$.  
Queries take time $O(\log(\log^4 n)\cdot\log n/\log\log n)=O(\log n)$.
This yields the amortized
time bounds claimed in Theorem~\ref{thm:boost}.

\subsubsection{Space Usage}

To analyze the space usage of the data structure we first observe that
Theorem~\ref{thm:mergeDS} implies that all Join-Delete data
structures require $O(n)$ space. In the interval tree we create at
most $O(n)$ leaves, i.e.,
the interval tree uses~$O(n)$ space.
Every line is
stored in at most one secondary structure, where it by our
bootstrapping assumption in Theorem~\ref{thm:boost}
 uses~$O(1)$ space.
This totals to a space usage of~$O(n)$.
This finishes the proof of Theorem~\ref{thm:boost}.

\subsection{Tangent / Arbitrary Line Queries}
\label{sec:ext}

The 
interval tree in Section~\ref{sec:IT}
can also be used to support the following query: Given a
point~$p \notin\UC(S)$, what are the two common tangent
lines of~$p$ and~$S$?  Translating the query into the dual setting, we
ask for the two intersection points of an arbitrary line~$p^*$\vote{good?} with the
lower envelope of~$L_{S^*}$ (see Figure~\ref{fig:duality}).
The geometric principles of how to use the interval tree
are the very same as of Chan~\cite{chan01}.  

Given that we may perform vertical line queries (Section~\ref{sec:IT}), we can easily verify
a hypothetic answer.  Therefore it is sufficient to consider the
situation under the assumption that the line~$\ell$ intersects the lower
envelope.  In the case that it does not intersect the lower
envelope, the verification by vertical line queries will fail.
  
We consider the dual problem of finding the right intersection of
line~$\ell$ with the lower envelope~$\LE(S)$ for a set of lines~$S$, i.e.,
we ask for the line of~$S$ with a slope strictly smaller than~$\ell$
that intersects~$\ell$ furthest to the left.  The following lemma, illustrated in Figure~\ref{fig:tangent-lemma},
provides the geometric argument that allows us to use queries in the
secondary structures to navigate in the interval tree~$\IT$.  
For a node in the interval tree, $S'\subseteq S$ will be 
the lines in the secondary structures from the node to the root
and $t_1$ and~$t_2$ are the vertical lines 
stemming from the search keys at the parent node in
the interval tree.

\begin{figure}[htb]
  \begin{center}
    \input{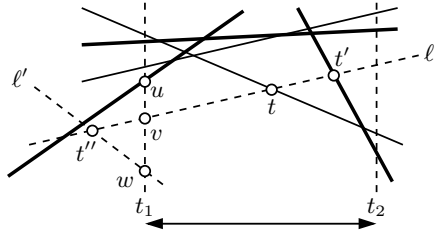}
    \caption{Illustration of proof of Lemma~\ref{lem:tangent}, where $S$ are
      the solid lines and $S'\subseteq S$ are the bold lines.}
    \label{fig:tangent-lemma}
  \end{center}
\end{figure}

\begin{lemma}
  \label{lem:tangent}
  Let~$t_1$ and~$t_2$ be two vertical lines, $t_1$ to the left of~$t_2$.
  Let~$S'\subseteq S$ be a subset of the lines~$S$ such that the lower
  envelope of~$S'$ at~$t_1$ and~$t_2$ coincides with the lower envelope
  of~$S$.  Assume that a query line~$\ell$ has a right intersection point~$t$ with 
  the lower envelope
  of~$S$, and that an arbitrary line query for~$\ell$ on~$S'$ results in
  the right intersection point~$t'$.  If~$t'$ lies between~$t_2$ and~$t_2$
  then~$t$ lies between~$t_1$ and~$t_2$.
\end{lemma}
\begin{proof}
  By the definition of the right intersection points 
  $t$ and $t'$ as the leftmost
  intersection of~$\ell$ with the lines of smaller slope in~$S$
  and~$S'$, respectively,
  we immediately have that~$t$ is not to the right of~$t'$ and
  hence not to the right of~$t_2$.
    
  Assume first that $\ell$ has also a left intersection with~$\LE(S')$ 
  between~$t_1$ and~$t_2$.
  If intersections of~$\ell$ with~$\LE(S)$ exist, they are
  between the intersections of~$\ell$ with~$\LE(S')$.
  In this case the lemma holds.
  Otherwise we know that the intersection~$v$ of~$\ell$ with~$t_1$ is
  below the intersection~$u$ of~$\LE(S')$ with~$t_1$.
  Assume by contradiction that~$t$ is to the left of~$t_1$.
  Let~$\ell'$ be a line of~$S \setminus S'$ that contains~$t$ (denoted~$t''$
  in Figure~\ref{fig:tangent-lemma}) and has
  smaller slope than~$\ell$.
  Then the intersection~$w$ of~$\ell'$ with~$t_1$ is below~$v$.
  But then the lower envelope of~$S$ intersects~$t_1$ at or below~$v$,
  contradicting the statement that the two lower envelopes coincide
  in $u$ on~$t_1$.
\end{proof}
  

Using this lemma, we can process an arbitrary line query for~$\ell$ in
the following way. Consider the search for the rightmost intersection
with the lower envelope. Starting at the root, we perform the query
for~$\ell$ at the secondary structure~$H$ at the root.  If we find
that there are no rightmost intersection of~$\LE(H)$ with~$\ell$, 
we proceed to the rightmost slab
(since $\ell$ has smaller slope than all lines in $H$).
Otherwise let~$u$ be the right intersection point found.
We find the slab that is defined by the keys stored at the root that
contains~$u$.  This slab identifies a child~$c$ of the root.  We
continue the search at~$c$ in the same way, that is, we perform
another arbitrary line query to the secondary structure at~$c$.  Now
we take the leftmost of the two results (stemming from~$c$ and the
root) and use it to identify a child of~$c$.  This process we continue
until we reach the leaf level.  There we take the leftmost of all
answers we got from secondary structures, and verify it by performing
a vertical line query.  It is necessary to use the currently leftmost
answer as we allow lines to be stored higher up in the tree.  It is
necessary to verify the outcome, since we cannot exclude the case that all
queries to secondary structures find two intersection points, whereas
there is no intersection of~$\ell$ and~$\LE(S)$.
 
For the bootstrapping we observe that Preparata's data structure~\cite{preparata79}
supports arbitrary
line queries in~$O(\log n)$ time.
The arbitrary line query in~$\IT$ performs precisely one query to a
secondary structure at every level of~$\IT$.
Hence the overall time bound is~$O(\log n)$, just like the
extreme-point query.


\subsubsection{Further Queries}

If we run a tangent query for a point~$p\in S$, we can determine whether~$p\in
\UV(S)$ and if this is the case, we find the neighbors of~$p$
in~$\UV(S)$ (with tie-breaking formulated carefully).  We can also
realize if~$p$ is a point on a segment $(p_1,p_2)$ of~$\UH(S)$
(in the dual $p^*$ passes through the lower envelope 
at the intersection of $p_1^*$ and $p_2^*$).  Tangent queries
allow us to report a stretch of~$k$ consecutive points on the upper
hull of~$S$ in time~$O(k\cdot \log n)$.
This is by an~$O(\log n)$ factor slower compared to an explicit representation of the convex
hull in a linear list of a leaf-linked search tree, like in~\cite{OvL81}.

\section{Geometry of Merging: Separation Certificate}\label{sec:SeparationCertificate} 

In this section we introduce the geometric concepts of certificates used by
our merger to achieve the performance stated in
Theorem~\ref{thm:merger} in Section~\ref{sec:merger-interface}.
The purpose of a merger is to efficiently maintain the upper hull of
the union of two upper hulls $A$ and $B$ that change over time, 
where the change is limited to replacing a point by a list of points such
that the upper hulls do not get bigger (see Figure~\ref{fig:merger-example}).
In the following sections we 
describe how to efficiently maintain
$\UV(A\cup B)$ from $A$ and $B$.

The construction is heavily based on the idea of maintaining a 
geometric certificate
asserting that all
equality points (intersections) between the two hulls $\UH(A)$ and $\UH(B)$ are identified.
Between two identified equality points of $\UH(A)$ and $\UH(B)$ we have a
\defword{separation certificate}
asserting that the two hulls are separated.
Sections~\ref{sec:Streaks}-\ref{sec:TangentCertificates} introduce the different components of a certificat, and 
Section~\ref{sec:CompleteCertificate} summarizes the complete certificate.
The algorithmic challenge is to ``repair'' the certificate  
after the deletion of a point and adapt
it to the changed geometric situation, which is discussed in Sections~\ref{sec:monotonicity}-\ref{sec:mainReestablish}.
%

For simplicity we in the following sections assume that input
points in~$A\cup B$ are in general position, i.e., 
no two points have the same $x$-coordinate, 
no three input points are on a line, and
no three lines defined by pairs of input points meet in one point.
Section~\ref{sec:degenerate} describes how to modify the algorithm for
input that violates these assumptions.

\subsection{Equality Points and Streaks}
\label{sec:Streaks}
   
A point on both upper hulls $p\in \UH(A) \cap \UH(B)$ is called an
\defword{equality point}.
By the general position assumption, equality points are
isolated points and every equality point is the intersection of one segment
of~$\UH(A)$ and one segment of~$\UH(B)$.

We call the slab between two neighboring equality points
a~\defword{streak}, and the horizontal span of the streak the
\emph{extent} of the streak.
For any two vertical lines $\ell_1$ and~$\ell_2$ within one streak, the
intersection of~$\ell_1$ with $\UH(A)$ is above the intersection
of~$\ell_1$ with~$\UH(B)$ if and only if the intersection of~$\ell_2$ with
$\UH(A)$ is above the intersection of~$\ell_2$ with~$\UH(B)$.
This defines the \defword{polarity} of the streak ($A$~over~$B$, or
$B$~over~$A$).
The situation is illustrated in Figure~\ref{fig:simp_equ}.

\begin{figure}[tb]
  \begin{center}
    \input{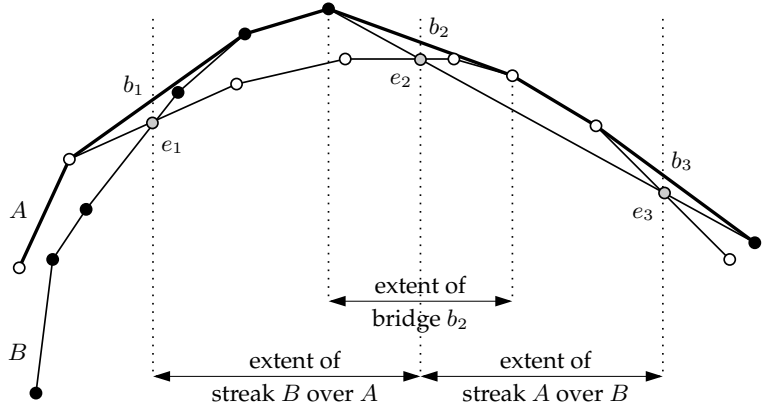}
    \caption
    {Merging two hulls ($A$ white nodes and $B$ black nodes) leads to
      equality points $e_i$ (grey nodes), streaks and bridges~$b_i$.}
    \label{fig:simp_equ}
  \end{center}
\end{figure}


\subsection{Bridges}
\label{sec:bridges}

A segment of~$\UH(A\cup B)$ that has one endpoint in~$A$ and the
other endpoint in~$B$ is called a \defword{bridge}
and the horizontal span of the bridge the
\emph{extent} of the bridge (see Figure~\ref{fig:simp_equ}).
Such a bridge has the property that it is a tangent to both
$\UH(A)$~and~$\UH(B)$.
Every tangent to both $\UH(A)$~and~$\UH(B)$ is a bridge, and it is a
local condition on the two endpoints and their neighbors in~$A$
and~$B$ that a segment is indeed a bridge.
Furthermore, if the points are in general position, there is a
one-to-one correspondence between equality points and bridges.

\subsection{Shortcuts}
\label{sec:sc_geom}

\relax

To simplify the locally outer hull we introduce
\defword{shortcuts}.
We consider a non-vertical line~$\ell$ and the closed half-plane~$h_\ell$
below~$\ell$ and replace in our considerations~$\UC(B)$ with $\UC(B) \cap
h_\ell$.
This situation is exemplified in Figure~\ref{fig:shortcut}.
For the line~$\ell$ we determine the \defword{cutting
  points}~$\ell\cap\UH(B)$ and call the line-segment~$\ell\cap\UC(B)$
between the cutting points a \defword{shortcut} extending
(horizontally) between the two cutting points.
Shortcuts reduce the number of segments on the locally outer
hull by ``cutting away'' the part of the convex hull that is
above~$\ell$.
More precisely, we define for a set~$H_B$ of non-vertical lines
the \defword{shortcut version} of~$B$ to be the
set of points
  $$ \SC{B} = \UV\Big(\UC(B) \cap \bigcap_{\ell\in H_B} h_\ell \Big)\;.$$
%
Similarly we define $\SC{A}$ for a set of shortcuts $H_A$ on $A$.

\begin{figure}[thb]
  \begin{center}
    \input{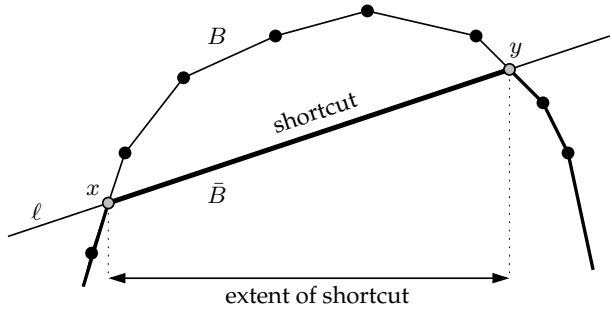}
    \caption{A shortcut defined by the line~$\ell$ with cutting points
      $x$~and~$y$.  The resulting shortcut version~$\SC{B}$ of~$B$ is
      indicated with bold line segments.}
    \label{fig:shortcut}
  \end{center}
\end{figure}

A set~$H_B$ of shortcuts on $B$ is \defword{conservative} 
w.r.t.~$A$ if $H_B$ is above~$\UH(A)$.
Two shortcuts $\ell_1$~and~$\ell_2$ are
\defword{non-overlapping} 
if their extent is disjoint.
An \defword{effective} shortcut has at least one vertex of $\UV(B)$ above it.
We will only use non-overlapping sets of effective and conservative shortcuts.

\begin{lemma}
\label{lem:FewShortcuts}
  Let $\overline{xr}$ and $\overline{ry}$ be two consecutive segments of
  $\UH(B)$ and $H_B$ a set of non-overlapping effective shortcuts
  on~$B$.
  Then at most three lines of~$H_B$ intersect
  $\overline{xr}$ and $\overline{ry}$.  
\end{lemma}

\begin{proof}
  The cutting points of an effective shortcut need to be on different
  segments.
  Hence, there is at most one shortcut cutting away~$r$ and at most
  two more intersecting $\overline{xr}$ and~$\overline{ry}$.
\end{proof}

All the shortcuts we use will be defined by lines through two input points 
(possibly points that have been deleted).

\subsection{Rays}

For every point $p$ in a point set $X$ in upper convex position
we define two \defword{rays} $\Lray{p}$ and $\Rray{p}$.
Let $\overline{xp}$~and~$\overline{py}$ be the segments of $\UH(X)$ that are incident to $p$,
where $x$ is to the left of~$p$ and $y$ is to the right.
These segments define the tangent lines $\ell_x$~and~$\ell_y$ through~$p$ by $x \in \ell_x$ and $y\in \ell_y$.
Define~$\Lray{p} \subset \ell_y$ to be the half-line left of~$p$, the \defword{left ray} of~$p$.
Symmetrically we define $\Rray{p} \subset \ell_x$ to be \defword{right ray} of~$p$.
See Figure~\ref{fig:rays} for an illustration.
If $p$ is the leftmost point in $X$ then $\Rray{p}$ is the vertical half-line above~$p$. 
Symmetrically, if $p$ is the rightmost point in $X$ then $\Lray{p}$ is the vertical half-line above~$p$.
\begin{figure}[htb]
  \begin{center}
    \input{rays}
    \caption{The two rays $\protect\Lray{p}$ and $\protect\Rray{p}$ of~$p$,
      and a line~$\ell$ like in Lemma~\ref{lem:rayAbove}.}
    \label{fig:rays}
  \end{center}
\end{figure}

The following lemma summarizes the geometric property that our construction is built upon:
\begin{lemma}\label{lem:rayAbove}
Let~$X$ be a set of points in upper convex position.
For any point $p\in X$ with rays $\Lray{p}$ and~$\Rray{p}$, let~$\ell$ be a 
line that intersects both rays~$\Lray{p}$ and~$\Rray{p}$.
Then the line~$\ell$ does not intersect~$\UCo(X)$.
\end{lemma}

\subsection{Point Certificates, Shadows, and Selected Points}
\label{sec:PointCertificate}

Now, we define the building block of our separation certificate, the
point certificate.
In the following we let $\SC{A}$ and $\SC{B}$ be the shortcut versions of 
$A$ and $B$ by two non-overlapping sets of effective and conservative shortcuts $H_A$ and~$H_B$.
Note that conservativeness ensures that the equality points of the original and the shortcut versions are identical
and that the set of interior points is also the same
i.e., $\UH(A)\cap\UH(B)=\UH(\SC{A})\cap\UH(\SC{B})$, 
$A\cap\UCo(B)=\SC{A}\cap\UCo(\SC{B})$ and 
$B\cap\UCo(A)=\SC{B}\cap\UCo(\SC{A})$.

We describe the situation of point certificates for streaks of
polarity $B$ over~${A}$.  Every $p\in {A}$ inside
$\UCo(\SC{B})$ defines a \defword{point certificate}.  Both the rays $\Lray{p}$
and $\Rray{p}$ intersect~$\UH(\SC{B})$.  We denote these points the
\defword{ray intersections} $\LrayI{p}$ and $\RrayI{p}$, respectively.
By Lemma~\ref{lem:rayAbove}, we are sure that there is no equality point
of~$\UH({A})$ and~$\UH({B})$ (in the vertical slab) between
$\LrayI{p}$~and~$\RrayI{p}$, since the segment $\LrayI{p}\,\RrayI{p}$
would be a geometrically conservative shortcut.  We call~$p,\LrayI{p}\,\RrayI{p}$ a
\defword{point certificate} that extends (horizontally)
from~$\LrayI{p}$ to $\RrayI{p}$.  See Figure~\ref{fig:shadow} for an
illustration.

\begin{figure}[htb]
  \begin{center}
    \input{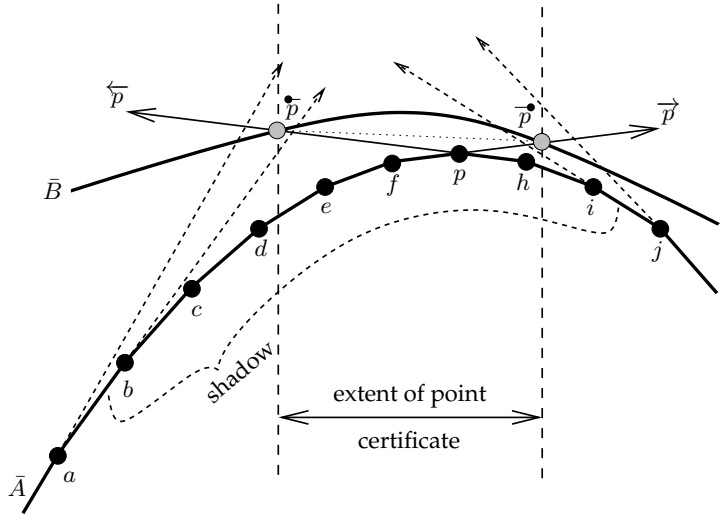}
    \caption{The shadow around the point certificate of~$p$ and implied by the
      rays $\protect\Lray{p}$ and~$\protect\Rray{p}$.
      The points~$b$--$i$ are in the shadow of~$p$, the points~$a$
      and~$j$ are not. The hull $\UH(\SC{B})$ is depicted
      as a convex curve, since it is here not important that it is a polygon.}
    \label{fig:shadow}
  \end{center}
\end{figure}

The efficiency of the complete certificate relies upon carefully
choosing the points for which point certificates are instantiated and
the ray intersections are determined.  This choice of point
certificates for all streaks of polarity $B$ over~$A$ is
reflected by the set~$Q_A\subseteq A \cap\UCo(\SC{B})$ of
\defword{selected points}.  Symmetrically, for all streaks of
polarity~$A$ over~$B$ there is a set of selected
points~$Q_B\subseteq B\cap\UCo(\SC{A})$.

For the set~$Q_A$ of selected points we require that the extent of all
point certificates are horizontally disjoint.  We say that all point
certificates enjoy the \defword{disjointness condition}.  We
furthermore require that $Q_A$ is \defword{maximal}, i.e., no further
point can be selected without violating the disjointness condition.
For a selected point~$p\in{A}$ the disjointness condition disallows
the selection of several other points of~${A}\cap\UCo(\SC{B})$, a
range in the left-to-right ordering of~${A}$ around~$p$, the
\defword{shadow} of~$p$, as illustrated in Figure~\ref{fig:shadow}.
More precisely, a point~$q\in{A}$ is in the right (left) shadow of
$p$, if $q\in\UCo(\SC{B})$ is to the right (left) of $p$ and the extent
of the point certificates of $p$ and $q$ overlap.

\subsection{Tangent Certificates}
\label{sec:TangentCertificates}

Given that we require that the extent of point certificates for 
selected points do not overlap, we have
another type of certificate, so called \defword{tangent certificates}, to
ensure that there is no further equality point outside the chosen
point certificates.
A tangent certificate is defined by a tangent line~$\ell$ containing a segment of the locally
inner hull and an interval 
defining the extent, in
which the line~$\ell$ is inside the other hull 
(see Figure~\ref{fig:dangling}).
%
These come in two main versions, the ordinary one for the gap
between two point certificates, and a boundary certificate between
a point certificate and an equality point.
Some special cases are considered to be tangent certificates as well, see Section~\ref{sec:trivial-streak}.

\subsubsection{Tangent Certificate between Point Certificates}
\label{sec:NormalTangentCertificate}

To cover the gap between the point certificates between 
two selected points, we have an \defword{ordinary tangent certificate}.
Consider the gap between two neighboring point
certificates for two selected points~$p,q\in Q_A$,
where~$p$ is to the left of~$q$, such that no point in the slab
between~$p$ and~$q$ is selected (is in~$Q_A$).
%
%
The disjointness condition states that~$\RrayI{p}$ is left of~$\LrayI{q}$.
The situation is illustrated in Figure~\ref{fig:dangling}.
%

\begin{figure}[htb]
  \begin{center}
    \input{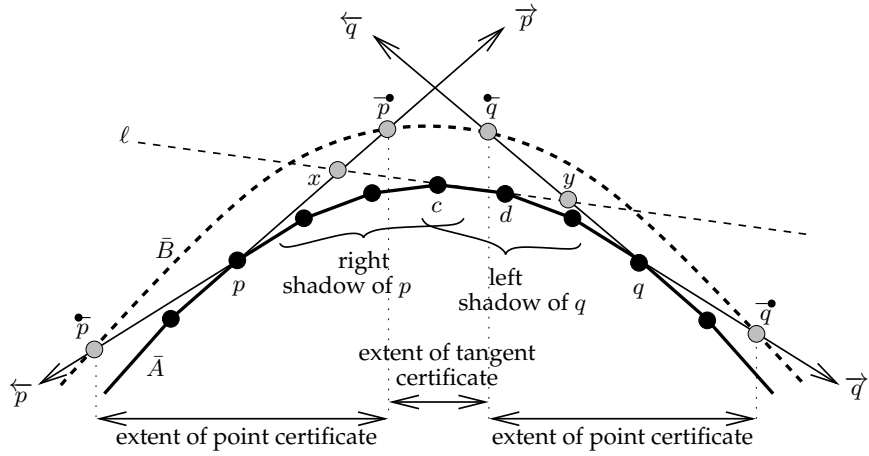}
    \caption{A tangent certificate.  Hull~$A$ is below hull~$B$.  The
      points~$p$ and~$q$ are selected.
      The gap between the point certificates is
      the slab between~$\protect\RrayI{p}$ and~$\protect\LrayI{q}$. 
      The tangent line~$\ell$ through $c$ and $d$ is a tangent certificate, 
	  since $c$ and $d$ are in the shadows of $p$ and~$q$, respectively.
      The tangent certificate extends (horizontally) from~$x$ to~$y$,
      covering the gap between the point certificates.}
    \label{fig:dangling}
  \end{center}
\end{figure}

A tangent certificate closes the gap of the separation certificate between $\RrayI{p}$ and $\LrayI{q}$. 
It consists of a tangent line~$\ell$ containing a segment of~$\UH(\SC{A})$ and its
intersections $x$ and $y$ with the rays~$\Rray{p}$ and~$\Lray{q}$, respectively.
The tangent certificate \emph{extends} from $\RrayI{p}$ to~$\LrayI{q}$.
If~$x$ and~$y$ are inside~$\UC(\SC{B})$ 
(the order of $x$ and $\RrayI{p}$ on~$\Rray{p}$
and respectively of $y$ and~$\LrayI{q}$ on~$\Lray{q}$),
then by the convexity of~$B$, there is no equality point in the slab between
$y$~and~$x$ and we have a valid tangent certificate.

Every point of~$A$ between $p$ and~$q$ is in at least one of the
shadows of $p$ and~$q$, 
since by the maximality of $Q_A$ no additional point between $p$ and $q$ can be selected. 
Therefore there exists two neighboring points $c$ and $d$ in $A$
such that $c$ is in the right shadow of $p$ and $d$ is the left shadow of $q$
(possibly $c=p$ and/or $d=q$).
By the definition of shadows, the tangent~$\ell$ through $c$ and $d$ is a valid tangent certificate.

\subsubsection{Boundary Certificate} \label{sec:boundaryCert}

\relax
Between the point certificate of a selected point and an equality point, we have a \defword{boundary certificate}.
Let~$e$ be an equality point and~$p\in A$ a selected point, such
that there is no selected point in the slab between $e$~and~$p$, where $p$ is to left of~$e$ ($p$ to the right of $e$ is symmetric).
The situation is illustrated in Figure~\ref{fig:half-open}.

\begin{figure}[tb]
  \begin{center}
    \input{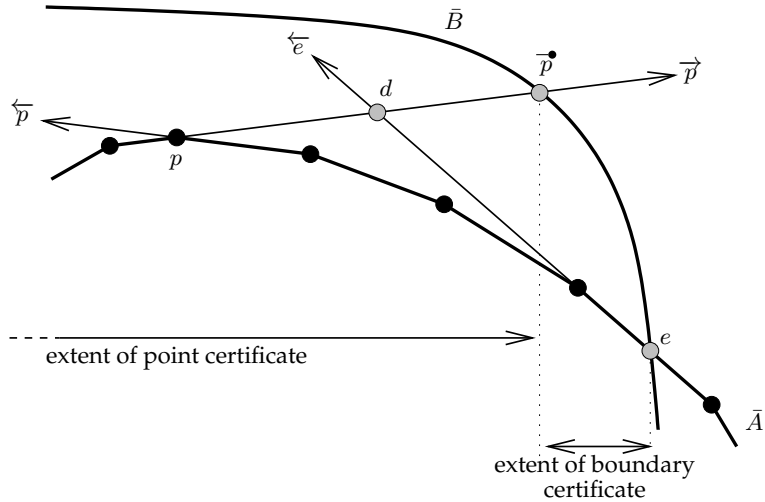}
    \caption[A boundary certificate]
    {A boundary certificate: $p$ is the rightmost selected point in the
      streak and $e$ is the equality point.
      The intersection~$d$ of 
	the rays $\Rray{p}$ and $\Lray{e}$
      shows that the boundary certificate is valid.}
    \label{fig:half-open}
  \end{center}
\end{figure}

Let~$\Lray{e}$ be the left ray of $e$ for $A\cup\{e\}$, 
i.e., the points to the left of $e$ on the line containing the
segment of $\UH(\SC{A})$ containing $e$.
Since we cannot select another point between $p$~and~$e$,
all points between $p$ and $e$ are in the shadow of~$p$,
i.e., $\Lray{e}$ intersects~$\Rray{p}$ inside~$\UH(\SC{B})$.
Therefore there cannot be an equality point between~$e$ and~$p$,
and the boundary certificate is valid.
The extent of the boundary certificate is between $e$ and~$\RrayI{p}$

\subsubsection{Special Cases}
\label{sec:trivial-streak}

\begin{figure}[tb]
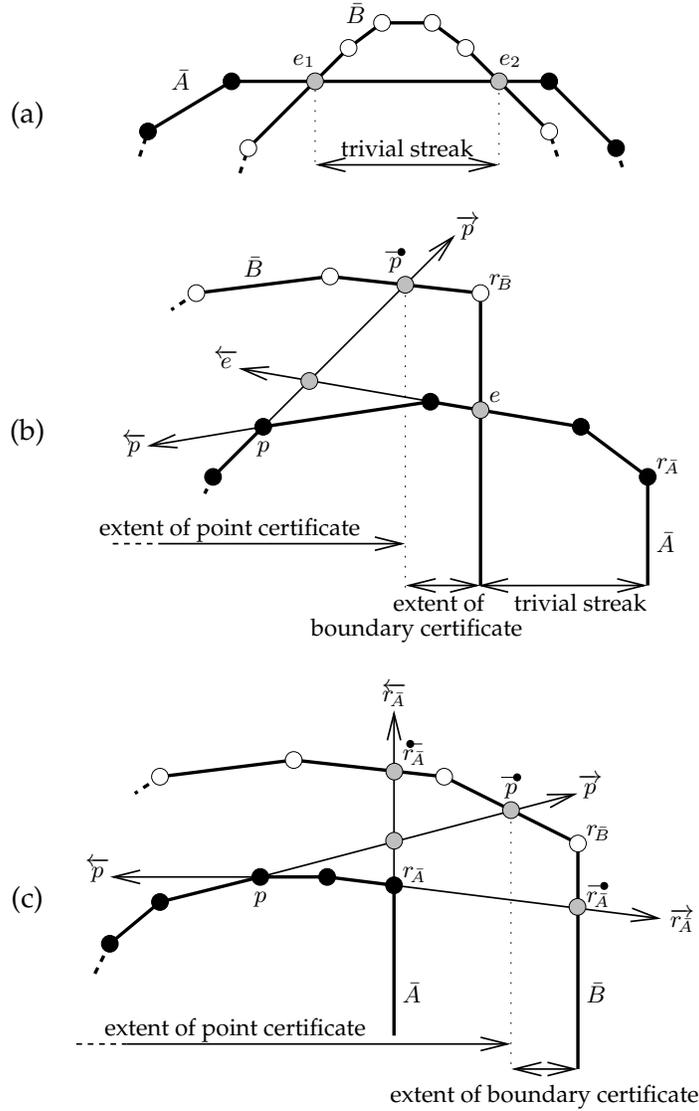

  \begin{center}
    \begin{tabular}{c@{\quad}c}
    \raisebox{4ex}{(a)} & \input{trivial-streak} \\[2ex]
    \raisebox{15ex}{(b)} & \input{special-boundary-1} \\[2ex]
    \raisebox{15ex}{(c)} & \input{special-boundary-2} 
    \end{tabular}
    \caption[Special boundary certificates]
    {Special cases of boundary certificates.}
    \label{fig:special-boundary-cert}
  \end{center}
\end{figure}

We say that a streak between two equality points~$e_1$ and $e_2$ is \defword{trivial}
if the locally inner hull $\UH(\SC{A})$
does not have a point of~$\SC{A}$ inside the streak, i.e., the locally
inner hull is a part of a segment of~$\UH(\SC{A})$, and the
endpoints of this segment are outside $\UC(\SC{B})$.
The convexity of $\UH(\SC{B})$ implies that the slab between $e_1$ and $e_2$
cannot contain any nother
equality point.
The situation is illustrated in Figure~\ref{fig:special-boundary-cert}(a).

Recall that on $\UH(\SC{A})$ and $\UH(\SC{B})$ we have vertical infinite segments below
the leftmost and rightmost points of $\SC{A}$ and $\SC{B}$. 
In the following we consider the rightmost points. The leftmost points
are handled symmetrically.
Let $\rA$ and $\rB$ be the rightmost points of $\SC{A}$ and $\SC{B}$, respectively, and
assume $\SC{A}$ is below $\SC{B}$ at the vertical line through the leftmost point of $\rA$ and $\rB$.
There is obviously no equality point to the right of this line.
Let
$p$ be the rightmost selected point of~$\SC{A}$, possibly $p=\rA$.
If $\rA$ is to the right of $\rB$, we have an ordinary boundary certificate
by the ray $\Lray{e}$ where $e$ is the rightmost equality point of $\UH(\SC{A})$ and $\UH(\SC{B})$.
The situation is illustrated in Figure~\ref{fig:special-boundary-cert}(b).
%
%
%
Otherwise $\rB$ is to the right of~$\rA$.
If $p=\rA$ we are done. Otherwise $p$ is to the left of $\rA$, and $\Rray{p}$ must intersect 
$\Lray{\rA}$ inside $\UH(\SC{B})$, 
since otherwise $\rA$ would be selectable.
It follows that the point certificate
of $p$ spans at least the slab $p$ to $\rA$, 
and ensures that there is no equality point in this slab.
The situation is illustrated in Figure~\ref{fig:special-boundary-cert}(c).

To unify the algorithm description, we consider the two vertical half-lines
below $\rA$ and $\rB$ to meet at an equality point at minus infinity.
With this view the slab between $\rA$ and $\rB$ in the first case,
Figure~\ref{fig:special-boundary-cert}(b), becomes a trivial streak,
and in the second case, Figure~\ref{fig:special-boundary-cert}(c), the
vertical line through $\rA$ becomes a boundary certificate.

\subsection{Complete Certificate}
\label{sec:CompleteCertificate}

There is a complete certificate of separation between the two upper convex point
sets $A$ and $B$,
if every vertical line not through an identified equality point is in the extend of some certificate.
A \defword{complete certificate} is geometrically described
by~$(Q,H)$:
\begin{description}
\item[$E$] The set of equality points.
\item[$H=(H_A,H_B)$] Sets of shortcuts on the locally outer
  hulls.
  The shortcuts are required to be effective, non-overlapping,
  conservative.
\item[$Q=(Q_A,Q_B)$] Maximal sets of selected points on the locally
  inner hulls.
  The selected points, together with the shortcuts,
  define ray intersections and point certificates.
  Point certificates are required to be non-overlapping.
  By the maximality of $Q$, 
  there exists a tangent (normal or boundary) certificate
  between any two neighboring point certificates or equality points.
\end{description}

By the maximality of $Q$,
every non-trivial streak between two equality points, i.e., every streak with at least one point
on the locally inner hull, contains at least one selected point.

%
The role of the shortcuts is to ``trivialize'' the locally outer hull
between ray intersections and equality points 
to contain only a constant number of segments. This is
made precise by the following requirement for our certificate.

\begin{description}
\item[\rm\defword{Aggressive shortcuts}]
For any four consecutive segments of $\SC{B}$  
within a streak of polarity $\SC{B}$ over $\SC{A}$ 
(symmetric for polarity $\SC{A}$ over $\SC{B}$),
at least one of the four segments contains
an equality point or an intersection with a ray
from a selected point in $Q_A$.
\end{description}

%
%

%

\subsection{Properties of a Certificate}

The representation of a complete certificate as a data structure is
considered in Section~\ref{sec:Representation}.
Here, we discuss some geometric properties of a complete certificate.

\begin{lemma}\label{lem:shortcutAllowed}
  Consider a (valid) point certificate and let $u$ and $v$ be the corresponding ray intersections or equality points on the outer hull that define the extent of the certificate, namely for a point certificate of a selected point~$p$ we have $(u,v)=(\LrayI{p} ,\RrayI{p})$, for a tangent certificate between two selected points~$p$ and~$q$ we have $(u,v)=(\RrayI{p} ,\LrayI{q})$, and for a boundary certificate between a selected point~$p$ and equality point~$e$ we have $(u,v)=(e,\LrayI{p})$ or $(u,v)=(\RrayI{p},e)$.
Then there always is a point certificate that ensures that the line connecting $u$~and~$v$ is outside the locally inner hull, and hence any shortcut with cutting points between $u$~and~$v$ will not intersect the locally inner hull.
\end{lemma}
\begin{proof}
 By  Lemma~\ref{lem:rayAbove}.
\end{proof}
\begin{lemma}
\label{lem:LimitedRays}
  Let $(Q,H)$ be a complete certificate for the merging of 
  $A$ and $B$, and let $\overline{xy}$ a segment of $\UH(\SC{B})$.
  Then there are at most 4 ray intersections on~$\overline{xy}$.
\end{lemma}
\begin{proof}
  Assume there are 5 ray intersections.
  Then there are 2 point certificates with extent contained inside the slab between $x$ and~$y$.
  Let~$p$ left of~$q$ be the corresponding selected points.
  Now, both rays $\Lray{p}$ and~$\Rray{q}$ must be completely below the line through
  $p$ and~$q$.
  Because $p$ and~$q$ are below $\overline{xy}$, an intersection
  between the line defined by $x$ and~$y$ and the line defined by $p$
  and~$q$ must be left of~$p$ or right of~$q$.
  If the intersection is on the right, then on the left side~$\Lray{p}$
  cannot intersect~$\overline{xy}$ because it is completely below the
  line through $p$ and~$q$.
  Similarly, it is impossible for~$\Rray{q}$ to intersect~$\overline{xy}$ if
  the intersection is on the left.
  Hence, it is impossible that one segment has more than 4 ray
  intersections of a complete certificate.
\end{proof}

\begin{lemma}
\label{lem:LimitedImpact}
  Let $(Q,H)$ be a complete certificate for the merging of 
  $A$ and $B$, and let $\overline{xr}$ and $\overline{ry}$ be two consecutive segments
  of $\UH({B})$.
  Then at most 11 selected points have ray intersections on
  $\overline{xr}\cup\overline{ry}$ or on shortcuts intersecting
  $\overline{xr}$ and/or $\overline{ry}$.
\end{lemma}
\begin{proof}
  Because shortcuts are required to be non-overlapping, there can be
  at most 3 shortcuts intersecting $\overline{xr}\cup\overline{ry}$.
  These 5 consecutive segments of $\UH(\SC{B})$
  can
  have by Lemma~\ref{lem:LimitedRays} at most 20 ray intersections,
  and since at most two of these are not part of a pair stemming
  from the same selected point, there can be at most 11 selected
  points with ray intersections on the described segments.
\end{proof}


\section{Geometric Monotonicity}
\label{sec:monotonicity}

In this section we consider the geometric consequences of applying
\pc{Replace} operations to the two point sets $A$ and $B$.  Recall
that these operations are required to ensure that $\UC(A)$ and
$\UC(B)$ shrink monotonically (see Figure~\ref{fig:merger-example}).
A point~$p$ that eventually is inserted into $A$ by a \pc{Replace}$_A$
operation goes through the following states during its
\defword{life-cycle}, in the given order but possibly skipping some
states (see Figure~\ref{fig:lifecycle}).

\begin{enumerate}
  \item\label{lc:a} $p$ is not inserted into~$A$ yet ($p\in\UCo(A)$).
  \item\label{lc:b} $p$ is in $A$ and below $\UH(B)$ ($p\in A\cap\UC(B)$).
  \item\label{lc:c} $p$ is selected ($p\in Q_A \cap\UC(B)$).
  \item\label{lc:d} $p$ \defword{surfaces} but is hidden by a bridge ($p\in A\cap\UCo(A\cup B)\setminus\UC(B)$).
  \item\label{lc:f} $p$ is on the merged upper hull ($p\in\UV(A\cup B)$).
  \item\label{lc:g} $p$ is deleted ($p\notin\UC(A)$).
\end{enumerate}

Note that we include a point being selected into the above life-cycle,
since our construction will ensure that a selected point remains selected until the point surfaces.

\begin{figure}[thb]
  \begin{center}
    \input{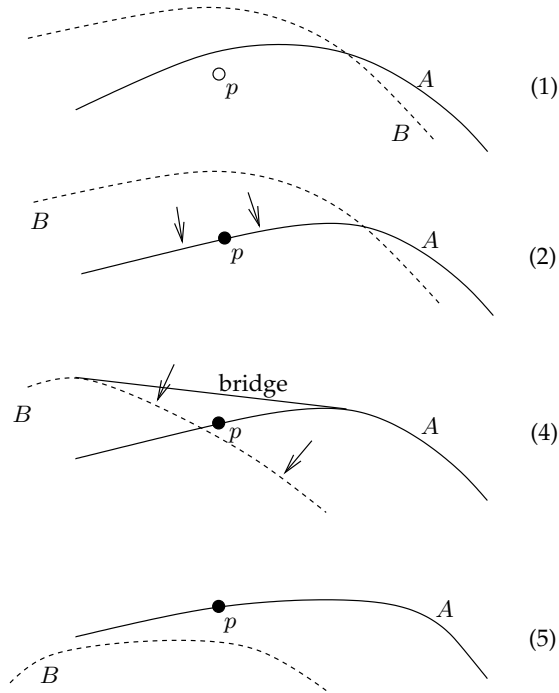}
    \caption[Life-cycle of a point]
    {The different stages of the life-cycle of a point in the merging.
    The hulls are depicted as curves, and the arrows indicate changes
    induced by deletions/replacements.}
    \label{fig:lifecycle}
  \end{center}
\end{figure}

Next, we consider the geometric influence of a
\pc{Replace}$_B(r,L)$ operation on a set of non-overlapping, effective
and conservative shortcuts $H=(H_A,H_B)$ and a maximal set of
selected points $Q=(Q_A,Q_B)$ with non-overlapping point
certificates (see Figure~\ref{fig:merger-example}, changes to $A$ are
handled symmetrically). Let $B'$ denote the new version of $B$ after
replacing $r$ by $L$, i.e., $B'=L\cup B\setminus\{r\}$.

The extent of each shortcut in $H_B$ can only shrink by replacing $r$
by~$L$, since $\UC(B')\subseteq\UC(B)$, i.e., $H_B$ remains a set of
non-overlapping shortcuts.  Since $A$ does not change, $H_B$ also
remains conservative shortcuts. All shortcuts in $H_B$ remain
effective, except possibly a single shortcut in $H_B$ that only
intersected the two segments adjacent to $r$ on $\UH(B)$ and that now
is above $\UH(B')$. Since $A$ does not change and
$\UC(B')\subseteq\UC(B)$, all shortcuts in $H_A$ remain
non-overlapping, conservative, and effective.


The following considerations show that shadows never get bigger.
Remember that the extent of the point certificate of a point is given by its ray intersections with the shortcut version of the other hull. 
If the point is on the locally outer hull, this is the singleton set of the x-position of the point.
Note that in the following lemma, it is irrelevant which points are selected.
\begin{lemma}[Monotonic Extent]
  \label{lem:MonotonicExtent}
  For the sets $A$ and~$B$ in convex position, assume there is a complete certificate with shortcuts $H=(H_A,H_B)$. 
  For an input point~$p \in A \cup B$, let the interval~$I$ denote the extent of the point certificate of~$p$.
  Let $I'$ be the extent of~$p$ after the operation \pc{Replace}$_B(r,L)$ for $r\neq p$.
  Then $I' \subseteq I$.
\end{lemma}
\begin{proof}
  As before, denote by~$B'$ the version of~$B$ after the replace operation.
  For~$p \in A$, both rays of~$p$ remain unchanged.
  The intersection with such a ray and~$\UH(B')$ can only move closer to~$p$ (if it still exists after the deletion). 
  The statement of the lemma follows.

  For~$p \in B \cap \UCo(A)$, if $r$ is the left neighbor of~$p$ on~$\UH(B)$, its right ray increases its slope, and hence the right ray intersection moves to the left.
  Symmetrically, if $r$ is the right neighbor of~$p$ on~$\UH(B)$, its left ray intersection moves to the right.
  In both cases $I'\subseteq I$.
  If~$p\not \in\UCo(A)$, $I'=I$ is a singleton.
  If~$p$ is not a neighbor of~$r$, the rays and the ray intersections do not change, leading to~$I'=I$.
\end{proof}

\begin{lemma}[Monotonic Shadow]
\label{lem:MonotonicShadows}
  Consider \pc{Replace}$_B(r,L)$ and shortcuts $H=(H_A,H_B)$ and selected points~$Q=(Q_A,Q_B)$.  

  If $p\in Q_A$ and $q\in A$ is in the shadow of $p$ after replacing $r$ by $L$ in $B$, 
  then $q$ was in the shadow of $p$ before the deletion.  

  If $p\in Q_B$ and $q \in B'$ is in the shadow of $p$ after replacing $r$ by $L$ in~$B$, 
  then $q$ is either a new point from $L$, or $q$ was also in the shadow of $p$ before the deletion.
\end{lemma}
\begin{proof}
  By the definition of shadows as non-empty intersection of extents and Lemma~\ref{lem:MonotonicExtent}.
\end{proof}

The situation of Lemma~\ref{lem:MonotonicShadows} is illustrated in Figure~\ref{fig:shadowlemma}.
Note that Lemma~\ref{lem:MonotonicShadows} also covers cases where $p$ and~$q$ are initially in different streaks, i.e., separated by a streak of the other polarity.

Our algorithm maintains geometric monotonicity, except for one special case.
To achieve that shortcuts are defined only by (possibly deleted) input points, it does change shortcuts in a way that the extent of a point certificate of a non-selected point can become bigger (cf. Section~\ref{sec:newSC}). 
Still, the algorithm guarantees the monotonicity of shadows of selected points, which is sufficient.

\begin{figure}[tb]
    \begin{center}
      \input{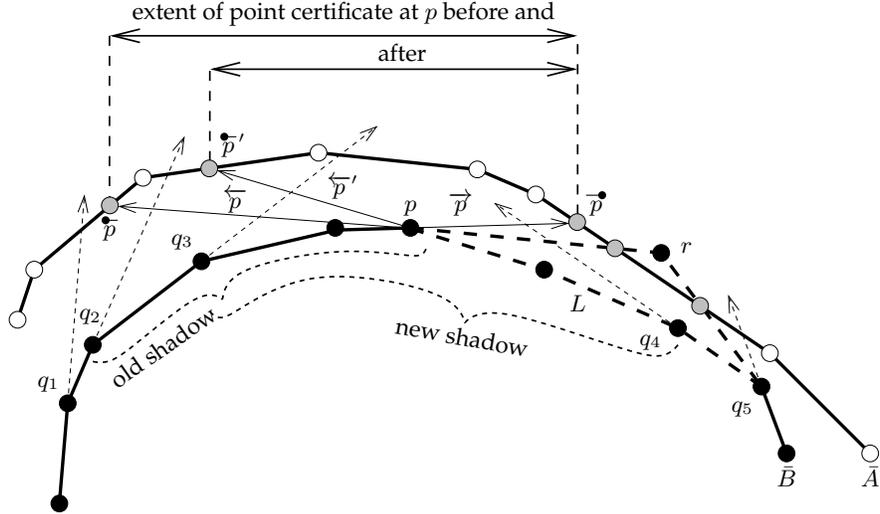}
      \caption[Ray changing]
      {The situation of a ray changing during a deletion.}
      \label{fig:shadowlemma}
    \end{center}
\end{figure}

\subsection{Reestablishing Tasks}\label{sec:geomReestablishing}
From the above discussion, we have the following reestablishing tasks
if we have a complete certificate $(Q,H)$ before a
\pc{Replace}$_B(r,L)$ operation. In the Sections~\ref{sec:algRepl} and~\ref{sec:mainReestablish} we describe
in detail how to handle each of the tasks below.

\begin{description}
\item[Shortcut effectiveness] At most one redundant shortcut needs
  possibly to be eliminated, namely a shortcut that intersected both
  segments adjacent to $r$ on $\UH(B)$.
\item[Valid selected points] The selected points $Q_A \setminus
  \UCo(\SC{B'})$ are now on the locally upper hull and cannot be
  selected any longer.  Note that these points must have had a ray
  intersection with the above eliminated shortcut on $\UH(\SC{B})$
  or a segment on $\UH(\SC{B})$ adjacent to $r$ or 
  adjacent to the eliminated shortcut, since
  otherwise their ray intersections would not change and we would
  still have point certificates proving that the points are on the
  inner hull.
\item[Identifying new equality points] In the extent of the deleted
  segments on $\UH(\SC{B})$, an arbitrary number of new equality points might appear.
\item[Updating bridges adjacent to $r$] Update bridges that had an
  endpoint in the deleted point~$r$.
\item[New bridges] Each new equality point defines a new bridge, or
  the possible change of an existing bridge.
\item[Maximality of selected points] Since shadows shrink and new
  points are introduced, we need to identify possible new points to
  select to assure the maximality of the set of selected points $Q_A$ and $Q_B$, i.e.,
  that all points on the inner hull is in the shadow of at least one
  selected point.
\item[Aggressiveness of shortcuts] The replacement of $r$ by $L$ in
  $B$ can introduce new segments in $\UH(\SC{B'})$, and segments of
  $\UH(\SC{A})$ may change from being on the inner hull to being on
  the outer hull. 
  We potentially need to introduce new shortcuts in both $H_A$ and $H_B$,
  to ensure the aggressiveness of the shortcuts after the update.
\end{description}

\section{Helper Data Structures}
\label{sec:helpers}

This section introduces the two data structures \emph{split-array} and
\emph{splitter} crucial for the efficient maintenance of our complete
certificate. 
Whereas splitters are part of the data structure between deletions, split arrays are only used temporarily during the reestablishing of the complete certificate.


\subsection{Split-Array}
\label{sec:split-array}

A \defword{split-array} stores a sorted sequence of elements.  Its
main purpose is to allow efficient searches as part of a split
operation.  The interface to a split-array is:
\begin{description}
\item[\pc{Build}($e_1,\ldots,e_k$)] Given elements
  $e_1<e_2<\cdots<e_k$, returns a new split-array~$S$ containing the
  elements $e_1,\ldots,e_k$.
\item[\pc{Split}($S,t$)] Given a split-array $S$ and an element $t$,
  splits $S$ into two split-arrays $S_1$~and~$S_2$ such that $S_1\leq
  t < S_2$, i.e., $x\in S_1 \Rightarrow x\leq t$, and $y\in S_2
  \Rightarrow t<y$.  The split-array~$S$ is destroyed.
\item[\pc{Min}($S$), \pc{Max}($S$)] Return the smallest and largest
  elements in the split-array~$S$, respectively.
\item[\pc{Delete-Min}($S$), \pc{Delete-Max}($S$)] Remove and return
  the smallest and largest elements in the split-array~$S$,
  respectively.
\end{description}

\begin{lemma}
\label{lem:split-array}
  There is a data structure implementing a split-array such that
  \pc{Build} takes amortized~$O(k)$ time, where~$k$ denotes the number
  of elements in the split-array, and \pc{split}, \pc{Min}, \pc{Max},
  \pc{Delete-Min} and \pc{Delete-Max} take amortized constant time.
\end{lemma}
\begin{proof}
  A \pc{build} operation stores the points in an array.  All
  subsequent split-arrays resulting from \pc{split} operations reuse
  appropriate sub-arrays by remembering the leftmost and rightmost
  positions into the array.  The search for the split point of a
  \pc{split} operation starts by comparing~$t$ with the element~$m$
  stored at the middle position of the array.  If~$t$ is larger
  than~$m$, we perform an exponential search in the array from the
  right, otherwise from the left.  An exponential search from the left
  compares~$t$ with the elements at positions $1, 2, 4, 8, \ldots$
  until it finds an element that is greater than~$t$.  From there it
  performs a standard binary search, i.e., it compares~$t$ with the
  middle element of the currently still possible outcomes.  Such an
  exponential search takes~$O(\log h)$ time, where~$h$ is the size of
  the smaller resulting split-array.  \pc{min} and \pc{max} simply
  return the leftmost and rightmost element in the sub-array,
  respectively, and \pc{delete-min} and \pc{delete-max} furthermore
  also increment and decrement the leftmost and rightmost positions,
  respectively.  These four operations clearly take worst-case $O(1)$ time.

  For the amortized analysis we define the potential of a split-array
  storing $n$ elements to be $\phi(n)=n - \log n$.  Then, clearly,
  \pc{build} takes amortized linear time and all other operations
  except \pc{Split} take amortized constant time.  A split operation releases
  some potential: Let $n_1+n_2=n$ be the respective sizes of the
  resulting split-arrays $S_1$ and $S_2$.  Assuming $n_1\leq n_2$
  (w.l.o.g.) we have $n_2 \geq n/2$.  The linear terms in the
  potential of $S$, $S_1$, and $S_2$ cancel and we get a release in
  potential of $\phi(n)-\phi(n_1)-\phi(n_2) = -\log n + \log n_1 +
  \log n_2 \geq \log n_1 + \log \frac{n}{2} - \log n = \log n_1 - \log
  2$.  Hence, the split operation takes amortized constant time.
\end{proof}

The data structure presented above is based on index calculations in
an array.  To completely avoid such calculations, one can use a
level-linked (2,4)-tree \cite{HuddlestonMehlhorn82,hoffmann86} and the search procedure of
Section~\ref{sec:ImplementationSplitter} without increasing the time
bounds.

\subsection{Splitter}
\label{sec:splitter}

A \defword{splitter} is a generalization of a split-array where a
split operation might have to be \emph{suspended} because a comparison
in the search for a splitting point can currently not be answered.  In
our application a search in a splitter is the search for a new
selectable point between two selected points $p$ and $q$, and a search
is suspended when a tangent certificate has been found instead.  
The protected intervals defined below are the shadows of $p$ and~$q$ (see Figure~\ref{fig:dangling}).

\subsubsection{Protected Intervals Search}
\label{sec:Game}

Consider a list of elements~$e_1,\ldots,e_n$.  Initially, all elements
are colored both \defword{red} and \defword{blue}.  As an invariant
over time, all red elements form a prefix of $e_1,\ldots,e_n$
(possibly empty), and all blue elements form a suffix of
$e_1,\ldots,e_n$ (possibly empty).  These two color intervals are
called \defword{protected intervals}.  The task of a search is to
\emph{inspect} the colors of a limited number of elements, to either 
\begin{enumerate}
\item identify an uncolored element where we can split the list into two new
lists, or 
\item to verify that all elements are still colored by at least
one color which suspends the search. 
\end{enumerate}
We need to continue the suspended search if the red and blue protected intervals had a chance to change. 
Here, it is important that they can only shrink, i.e.\ that some elements have lost their red and/or blue color (corresponding to the monotonicity of shadows,
see Lemma~\ref{lem:MonotonicShadows}).
Note that all elements have at least one color if and only if $e_1$ is blue or $e_n$ is red, or there is an~$i$ where $e_i$ is red and $e_{i+1}$ is blue.
Figure~\ref{fig:spl_game} illustrates one protected intervals search.  
In some way, the whole data structure and its algorithms are built around this suspended search.  
The actual algorithm using it is described in Section~\ref{sec:mainReestablish}.

\begin{figure}[htb]
  \begin{center}
    \input{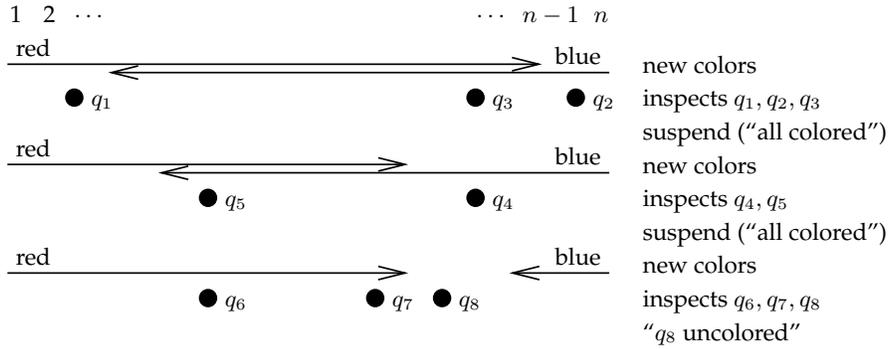}
    \caption
    {Protected intervals search.}
    \label{fig:spl_game}
  \end{center}
\end{figure}

The following observation is crucial for our search strategy:
\begin{lemma}
\label{lem:SplitterMove}
  Let~$e_i$ be colored red-only (blue-only) at some stage of the
  protected intervals search.  Assume at this stage or a later stage
  some element is uncolored.  Then, there is an uncolored $e_j$ where
  $j\geq i$ ($j\leq i$).
\end{lemma}
\begin{proof}
  Since the protected intervals can only shrink, $e_i$ cannot be
  colored blue again once it is red-only.  If~$e_i$ is not colored red
  anymore, $i=j$ proves the lemma.  Otherwise, let~$e_j$ be some
  uncolored point.  Since~$e_i$ is colored red, $e_j$ must be to the
  right of~$e_i$.  The statement in parenthesis follows symmetrically.
\end{proof}

As an immediate consequence of Lemma~\ref{lem:SplitterMove}, a search
can focus on an interval $[l,r]$ between the rightmost inspected
previously red-only element~$e_l$ and the leftmost inspected
previously blue-only element~$e_r$, and only inspect elements in this
interval.  Initially $[l,r]=[0,n+1]$.  If $r=l+1$, we inspect if $e_l$
and $e_r$ are still colored (we only inspect $e_l$ and $e_r$ if $1\leq
l$ and $r\leq n$, respectively).  If the inspected elements are still
colored, then all elements are still colored.  Otherwise, we have
found an uncolored element.  When $l+1<r$ we inspect an element $e_i$,
where $l<i<r$.  If $e_i$ is not colored the search terminates, and if
$e_i$ has only one color we can shrink the interval using
Lemma~\ref{lem:SplitterMove}, by setting $l=i$ ($r=i$) if $e_i$ is
red-only (blue-only).  Otherwise, $e_i$ has both colors, and all
elements are still colored and the search is suspended.  The next
reinvocation of the search starts inspecting~$e_i$ one more time.  By
choosing~$e_i$ according to an exponential search (similar to the
search in the split-array in Section~\ref{sec:split-array}), the
search uses a limited number of inspections:

\begin{lemma}
  For a protected intervals search there is a strategy that leads
  to~$O(m+\log \min\{i,n-i+1\})$ inspections, where $m$~is the number
  of suspends and $e_i$~is the uncolored element eventually
  identified.
\end{lemma}

\subsubsection{Splitter Interface}
\label{sec:splitterIF}

Beyond supporting protected interval searches, splitters have a richer
interface tailored to their use in mergers.  A splitter stores a
sorted sequence of elements $e_1,\ldots,e_n$ and can be in two states:
\defword{open} and \defword{closed}.  Open splitters support the
insertion and deletion of rightmost and leftmost elements.  An open
splitter $M$ can be split into two new open splitters $M_1$ and $M_2$
by a protected intervals search: First we make $M$ closed, and
initialize all elements in $M$ red and blue.  Repeatedly now we query
the splitter, where the splitter for each query verifies that all
elements are still colored or terminates the search by having
identified an uncolored element~$e_i$ where $M$ can be split.  
The current coloring of the elements is provided as a
function~$\mathcal{C}$, that can determine the color of an element in
$O(1)$ time.  The splitter $M$ remains closed until an uncolored
element has been identified.  Finally, we have an operation to join
two open splitters $M_1$ and~$M_2$ with some fresh points between them
to form a closed splitter $M$.  Initially, $M_1$ is colored red-only
and~$M_2$ is colored blue-only, and the fresh points are colored both
red and blue.

The precise interface to a splitter is:
\begin{description}
\item[\pc{Build}($e_1,\ldots,e_k$)] Given elements $e_1<\cdots<e_k$,
  returns a new open splitter~$M$ containing $e_1,\ldots,e_k$.
\item[\pc{Min}($M$), \pc{Max}($M$)] Return the smallest or largest
  element in the splitter~$M$, respectively.
\item[\pc{Extend}($M,e$)] Given an open splitter $M$ and an element
  $e$, where $e<$~\pc{Min}$(M)$ or \pc{Max}$(M)<e$, inserts $e$ as the first element of $M$.
  if $e<M$ or as the last element of $M$ if $e>M$.
\item[\pc{Delete-min}($M$), \pc{Delete-max}($M$)] Given an open
  splitter~$M$, removes the smallest and largest elements from $M$,
  respectively.
\item[\pc{Close}($M$)] Given a nonempty open splitter $M$, makes $M$ closed.  
  With respect to monotonicity of colors all elements in $M$ both red and blue, i.e., there is no restriction.
\item[\pc{Split}($M,\mathcal{C}$)] Given a closed splitter
  $M=e_1,\ldots,e_k$ and a new coloring $\mathcal{C}$ of the elements
  of $M$ satisfying the monotonicity property, inspects colors of
  elements in $M$ until one of the following can be returned: If all
  elements of $M$ are colored, $\bot$ is returned. Otherwise, an
  uncolored element $e_i$ exists, and $(M_1,e_i,M_2)$ is returned,
  where $M_1=e_1,\ldots,e_{i}$ and $M_2=e_{i},\ldots,e_k$ are new
  open splitters.  In the latter case $M$ is destroyed.
\item[\pc{Join}($M_1,(e_1,\ldots,e_k),M_2$)] Given two open
  splitters~$M_1$ and~$M_2$ and a list of elements $e_1,\ldots,e_k$,
  where \pc{Max}$(M_1)<e_1<\cdots<e_k<$~\pc{Min}$(M_2)$, constructs a closed splitter storing
  $M_1,e_1,\ldots,e_k,M_2$, where all elements in $M_1$ are colored
  red-only, $e_1,\ldots,e_k$ are colored both red and blue, and all
  elements in $M_2$ are colored blue-only.  Both $M_1$~and~$M_2$ are
  destroyed.
\end{description}

\subsubsection{Implementation of the Splitter}
\label{sec:ImplementationSplitter}

Despite the somewhat unusual interface, it is sufficient to implement
a splitter by a level-linked (2,4)-tree supporting exponential
searches~\cite{hoffmann86}.  A~\defword{finger} into a level-linked
(2,4)-tree is a pointer to the leaf containing the element~$x$.
\begin{theorem}[\cite{hoffmann86}]
\label{thm:hoffmann86}
  Level-linked (2,4)-trees support the following operations for a
  finger~$x$:
  \begin{itemize}
  \item 
    Finger-searches for an element~$y$ in time $O(\log k)$, where~$k$
    is the minimum of the number of elements between $x$ and~$y$, the
    number of elements smaller than~$y$, and the number of elements
    greater than~$y$.
  \item 
    Insertion of a new neighbor of~$x$ or the deletion of~$x$ in
    amortized $O(1)$ time.
  \item 
    The joining of two trees or the splitting of the tree at~$x$ in
    amortized $O(\log k)$ time, where~$k$ is the size of the smallest
    of the two involved trees.
  \end{itemize}
\end{theorem}

\begin{theorem}
\label{theorem:splitter}
  If the colors of an element can be inspected in constant time, then
  there exists an implementation of a splitter supporting
  \begin{itemize}
  \item \pc{build} and \pc{join} in amortized~$O(k)$ time, where~$k$
    is the number of new elements $e_1,\ldots,e_k$, and
  \item \pc{split}, \pc{extend}, \pc{delete-min} and \pc{delete-max}
    in amortized~$O(1)$ time.
  \end{itemize}
\end{theorem}

\begin{proof}
  We first describe a solution without the \pc{Join} operation.  The
  elements of a splitter are stored at the leaves of a level-linked
  (2,4)-tree \cite{hoffmann86}.  Internal nodes store search keys that
  are double-linked pointers to the corresponding leafs.  We keep
  pointers to the root and the leftmost and rightmost leaves of the
  tree.  For a closed splitter, we have a pointer to an element~$q$ in
  the tree, encoding an ongoing protected interval search.

  The operations \pc{build}, \pc{extend}, \pc{delete-min} and
  \pc{delete-max} are performed as described in~\cite{hoffmann86}, and
  a \pc{close} operation changes the status of the splitter to closed
  and lets $q$ be an element at the root.  A call to \pc{split}
  continues the suspended search for an uncolored element as encoded
  with~$q$, essentially as described in Section~\ref{sec:Game} using
  Lemma~\ref{lem:SplitterMove}.  The precise order for inspecting the
  colors of the elements in the (2,4)-tree is to first inspect one of
  the elements stored at the root, and then to perform a finger
  search, either from the leftmost or from the rightmost leaf,
  depending of the outcome of the first comparison.  If we start at
  the leftmost leaf, we inspect bottom-up the keys on the path to the
  root, effectively moving to the right, until we inspect a key and
  find that we have to move left.  From there we follow the standard top-down
  search procedure of the (2,4)-tree.  If the current $q$ is not
  colored, we split the search tree of $M$ at $q$ (using the algorithm
  from \cite{hoffmann86}) into $M_1$, $q$, and $M_2$ and mark $M_1$
  and $M_2$ as open splitters, and return $(M_1,q,M_2)$.  If $q$ has
  both colors, we return~$\bot$, since all elements are still colored
  and the search needs to be suspended.

  If $q$ has precisely one color, we continue as described in
  Section~\ref{sec:Game} until two neighboring elements have been
  inspected. These two elements are now repeatedly inspected until one
  of them can be returned as the splitting point.  After each
  inspection identifying that both elements are still colored (i.e.,
  all elements are still colored), we return $\bot$.
  
  By Theorem~\ref{thm:hoffmann86}, \pc{build} takes amortized $O(k)$
  time and \pc{extend}, \pc{delete-min} and \pc{delete-max} take
  $O(1)$ time.  A call to \pc{close} followed by $t$ calls to
  \pc{split} takes $O(t+\log \min\{n_1,n_2\})$, where $n_1$~and~$n_2$
  denote the size of the two resulting splitters (possibly after
  further calls to split).  To argue about the amortized $O(1)$ time
  of the \pc{close} and \pc{split} operations, we assign an additional
  $n-\log n$ potential to an open splitter of size~$n$.  This only
  increases the amortized cost of \pc{extend} by $O(1)$.  Similar to
  the split-array (Lemma~\ref{lem:split-array}), eventually splitting
  such an open splitter into two open splitters of size
  respectively~$n_1$ and~$n_2$ releases $\log \min\{n_1,n_2\}-1$
  potential, achieving the claimed amortized time bounds.

  To support \pc{Join}, we construct a special type of closed splitter
  consisting of the triple $(T_1,E,T_2)$, where $T_1$ and $T_2$ are
  the (2,4)-trees for $M_1$ and $M_2$, and $E=(e_1,\ldots,e_k)$. By
  Lemma~\ref{lem:SplitterMove} and the initial coloring of $M_1$ and
  $M_2$ it is guaranteed that eventually there will be a splitting
  point within the range $\max(M_1),E,\min(M_2)$. A \pc{split}
  operation on such a closed splitter is performed as follows: While
  the leftmost element $e$ of $E$ is colored red-only, we move $e$
  from $E$ to $T_1$.  We return $\bot$, if $e$ and $\max(T_1)$ are
  colored, or $E=\emptyset$ and both $\max(T_1)$ and $\min(T_2)$ are
  colored. Otherwise, we move the remaining elements from $E$ to
  $T_2$. Now either $\max(T_1)$ or $\min(T_2)$ is uncolored. Let $e$
  be one of these uncolored elements and remove $e$ from the relevant
  $T_i$. Let $T_1$ and $T_2$ be new open splitters $M_1$ and $M_2$,
  and return $(M_1,e,M_2)$. The \pc{join} operation with $k$ new
  elements followed by $t$ \pc{split} operations in total perform at
  most $k$ insertions into the (2,4)-trees, taking amortized $O(k)$
  time, and the increase in potential from the open splitter~$M$ to
  the open splitters~$M_1$ and $M_2$ is at most $O(k)$. The total time
  for these operations becomes amortized $O(k+t)$.
\end{proof}

\section{Invariants and Representation of the Certificate}\label{sec:Representation}


In this section we describe our representation of the two upper hulls $A$ and $B$ of a merger and the representation of a complete geometric separation certificate for $A$ and $B$ as described in Section~\ref{sec:SeparationCertificate}.
More precisely, we will discuss the interface of the data structure that encapsulates this representation that is the backbone of our algorithms in two roles:
\begin{itemize}
\item between two replacement operations to store the complete certificate of the static situation as described in Section~\ref{sec:SeparationCertificate}, and
\item enabling the reestablishing algorithm ``Replace'' of Section~\ref{sec:algRepl}, that rests upon the geometric considerations of Section~\ref{sec:geomReestablishing}.
\end{itemize}

These concepts rely upon the following types of points that are the basis of the geometric certificate:
\begin{itemize}
\item input points, possibly selected, possibly deleted, possibly above a shortcut (they must know if this is the case, but there is no need to link to the shortcut)
\item cutting points of the shortcuts; know their shortcut
\item equality points; know their partner and the corresponding bridge
\item ray intersections; know their selected point
\end{itemize}

Note that for a point that is not an input point, its position is the intersection of two lines defined by at most four input points, also allowing vertical lines through inputs points.
This is important with respect to avoiding numeric precision issues, because it allows us to work solely with constant degree polynomials of input values.

\subsection{Tangent certificate and splitter}
\label{sec:certSplitter}

Geometrically, the complete certificate requires a maximal set of selected points, meaning that no further point can be selected, as certified by all not-selected locally inside points being in some shadow of a selected point.
We place the points between neighboring selected points into one splitter and identify the shadows of the selected points with the protected intervals (that is colors).
In general, the left and right selected points $p$ and $q$ of a splitter define the red and blue color by their shadows (cf. Section~\ref{sec:PointCertificate}):
a point in the right shadow of $p$ is red, and each point in the left shadow of $q$ is blue.
The colors of a point $c$ can be determined in constant time.
It is red (blue) if and only if the line through $c$ and its right (left) neighbor~$d$ intersects the line segment between $p$ and $\RrayI{p}$ (between $q$ and $\LrayI{q}$). 
See Figure~\ref{fig:dangling}, and remember that the selected points $p$ and~$q$ are part of the splitter.
By the definition of shadows and the disjointness of point certificates, $p$ is not blue and $q$ is not red. 
In a complete certificate (i.e., once the reestablishing algorithm is finished), by the maximality of the set of selected points, each point in a splitter has at least one color.
This means that no split of the splitter is allowed and means that there exists a tangent certificate (as explained in Section~\ref{sec:NormalTangentCertificate}).

Summarizing, the non-trivial colors are defined by a single ray intersection.  
We can assume that the splitter stores this ray intersection point and is hence in the position to keep any of the two colors unchanged even if the defining ray intersection is actually already deleted.\vote{lost?}

We also place the locally inside points between a selected point and a neighboring equality point inside an open splitter.
If there is a valid boundary certicate (see Section~\ref{sec:boundaryCert}), we know that one of the protected intervals covers the whole splitter. 
This is a geometric condition between the point certificate of the selected point and the tangent at the equality, which fits to the splitter being open, i.e. not having an ongoing search. 
This will become important in the reestablishing algorithm because only open splitters can be shrunk, extended or joined.

\subsection{Representation details}
\label{sec:representation-details}

The interface to the outside (Section~\ref{sec:merger-interface}) is to receive points and  provide locators.
In our setting, an input point is a pointer to a base record to avoid copying coordinates.
A locator is a pointer to the merger internal record representing the input point.

The important geometric promise with respect to the low level representation of the certificate is that there are at most constantly many auxiliary points on each segment of the original and shortcut hull(s) (as argued for in Section~\ref{sec:navigation}). 

We have the following 5 types of records:
\begin{description}
\item[input point] a local record representing a point, storing an immutable pointer to the base record storing the $x,y$ coordinates of the point; references to these records are the locators.
If the point is selected there are two pointers to (and from) the input segment or shortcut containing the ray intersection.
\item[input segment] a segment of $\UH(A)$ or $\UH(B)$ storing pointers to other records, representing the constantly many auxiliary (ray intersection, cutting and equality) points on the segment (by pointers as detailed below). 
The pointers to the defining input points are immutable.
The sets $A$ and $B$ are two disjoint doubly linked lists alternating between an input point and an input segment of the respective participating upper hulls.
\item[shortcut] a record pointing immutably to the to base records of the points geometrically defining the line of the shortcut; pointers to (and from) the two input segments containing the cutting points. If the cutting point is an input point defining the shortcut, the first segment cut by this shortcut is considered to contain the cutting point.
Space for constantly many ray intersections.
\item[bridge and equality point] a record for an equality point and its corresponding bridge.
Doubly linked to (immutably) the two input segments defining the equality point and (mutably) the two input points defining the bridge.
Pointers to and from the splitters that end at the equality point.
\item[splitter] a wrapper around the splitter data structure of the previous chapter; the splitter stores (pointers to) input points;
the bounding selected point or the bounding equality point are doubly linked to this record.
This allows to find from one selected point the next selected point in constant time.
(The selected points 'know' the splitters they are in, not selected inside points are in a splitter, but they do not 'know' in which splitter they are).
The special case of a trivial streak, where the locally inside hull consists of a part of a single segment of one of the hulls, has an empty open splitter. 
\end{description}

Observe that this specific representation is one of many similarly reasonable choices. 
The interface used in the pseudocode (mainly for the reestablishing Algorithm presented in Section~\ref{sec:mainReestablish}) is the following:\vote{needed?}
\begin{description}
\item[\pc{Delete}($r$)] Given a locater~$r$ of an input point on hull~$A$ (wlog), remove~$r$ from the representation of~$A$ (and $A \cup B$) and return the two neighboring points~$X$ and~$Y$ on $\UH(A\cup B)$ (there are no auxiliary points on bridges), the two neighboring points~$x,y\in A$ and a constant length list of auxiliary points (cutting,  equality) (actually their partners/what they know) on the deleted edges $\overline{xry}$.
(the lost ray intersections just lead to null pointers at the selected point)
\item[\pc{Delete}($a$)] Delete the auxiliary point~$a$ (actually only used for a ray intersection) from the representation of the (shortcut) hull.
\item[\pc{Insert}($a,n$)] The (input or auxiliary) point~$a$ of the currently existing hull~$A$ (wlog) (some special case for empty ds), $n$ is a new input or auxiliary point.
  It is guaranteed that for an input point $n$, the neighbor~$a$ is also an input point in the new situation.
  For an auxiliary point, it is guaranteed that~$n$ is on the same edge of the (input) hull as~$a$.
\item[\pc{nextRight}$_M(a)$] Given the point~$a\in A$ (wlog), return the next point (auxiliary) point to the right of~$a$ of the type specified in~$M$ 
\end{description}


\subsubsection{Invariants in and not in a construction site}

The reestablishing algorithm of Section~\ref{sec:algRepl} and~\ref{sec:mainReestablish} defines a construction site that encloses parts of the certificate that are not yet reestablished.
The following invariants are maintained unless the points are enclosed in such a construction site:
\begin{itemize}
\item all equality points are identified,
\item a splitter that ends at an equality point is open, and
\item every point in a splitter has at least one color.
\end{itemize}

The following invariants are maintained even in construction sites
\begin{itemize}
\item separation of shortcuts,
\item separation of selected point certificates, and
\item constantly many auxiliary points on edges.
\end{itemize}
\vote{mention more?}

\subsubsection{Navigation}\label{sec:navigation}

\begin{figure}[t]
  \begin{center}
    \input{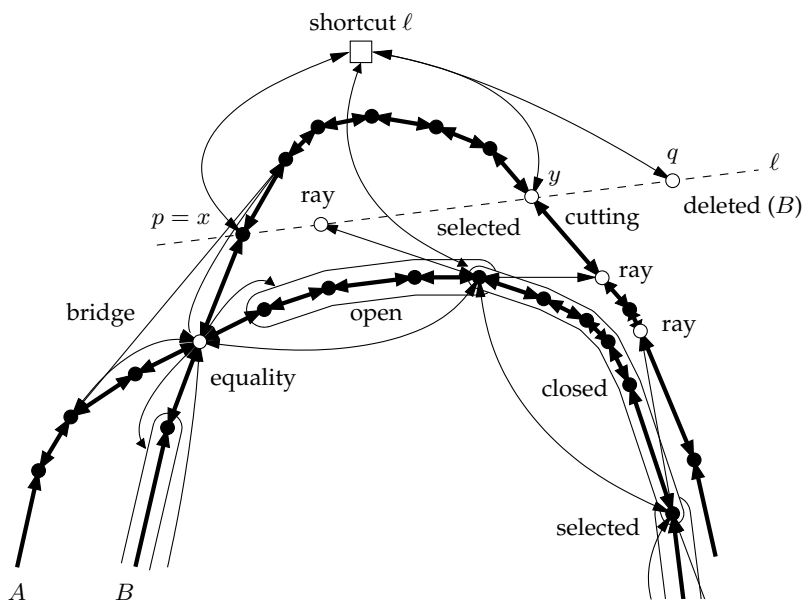}
    \caption
    {The pointers and records of the representation of a certificate.}
    \label{fig:representation}
  \end{center}
\end{figure}

Figure~\ref{fig:representation} illustrates all the pointers
stored. Black points are the points in $A$ and $B$.  All pointers are
bidirectional except for the pointers from selected and equality
points to the adjacent splitters.

\subsection{Potential for the runtime analysis}\label{sec:DSpotential}

Remember that the performance of the merger data structure is amortized constant per inserted point.
Our argument for this running time is based on a potential function. 

One explanation of the operations of the data structure is to see it as one algorithm that is interrupted when a separation certificate is established, and continued when the next replace operation takes place, very much in the spirit of the continued search in the splitter.
Accordingly, we define the potential function for the data structure in a way that is valid both during the reestablishing algorithm and in the static situation between replace operations.

In Section~\ref{sec:monotonicity} \vote{move this there?}
the concept of the life-cycle of a point was introduced, based on a geometric monotonicity inherent in the situation. 
Here, we turn this idea into a potential function.
Because the details of the replace algorithm are not yet introduced, we base this on a parameter~\CSP{} that reflects the stages a point can go through within the replace algorithm. 

\begin{definition}[Potential of the certificate]\label{def:potentialCert}
  Given a parameter~$\CSP$, the potential is the following sum over all input points (in $A$ and $B$):
  Every non-selected locally inside input point (in a splitter) has potential $\CSP+3$, every selected point has potential 3.
  Every locally outside that is hidden by a bridge has potential 2, a point on the merged upper hull has potential 1.
\end{definition}

Observe that in particular, every locally inside point has potential at least~5.

\section{Algorithm Replace}\label{sec:algRepl}\label{sec:DeleteCases}
This section describes how to handle a \pc{Replace} operation.
As detailed in Section~\ref{sec:merger-interface} and illustrated in Figure~\ref{fig:merger-example}, recall that the interface is:
 
\pc{Replace}$_B(r,L)$ - Given the locator to a point $r\in B\cap\UV(A\cup B)$ and a list $L=(b_1,b_2,\ldots)$ of points replacing $r$ in $B$, return a list $L'$ of locators to points replacing $r$ in $\UV(A\cup B)$, i.e., the points $\UV(A\cup (B\setminus\{r\}) \cup L) \setminus \UV(A\cup B)$.

Before going into the details, let us illustrate one of the challenges of the task, namely the variety of cases that we try to treat uniformly.
Consider the possibilities how equality points can be affected by the deletion.
We classify by the number of equality points on $\overline{xr}$ and $\overline{ry}$ before the replace operation, where $x$ and~$y$ are the former \vote{previous, so far, late} neighbors of~$r$ on~$\UH(B)$.
Because~$B$ is convex and we assume general position, a single edge of~$B$ can contain at most 2 equality points.
We do not consider mirrored situations (for cases (b), (d) and (e)) separately, leading to the $6$ cases illustrated in Figure~\ref{fig:cases}. 

\begin{figure}[htb]
  \begin{center}
    \input{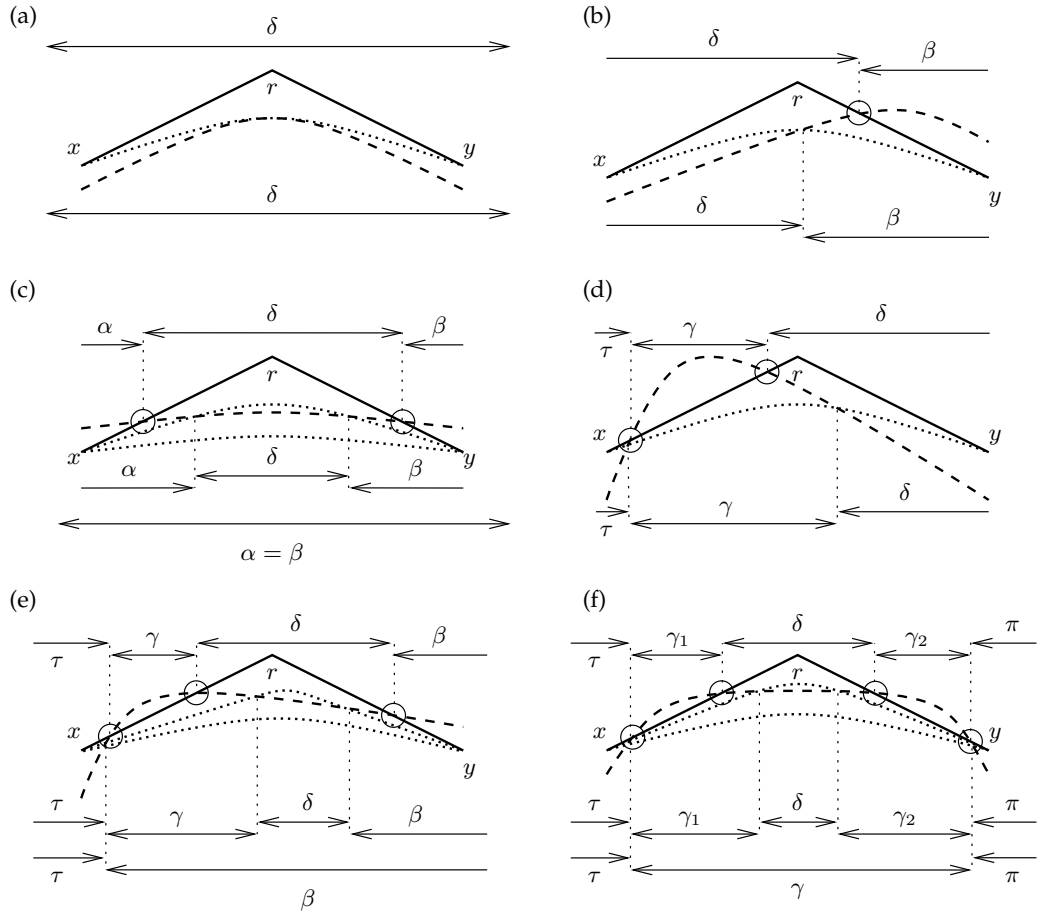}
    \caption[Different cases of lost equality points]{%
Different cases of lost equality points and new streaks.  
The arrows above the situation indicate the extent of streaks before the replace operation,
arrows below the situ\-ation afterwards.
The streaks named~$\gamma,\gamma_1,$ and $\gamma_2$ before the replace operation are trivial (the locally inner hull is entirely on $\overline{xr}$~or~$\overline{ry}$, there are no selected points and the splitters are empty, the shortcut outer hull consists of less than~6 segments).
In the cases~(c), (e)~and~(f) there are two outcomes possible, e.g.\ in case~(c) we might have to join the two non-trivial streaks~$\alpha$ and~$\beta$.
In the new streaks named~$\delta,\tau,\pi$ there can always be an even number of new additional equality points, whereas in the new streaks $\alpha, \beta$ and $\gamma$ there are no equality points.
The streaks $\alpha, \beta, \gamma, \gamma_1$ and $\gamma_2$ always expand their extent, possibly merging with another streak.
}
\label{fig:cases}
  \end{center}
\end{figure}
\clearpage
\begin{figure}[htb]
  \begin{center}
    \input{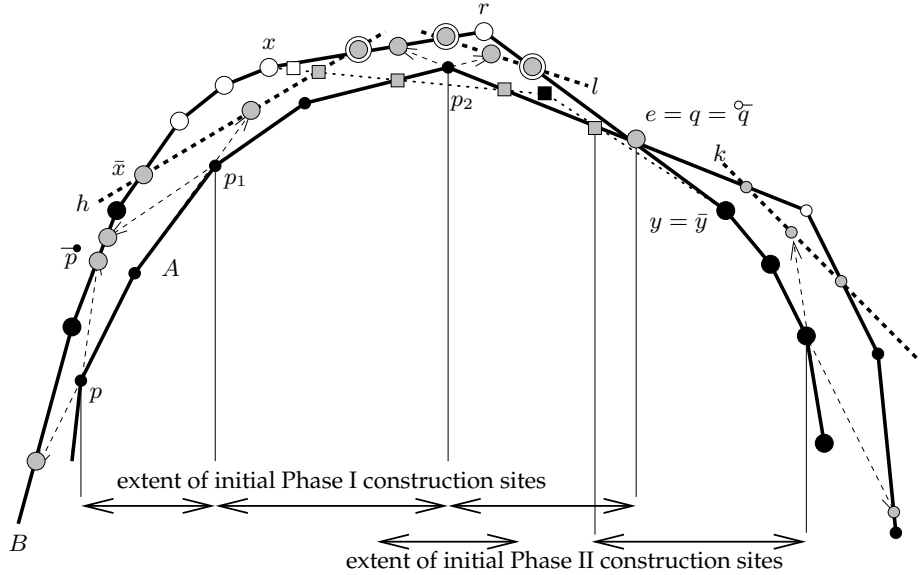}
    \caption[Algorithm using Construction Sites]{%
      The point~$r\in B$ gets replaced by a list~$L$ of points (black and white squares).
      White points are input points cut away by one of the shortcuts $h,l\in H_B$ and $k\in H_A$.
      Gray points and gray squares are non-input points, namely  cutting points, ray intersections and equality points.
      Small points are points in~$A$ or on $\SC A$, large points and squares are in $B$ or on $\SC B$.
      Before the replace operation the segments $\overline{xr}$ and $\overline{ry}$ are cut by the two shortcuts $l$ and~$h$.
      Considering shortcuts, the points $\SC r \subset \SC B$ (doubled circled points, here all cutting points)
      are replaced by $\SC L$ (gray and black squares). 
      Further, there is the lost equality point~$e$.
      Section~\ref{sec:mainReestablish} introduces the shown construction sites and notation like~$\RrayP{p}$.
}
\label{fig:constructionsites}
  \end{center}
\end{figure}

More sources of variety need to be handled:
If $r$ is the endpoint of a bridge, this bridge needs to be adjusted. 
Additionally, the corresponding equality point can also be affected (but it does not have to be), and it could even disappear (see e.g. Figure~\ref{fig:cases}~(c)).
Similarly, a ray intersection could change because the point defining its slope is removed.
It could also disappear because the selected point shooting the ray is no longer locally inside (Figure~\ref{fig:constructionsites}, point~$p_2$).

In the following, we discuss how a call to \pc{Replace} is implemented to produce the described return values and update the representation of the certificate detailed in Section~\ref{sec:Representation}.
Reestablishing the separation certificate (see also Section~\ref{sec:CompleteCertificate}) is the core of this task and detailed in Section~\ref{sec:mainReestablish}.
Maintaining correct bridges  is described in Section~\ref{sec:handleBridges} (see also Section~\ref{sec:bridges}) and ensuring aggressive shortcuts is described in Section~\ref{sec:SCAlgorithms} (see also Section~\ref{sec:sc_geom}).
An example of how the whole algorithm executes is given in Figure~\ref{fig:constructionsites}.

The guiding principle of this exposition is to unify the algorithm as much as we can, and use concepts that make it immediate that  we focus on what is doable within the time budget. 
In particular, we avoid (perhaps obvious) optimizations if this unifies or simplifies the algorithms and does not exceed the time budget.

\subsection{Replace$_B$}
\vote{change typesetting of return value in pseudocode}
\code{replace}
The pseudocode of the \pc{Replace} Algorithm, Algorithm~\ref{alg:replace} is augmented by the interfaces to the algorithms of Sections~\ref{sec:SCAlgorithms}, and \ref{sec:mainReestablish}.
Beyond the function call, the algorithmic task is sketched by the situation before the call (marked with ``\textbf{Pre:}'') and when returning from the call (``\textbf{Post:}'').
It is important to note that our approach (or at least the presentation) is not functional, and the main interaction and communication between the different functions is by side effect.
In other words, the main purpose of the functions is to change (update) parts of the representation of the separation certificate.
Accordingly, return values and parameters are mainly points of the certificate that enclose affected parts of the certificate.
The main reestablishing task is split up into two phases (Lines~\ref{algrep:reestI} and~\ref{algrep:reestII}).
First, we reestablish a certificate wherever $B$ is above~$A$, identifying all new or changed equality points in the process. 
This is the heart of the whole construction that achieving overall amortized constant time by avoiding to scan along any part of the upper hull of~$A$ that is still inside~$B$.
The variables/parameters $e_x$ and $e_y$ (Line~\ref{algrep:ifReest}) are only different from~$\bot$ if the corresponding point~$x$ (or respectively $y$) is inside~$A$. 
In that case, we take the (single) equality point on~$\overline{xr}$ as a point of~$A$ that we use to start searching for the new equality point delimiting a streak of polarity~$B$ above~$A$ if it exists.
Then we  reestablish a certificate wherever $A$ is above~$B$, a considerably easier task because here both participating hulls can be scanned along (see also Section~\ref{sec:monotonicity}, both points on $A$ and~$B$ advance in their life-cycle, as made precise with the potential in Section~\ref{sec:representation-details}).
As part of the algorithm a \defword{working slab} is maintained that encloses the part of the certificate that might not be valid.
It is storeded in the variables $\SC{x}$ and $\SC{y}$ that enclose~$\SC L$, the changes to~$\SC B$.
The points $x',y' \in B$ enclose all new or changed parts of the certificate of polarity $B$ above~$A$, and $x'',y''\in \SC A$ changed parts of the certificate with polarity $A$ above~$B$.

These two tasks work on the shortcut version of the hulls.
Accordingly, we start by applying shortcuts (modifying the deletion to reflect currently existing shortcuts) in Line~\ref{algrep:applySC}, and finish by creating new shortcuts that are aggressive in Line~\ref{algrep:createSC}.
Similarly, the handling of bridges happens as an initial removing of bridges (Line~\ref{algrep:rmBr}), and a final searching of new bridges (Line~\ref{algrep:searchBr}).
To prevent the bridge searching from scanning over the same part of one of the hulls several times, we remember the previous extent of a bridge in a so called bypass (see also Figure~\ref{fig:protectors}).
Such bypasses allow the bridge searching to work the same for all equality points (new or old) that have a new or changed bridge.

The analysis of termination and amortized runtime is intertwined with the exposition of the algorithm and is centered around the potential. 
Note that it is legitimate to charge runtime or initial potential of construction sites to the points of~$L$ several times.
Between the two phases, there will be a potential of 4 for every locally outside point of~$A$ in the working slab and not protected by a bypass.
After reestablishing the certificate, before creating new shortcuts, there will be a potential of 3 for every (freshly) locally outside point of~$A$ that is in the working slab.
Before bridge finding, there will be a potential of 2 for every (freshly) locally outside point of~$A$ that is in the working slab.

We present the different ingredients not in the order in which they are used, but in an order that reduces the number of forward references and that tries to first introduce the concepts in their simplest form. 
\vote{mention polarity?}
We start by explaining how searching for a bridge works (Algorithm~\ref{alg:FindBridge}), which makes it natural what we need to remember about lost bridges and equality points (Algorithm~\ref{alg:rmBr}).
Similarly, we start by introducing the simplest version of the recursive reestablishing (Algorithm~\ref{alg:refineAB}) based on the concept of construction sites in an easy situation, 
then we show its generalization (Algorithm~\ref{alg:refineBA}), the heart of the construction. 
Finally we present the initialization (Algorithms~\ref{alg:reestablishBA} and \ref{alg:reestablishAB}) that connect to Algorithm~\ref{alg:replace}.
\clearpage

\subsection{Bridge Searching}\label{sec:BridgeSearch}\label{sec:handleBridges}

As part of the \pcreplace{} algorithm (Line~\ref{algrep:searchBr} of Algorithm~\ref{alg:replace}), new bridges must be identified, either because new equality points of $\UH(B)$ and $\UH(A)$ come into existence or the endpoint of an existing bridge has been deleted.
By the geometric considerations of Section~\ref{sec:bridges}, we know that there is precisely one bridge above each equality point. 
Here, we assume that equality points requiring a new bridge are identified.
To find such a bridge, we use a variant of the bridge finding algorithm of Overmars and van Leeuwen~\cite{OvL81}, achieving amortized constant time per point on the participating hulls.
The local situation at an equality point is that of two hulls separated by a vertical line and we can use the geometric insights of~\cite{OvL81}:
Given a line-segment (tentative bridge) connecting the two hulls, the tangents on the hulls allows on at least one of the two sides to decide on a direction towards the endpoint of the actual bridge.
In~\cite{OvL81}, this is used to guide two simultaneous binary searches, one on each hull.
Here, we use a degenerate version of such a binary search. 

We start at the neighbors of an equality point.
Following the geometric guidance of~\cite{OvL81}, we move either on the left or on the right side further away from the equality point, until we find the bridge.

For some points, the bridge above it can change repeatedly over a sequence of \pcreplace{} operations. 
Then we should avoid scanning over this point several times because we can pay for the scan only by an advance in the points life-cycle, see Section~\ref{sec:monotonicity}.
We use a so called \defword{bypass}~$w$ consisting of the two points $w.s$ and $w.t$ on the same hull, say~$A$, where~$s$ is the horizontally closest to the equality point.
Assume the bridge-searching algorithm seeks a bridge that goes from~$B$ on the left to~$A$ on the right.
If this search considers a candidate bridge with right endpoint~$w.s$ and realizes that this endpoint is too far to the left, instead of moving to the right neighbor of~$w.s$ it takes as the next right endpoint of the candidate bridge~$w.t$.
In particular if~$w.t$ turns out to be too far to the right, the search continues to the left neighbor on~$A$ of~$w.t$.
This is consistent with bypasses being directed from $w.s$ to $w.t$ and because a binary search never returns to the same candidate, a bypass is used at most once.
More over, as long as a point remains properly below a bypass, it is not touched by the bridge search, unless it is no longer below a bridge. 

We place a bypass from the neighbor of the equality point to the former endpoint of the bridge or the neighbor of the deleted endpoint of the bridge, see also Figure~\ref{fig:protectors}.
The creation of bypasses is spelled out in Algorithm~\ref{alg:rmBr} (\pcrmBr{}). 
For the sake of uniformity of our algorithms, we also recreate a bridge if its equality point is affected by the deletion. 
This can happen for the equality point further away from~$r$ if a deleted segment has two equality points.
In that situation we place two bypasses that are both necessary to make sure that finding the already existing bridge again is only scanning over new points. 
Alternatively, we could keep such bridges, remember them, and do not start a bridge search from an equality point that is below an existing bridge, but instead link the old bridge with the new equality point.

In Line~\ref{alg:rmBr:foreach} of \pcrmBr{}, we loop over all bridges where either the bridge or the corresponding equality point is affected by the deletion.
These are at most 6 bridges (2 with endpoint~$r$ and 4 with a deleted equality point), considering 12 possible bypasses. 
Hence the number of bypasses is bounded by a 12 (actually, by geometry, by~6, but that is not important).
\vote{more details about rmBr?}
Because there is no need to keep bypasss between different deletions, we store them in local variables of Algorithm~\ref{alg:replace} (\pcreplace), where they are created when removing bridges using Algorithm~\ref{alg:rmBr} (\pcrmBr{}) and used when searching for new bridges using Algorithm~\ref{alg:FindBridge} (\pcFindBridge).

This achieves that the bridge finding scans only over points that have the potential to pay for it:
If the point remains under a bypass, it is not scanned over.
If it is under a bypass but is scanned over, it is no longer under a bridge and hence has an according decrease in potential to pay for scanning over the point.
Otherwise the point is freshly locally outside, i.e., a freshly locally outside point of~$A$ or a point of~$L$.
Hence, also in that case there is one unit of potential released to pay for the scan.
Everything else, like creating the bypasses, takes only constant time.

\code{FindBridge}


  \begin{figure}[htb]
    \begin{center}
      \input{protectors}
      \caption[Placing bypasses under obsolete bridges]
      {The situation of placing bypasses.  When processing
        the deletion of point~$r$ we have to prevent parts of the
        upper hulls of~$A$ and~$B$ that remains inside from being scanned again
        when we search for the new bridge.}
      \label{fig:protectors}
    \end{center}
  \end{figure}

\code{rmBr}

\clearpage

\subsection{Shortcut algorithms}\label{sec:SCAlgorithms}

In this section we describe the two steps of the replace operation that concern shortcuts.
The initial step (Line~\ref{algrep:applySC} of Algorithm~\ref{alg:replace}) is to apply the already existing shortcuts to the new segments of~$L$ on~$B$ before finding new equality points by establishing a new certificate without aggressive shortcuts.
The final step makes sure that the set of shortcuts is aggressive (Line~\ref{algrep:createSC} of Algorithm~\ref{alg:replace}) by inserting new shortcuts if necessary.

\subsubsection{Applying Existing Shortcuts}\label{sec:ApplySC}
During a \pcreplace$(r,L)$ operation, a shortcut defined for~$B$ by the line~$\ell\in H_B$ can be shortened or even become not effective because it is completely outside of~$\UH(B\setminus{r}\cup L)$.
Let the points $x,y \in B$ be the left and right neighbors of $r$ on $\UV(B)$.

Let~$\ell$ be the defining line of some effective shortcut on~$B$.
If~$\ell$ does not intersect $\overline{xr}$ or~$\overline{ry}$, then there is no change to the shortcut, either because $x,r,y$ are all above~$\ell$, or because they are all below~$\ell$ (the shortcut is to the left of~$x$ or the right of~$y$).
Otherwise $\ell$ has a lost cutting point and only for points of~$L$ it is not a priori clear if they are above~$\ell$ or not.
The line~$\ell$ defines now a smaller shortcut on~$B$, and the new cutting point(s) are defined as the intersection of~$\ell$ and a segment of the upper hull of $\{x,y\}\cup L$.
If there is no point of~$B$ above~$\ell$ any longer, the shortcut~$\ell$ is not effective and is removed.
This situation is illustrated in Figure~\ref{fig:sc_appl}.
\vote{is this too verbose?}
\vote{do we need the figure in the new style?}

\begin{figure}[htb]
  \begin{center}
    \input{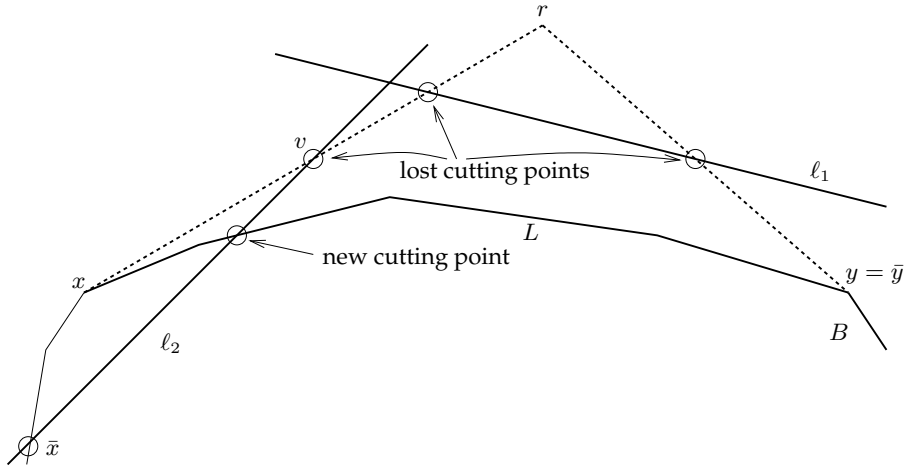}
    \caption
    {Applying shortcuts to~$B$ after the deletion of~$r$.  
      The shortcut defined by line~$\ell_1$ becomes not effective.
      The shortcut defined by~$\ell_2$ changes one of its cutting points and gets shorter.}
    \label{fig:sc_appl}
  \end{center}
\end{figure}

By Lemma~\ref{lem:LimitedImpact} there are at most 3 existing shortcuts that have a cutting point on $\overline{xr}$ or~$\overline{ry}$ and are affected.
Altogether, the changes to the shortcut version of the hull are the deletion of up to 4~consecutive points (cutting points, $r$) that are replaced by a (possibly empty) list of points, some of which are new cutting points defined by previously existing shortcuts.
This new part of~$\SC B$ is called $\SC L$ and is enclosed by the points~$\SC x$ and $\SC y$.

The algorithm \pcApplySCn{}, Algorithm~\ref{alg:ApplySCn}, establishes all new cutting points of existing shortcuts.
Additionally, it removes from the representation of~$B$ ray intersections (on~$\SC B$) and equality points that are lost. 
The corresponding selected points and lost equality points on~$A$ are marked to have lost their correspondence on~$B$.
The selected points of~$A$ might either be outside~$B$ now, or they will remain selected, and the corresponding ray intersection will be determined in \pcreestablishBA, Algorithm~\ref{alg:reestablishBA}, Section~\ref{sec:initBA}. 
For uniformity, this is true even if there is a ray intersection on a modified shortcut that would remain valid.
Similarly, the lost equality points on~$A$ will be deleted from the data structure later.
This means that both hulls are intact again, only the connections are disturbed.
More precisely, we determine the closest two points $x'$ and~$y'$ of~$B$ that are unchanged by the deletion (not defined by $\overline{xr}$ or $\overline{ry}$).
Further, we create base records for all the new input and cutting points, and reestablish previously existing shortcuts.
This approach is made precise in 
Algorithm~\ref{alg:ApplySCn}.

For the running time analysis it suffices to note that because the number of involved shortcuts is constant, the overall amortized time is linear in~$|L|$.

\code{ApplySCn} 

\subsubsection{Reestablishing Aggressive Shortcuts}\label{sec:newSC}

The structure of \pcreplace{} (Algorithm~\ref{alg:replace}) is that the establishing of the new complete certificate (\pcreestablishBA, Algorithm~\ref{alg:reestablishBA} and \pcreestablishAB, Algorithm~\ref{alg:reestablishAB} in Section~\ref{sec:mainReestablish}) is separated from creating new shortcuts.
\vote{too verbose?}
The already existing set of shortcuts on~$A$ and $B$ might not be aggressive as defined in Section~\ref{sec:CompleteCertificate}. 
In this section, we show how to reestablish that shortcuts are aggressive by creating new shortcuts (perhaps changing existing ones that have no longer ray intersections).
See Fig.~\ref{fig:createSC}.

\code{CreateSC}

The pseudocode of \pcCreateSC (Algorithm~\ref{alg:CreateSC}) is based on a scan of the locally outer hull (potentially alternating between $A$ and $B$) between two points $u_0$ and~$v_0$.
Aggressive shortcuts require that in the extent of every point certificate and tangent certificate there are at most 5 segments of the shortcut version of the locally outer hull. 
For one such certificate, let $u$ and $v$ be the corresponding ray intersections or equality points on the outer hull that define the extent of such a certificate.
Geometrically, the segment~$\overline{uv}$ is outside the locally inner hull, and hence any shortcut above~$\overline{uv}$ is conservative, see also Lemma~\ref{lem:shortcutAllowed}.
On the other hand, all existing shortcuts between $u$ and~$v$ that do not contain a ray intersection can be deleted without violating monotonicity of shadows (cf. Section~\ref{sec:monotonicity}).
Hence, we create new shortcuts (see also Figure~\ref{fig:createSC}) by the two input points between~$u$ and~$v$ that are furthest away from each other and are not above a shortcut that needs to remain because it defines $u$ or~$v$.

\begin{figure}[tb]
    \begin{center}
      \input{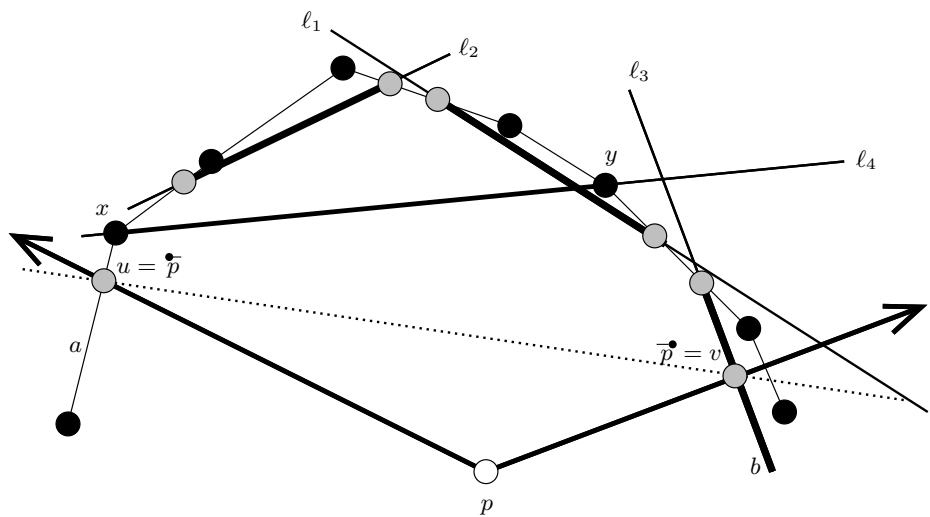}
      \caption[Creating a new shortcut]
      {The situation of creating a new shortcut~$\ell_4$ replacing the shortcuts~$\ell_1$ and~$\ell_2$.  The na\"{\i}ve and most
        aggressive choice~$\overline{uv}$ is forbidden by shortcut
        separation ($\overline{uv}$ and $\ell_3$ would intersect). 
        The new shortcut~$\ell_4$ preserves monotonic
        ray intersections and is defined directly by two input
        points.}
      \label{fig:createSC}
    \end{center}
\end{figure}

\subsubsection{Running Time Considerations}

As noted before, every point on the shortcut locally outer hull has one unit of potential available to pay for creating new aggressive shortcuts. 
This is clearly enough to make the additional overhead only constant.




\section{The main reestablishing algorithm}\label{sec:mainReestablish}

The algorithmic task discussed in this section is the reestablishing of all point and tangent certificates as part of a \pc{replace} operation between two points $\SC x, \SC y\in \SC{B}$, where the previously existing shortcuts have been applied.
This includes finding all new equality points and by this identifying all streaks of the new geometric situation.
In the representation of~$A$, there might be lost equality and lost ray intersection points.
The tasks correspond to the geometric considerations of Section~\ref{sec:geomReestablishing}.
\vote{other refs needed? not really}

As a result, the certificate and its representation is updated.
Additionally, two unchanged (input, equality or ray intersection) points on the locally upper hull are returned, that enclose all possible changes (by these procedures, i.e., ignoring bridges) to the certificate, in particular ray intersections and equality points.

The algorithm consists of the following two phases (already introduced in Section~\ref{sec:algRepl}, namely Algorithm~\ref{alg:replace}, Line~\ref{algrep:reestI} and Line~\ref{algrep:reestII}):
\begin{description}
\item[Phase I] \pcreestablishBA{} finds all new equality points and establishes the certificates for all new/changed streaks of polarity $B$ over $A$;
  the central work is done by the recursive procedure \pcrefineBA.
\item[Phase II] \pcreestablishAB{} establishes certificates in all streaks of polarity $A$ over $B$;
  (the the streaks are identified in Phase~I)
  the central work is done by the recursive procedure \pcrefineAB.
\end{description}

Observe that Phase~I can be understood as the algorithmic heart of the whole construction: While it is acceptable to scan over the new points~$L$ of~$B$ (on the locally outer hull), the locally inner hull~$A$ is unchanged and can not be scanned over within the aimed at time bounds.
\vote{Keep?}
In contrast, for Phase~II it would be acceptable to scan over both hulls.
To emphasis the (geometric) similarity of the two tasks, we present for Phase~II a simplified version of the algorithm for Phase~I.

The structure of Sections~\ref{sec:refineAB}--\ref{sec:initAB} is the following.
We start with the simplest recursive refinement of Phase~II, i.e., with
polarity $A$ above~$B$ in Section~\ref{sec:refineAB}. 
Then we discuss the generalized refinement algorithm (for the case of polarity $B$ above~$A$) in Section~\ref{sec:refineBAq}.
It can discover and handle inversions.
Finally we describe the initialization of Phase~I in Section~\ref{sec:initBA} and the more complicated case of initializing Phase~II in Section~\ref{sec:initAB}, where we also cover the situation of extending and joining streaks of polarity $A$ over~$B$.


Both phases start by defining initial \emph{construction sites} and then have a main recursion that continues until all slabs of the respective polarity within the construction site have a certificate established everywhere.
The role of a construction site is to enclose a slab and by this a piece of the certificate that might still need some repair, i.e., where the certificate is not yet verified to be valid.
\label{sec:genericConstructionSite}
Because we will use them in both polarities, we introduce them already here with the names $C$ and~$D$ (with $\{C,D\} = \{A,B\}$). \vote{needed?}
\begin{definition}[Generic Construction Site]
  \label{def:genericConstructionSite}
  A \defword{construction site} of polarity $D$ over~$C$ is defined by the 5-tuple $(p,q,\RrayP{p},\LrayP{q},T)$.
  The points $p,q$ are on~$C$.
  The point $p$ is to the left of~$q$, and all points of~$C$ between them are stored in a single splitter~$M$. 
  The \defword{alignment intersection} $\RrayP{p}\in \UH(\SC D)$ is a point on the shortcut version of the upper hull of~$D$ that is aligned with~$p$.
  It can be a ray intersection if~$p$ is a selected point, an equality point at the same position as~$p$, or, a point on the same vertical line as~$p$. 
  The input points and cutting points of the shortcut hull~$\SC{D}$ between $\RrayP{p}$ and $\LrayP{q}$ are stored in a split-array~$T$.
  The \defword{extent} of the construction site is the slab between~$p$ and $q$.

  There are no further selected points, ray intersections, lost or identified equality points in the data structure in the extent of the construction site.
  The extent of different construction sites does not overlap.
\end{definition}

This need for repair might be obvious from non-geometric considerations, for example if a selected point is missing one of its ray intersections.
In contrast, we can have situations that have the structure of a certificate, but the geometry has not yet been considered, like a selected point with established right ray with an open splitter to the right ending at an equality point, but where it is not a valid boundary certificate.
Importantly, checking if we actually look at a valid geometric certificate is a single test, in the example of a boundary certificate, checking that the ray intersection is above the tangent through the equality point.
When our algorithm performs this test, there are of course both outcomes possible.
In the positive case we consider the situation a valid certificate (and stop working there). 
In the negative case we typically gained some geometric insight that allows us to shrink or split the construction site and continue.
The algorithm can be understood as a recursive refinement procedure, that either calls itself with a single, simpler construction site (`closer' to valid) or two (possible more complicated but geometrically smaller) construction sites, or it finishes because the construction site is empty or a valid certificate.
Accordingly, the recursive functions are named \pc{refine}, and the initialization functions that perform the initial calls to them \pc{reestablish}, whose calls are already spelled out in Algorithm~\ref{alg:replace}.
\vote{make the invariants inside and outside CS precise? somewhat done}

In this context, the closed splitters can also be understood as a continuation of this adaptive refinement process between several replacement operations.
\vote{useful?}

\label{sec:potentialCS}

The runtime analysis of the algorithms working with construction sites are based on a potential of a construction site that extends the potential of the certificate in the static situation, where we defined~$\CSP$ as the potential of a non-selected point in a splitter.
Here, we need~$\CSP\geq5$, which means that we can set~$\CSP=5$.

\begin{definition}[Potential of a construction site]\label{def:potentialCS}
  Let $(p,q,\RrayP{p},\LrayP{q},T)$ be a construction site.
  The non-selected points stored in the splitter have a potential of~$\CSP$ each, as defined for the static situation.
  The construction site has potential of~\hbox{$|T|+2$}, i.e., the number of points stored in the split-array plus 1 for each boundary.
    There is an additional potential of 1 for each boundary that is formed by a vertical line only, i.e., the alignment point is not a ray intersection or an equality point.
    If the splitter of the construction site is closed, the potential is increased by~1.
\end{definition}

Observe that the potential of a construction site is linear in the number of participating points. 
The details will be discussed as part of the reestablishing algorithm.



\subsection{Phase~II Recursion: Refine Construction Sites Known as $A$ above $B$}\label{sec:refineAB}
In this section we consider the simpler refinement algorithm for Phase~II (of Algorithm~\ref{alg:reestablishAB}), operating on a streak of known polarity $A$ over~$B$.
The equality points have been established in Phase~I.
The task is to create a complete valid separation certificate as detailed in Section~\ref{sec:SeparationCertificate}.
\vote{state what and situation?}


A simple construction site is a special case of the construction site of Definition~\ref{def:genericConstructionSite}.

\begin{definition}[simple construction site guaranteed $A$ over $B$]\label{def:simpleConstructionSite}
  A \defword{simple construction site} of polarity $A$ over~$B$ is defined by the 5-tuple $(p,q,\RrayP{p},\LrayP{q},T)$ with $p,q\in B$, $T$ on~$A$, whose extend is part of a streak of polarity~$A$ over~$B$.
\end{definition}

Note that if the left endpoint of a construction site is the equality point~$e$, then it is identical to its left alignment intersection~$\RrayP{e}=e$.
In this situation $\RrayP{e}$ is not considered vertically above $e$.
Because we assume that all points are in general position, a selected point has $p\neq \RrayP{p}$.
\vote{here representation details?}

Algorithm~\ref{alg:refineAB} provides the pseudocode of the recursive function~\pcwo{refineAB} that has as argument a simple construction site.
The following Sections~\ref{sec:shootRay}--\ref{sec:searchSplit} explain the case distinction and actions summarized in that pseudocode.
The overall control flow is that the first applicable case is executed.

\code{refineAB}

\subsubsection{Establish Ray Intersection at Selected Point}\label{sec:shootRay}\label{sec:shootRayBA}

(See Alg.~\ref{alg:refineAB}, Line~\ref{alg:refineAB:shootRay})
Assume wlog the left boundary $p\in B$ is a selected point but without established ray, i.e., $\RrayP{p}$ is vertically above~$p$. 
To establish the ray intersection $\RrayI{p}$, we identify the line~$\Rray{p}$ of the ray and scan on~$\SC{A}$ starting from $\RrayP{p}$ to the right until the intersection $\RrayI{p}$ with~$\Rray{p}$ (or reaching the other end of the construction site).
Then we establish $\RrayI{p}$ as part of the representation of~$\UH(\SC{A})$.
The scanned over points of~$A$ are removed from the front of the split array~$T$.
By the geometry of the situation, when $q$ is selected or an equality point, it is guaranteed that the ray intersection~$\RrayI{p}$ is to the left of~$\LrayP{q}$.
As required by the interface, the split-array~$T$ stores all points of~$\SC A$ between $\RrayP{p}$ and~$\LrayP{q}$, and we can recurse with the modified simple construction site.

The situation without an established left ray when~$q$ is a selected point is handled symmetrically.

The change in potential by reducing the size of the split-array can pay for the scan.
The constant overhead of the function call (up to the tail recursion) is paid for by the reduction in potential of the boundary.

\subsubsection{Finish at Valid Boundary Certificate}\label{sec:finishBoundary}

(See Alg.~\ref{alg:refineAB}, Line~\ref{alg:refineAB:finishBoundary})
The construction site between a selected point~$p\in B$ and an equality point~$e=q=\RrayP{q}$ can be a valid boundary certificate if the left ray~$\Lray{e}$ through~$e$ (given by the segment of~$\UH(B)$) passes below~$\RrayI{p}$ (see Section~\ref{sec:boundaryCert}, in particular Figure~\ref{fig:half-open} and observe that $A$ and~$B$ are exchanged in the exposition).
In this case, the construction site has a valid certificate and the function returns.
Note that this includes the situation that the streak is trivial (the splitter is empty, i.e., both $p$ and $q$ are equality points and on the same segment of~$\UH(B)$).
(geometrically, this is a special case of a valid boundary certificate, see Section~\ref{sec:trivial-streak}).

Otherwise, if the boundary certificate is not valid, at least one more point inside the construction site needs to be selected.
We use the split functionality of the splitter for this purpose.
The control structure is shared with selecting more points if a tangent certificate is not valid, as explained in Section~\ref{sec:searchSplit}.

If the boundary certificate is valid, the function call is paid for by the potential of the construction site.
Otherwise, selecting a point releases potential $\CSP \geq 5$ which pays for the potential of 2 in each of the two new boundaries of the construction sites and this function call.

\subsubsection{Selecting More Points / Finish Construction Site} \label{sec:searchSplit}\label{sec:finishAtEqBA}

(See Alg.~\ref{alg:refineAB}, Line~\ref{alg:refineAB:searchSplit})
If none of the special cases of the previous sections finishes the construction site, we are (finally) using the splitter to check if the construction site consists of a valid point certificate. 
We define the colors of the splitter by the geometric situation:
The point~$p$ defines the red color of~$M$.
If $p$ is a selected point, this is the shadow, the function~${\mathcal C}(u)$ checks if a point~$u \in B$ is red by determining if its left ray $\Lray{u}$ passes below~$\RrayI{p}$.
This determines the shadow correctly for~$u$ because it is by assumption inside~$\UCo(A)$.
This geometric definition leads in particular to the selected point~$p$ being colored red, as desired.
If~$p$ is an equality point there is no red colored point (${\mathcal C}(u)$ never returns red) because even the neighbor of an equality point is allowed to be selected.
Symmetrically, the point~$q$ defines the blue color.
Note that only if~$M$ is closed we have to worry about monotonicity of the colors, which we will address when creating such a construction site.

With ${\mathcal C}$ as the definition of colors, we perform a split operation on~$M$.
If the result of this split is~$\bot$ (this is impossible for an invalid boundary certificate), the interface of the splitter guarantees the existence of a valid tangent certificate (cf. Section~\ref{sec:NormalTangentCertificate}) of the whole construction site, and the function returns. 
Otherwise, the split operation returns a point~$o\in M$, and splits~$M$ at~$o$ into the open splitter $M_1$ and $M_2$. 
We select~$o$, which by the definitions of red and blue is allowed, and we create two new wrapper records for the new splitters and link them into the data structure.
This creates pointers between the extreme points of the splitters and the splitters, in particular between~$o$ and $M_1$ and~$M_2$ (denoted in the pseudocode as ``in data structure'').
We split~$T$ at $\LrayP{o}=\RrayP{o}$, being the point on the segment of~$\UH(\SC{A})$ intersecting the vertical line through~$o$, resulting in $T_1$ and~$T_2$.
The whole construction site is split vertically at~$o$ into a left one defined by $p$, $o$, $M_1$ and $T_1$, and a right one defined by $o$, $q$, $M_2$ and $T_2$.
We recurse on these two simple construction sites.

If the boundary certificate is valid, the function call is paid for by the potential of the construction site.
Otherwise, selecting a point releases potential $\CSP \geq 5$ which pays for the potential of 2 in the two new boundaries of the construction sites and this function call.





\subsection{Phase~I Recursion: Refine Construction Sites With Hypothesis $B$ above $A$}\label{sec:refineBAq}

In this section, we describe the recursive procedure \pcrefineBA{}, that for a general construction site between~$p,q\in\UH(A)$, identifies all streaks $B$ above~$A$, including equality points, between~$p$ and~$q$, and establishes valid certificates for them.
We say that this algorithm operates under the hypothesis that the whole construction site is part of a streak of polarity~$B$ above~$A$.
This is only a hypothesis because the replace operation on~$B$ might have led to an arbitrary number of new streaks of the opposite polarity.
Actually, the primary goal of the discussed algorithm is to determine the new geometric situation, in particular all new equality points.
Still, this hypothesis guides us well in the sense that it forces us to avoid scan operations on~$\UH(A)$.
If the hypothesis holds, this efficiently leads to a certificate.\vote{bad repeat?}
Otherwise, the hypothesis fails in a benign way, namely by the algorithm finding an inversion, that is a (to be selected) point of~$A$ that is now outside~$B$, and we can use a scan to find the extent of the newly discovered inverted streak.
With the new boundaries of that streak, the algorithm continues to establish a certificate. 

\subsubsection{General Construction Sites}  \label{sec:generalConstructionSite}

The following definition of a general construction site is a special case of the construction site of Definition~\ref{def:genericConstructionSite}.
It is possible that there are (not yet known, arbitrarily many) new equality points and streaks of polarity $A$ over~$B$ in the extent of such a construction site.
If the algorithm finds any pair of aligned points of opposite polarity, we say that there is an inversion.
The purpose of the algorithm working with this kind of construction site is to either discover these, or to establish a certificate that shows that they are not present.

\begin{definition}[General construction site hypothesized $B$ over~$A$]
  \label{def:generalConstructionSite}
  A \defword{general construction site} is a construction site  defined by the 5-tuple $(p,q,\RrayP{p},\LrayP{q},T)$ with $p,q\in A$, $T$ on~$B$,
  based on the hypothesis that the polarity is $B$ over~$A$.
  If $\RrayP{p}$ is below~$p$ it is on the vertical line through~$p$, an \defword{inverted} situation.
  If $p$ is an equality point~$\RrayP{p}=p$.
  If $p$ is a lost equality point, $\RrayP{p}$ is vertically below~$p$ (inverted).
  Similarly for $\LrayP{q} \in \UH(\SC B)$.
  The splitter~$M$ must be open unless both $p$ and~$q$ are selected points.
\end{definition}

It is easy to implement a construction site that fits to the representation of the certificate.
We spell out some of the details in Section~\ref{sec:reprCS}.

The following table summarizes the possible boundaries of a general construction site.
The first lines, marked with +1, induce a unit of extra potential as described in Section~\ref{sec:potentialCS}.
\begin{description}
\item[+1] \label{rf:loeq} a lost equality point (inversion)
\item[+1] \label{rf:slin} a selected point found outside the other hull (inversion)
\item[+1] \label{rf:slwo} a selected point, whose ray intersection
  with~$\UH(B)$ inside the construction site is not yet determined (no inversion)
\item[0] \label{rf:nweq} an equality point
\item[0] \label{rf:slwr} a selected point with determined ray intersection
\end{description}


\subsubsection{Interface of \pcrefineBA}
\code{refineBA}

Let us now be more specific about the interface of Algorithm~\ref{alg:refineBA}. 
It is a recursive algorithm, an augmented variant of Algorithm~\ref{alg:refineAB} in Section~\ref{sec:refineAB}.
The notion of a construction site is generalized as already introduced as Definition~\ref{def:generalConstructionSite}.
The generalized refinement process needs the capability to handle the situation that an inversion, i.e., a new streak of polarity~$A$ over~$B$ is discovered.
\vote{summarize operation?}
Additionally, \pcrefineBA{} must be ready to work with the already existing certificate, most prominently with the previously selected points as boundaries and already existing splitters that might be closed.
\vote{ref to init?}

Within one construction site, when returning, all parts of a (new) streak of polarity~$B$ over~$A$ will have a valid separation certificate.
The other parts of the construction site are now guaranteed to have polarity~$A$ over~$B$, and will get a separation certificate when \pcreestablishAB{} is called in Phase~II.

Algorithm~\ref{alg:refineBA} summarizes the recursion in pseudocode.
The following Sections~\ref{sec:refineBAfirst}--\ref{sec:refineBAlast}, together with the already discussed cases in Sections~\ref{sec:shootRay}--\ref{sec:searchSplit} reflect the case distinction, again tried in this order.

\subsubsection{Opening a Closed Splitter in the case of an Inversion}\label{sec:openInverted}
(See Alg.~\ref{alg:refineBA}, Line~\ref{alg:refineBA:open})
This case is invoked if both $p$ and~$q$ are previously selected points, $M$ is closed and at least one of $p$ and $q$ is inverted.
The purpose is to open the splitter to be able to shrink the construction site at the inversion, a situation that does not arise from recursive calls, but an important possibility as a first call in the initialization by Algorithm~\ref{alg:reestablishBA}.

Assume~$p$ is inverted. 
We define red (the shadow of~$p$) to be empty.
If~$q$ is also inverted, also blue is empty.
Otherwise $q$ has (because of the case in Line~\ref{alg:refineBA:shootRay}, discussed for a simple construction site in Section~\ref{sec:shootRay}) the ray intersection established, i.e., $\LrayP{q}=\LrayI{q}$, such that we can define that a point~$u$ is blue if its right ray passes below~$\LrayI{q}$.
This makes sure that (at least) the shadow of~$q$ is colored blue.
Because~$\Lray{q}$ is unchanged, $\LrayI{q}$ did not move away from~$q$, and this definition of blue is monotonic with respect to the last split operation of~$M$.
With this definition of colors, we continue as follows, similarly as in Section~\ref{sec:searchSplit}, sharing the code with the case of Section~\ref{sec:searchSplitBA}.

The algorithm continues with the split operation of~$M$.
With the discussed colors, this will return a point~$o \in A$ that is geometrically allowed to be selected because it is not in the left shadow of~$q$, and (importantly) two open splitters.
Observe that it is possible that the splitter decides on~$o=p$, a trivial split.
In that case the search in Line~\ref{alg:refineBA:vertical} for~$\RrayP{o}=\LrayP{o}$ in~$T$ is also not going to split~$T$, and one of the recursive calls happens with a geometrically trivial construction site with both endpoints inverted.
This is benignly handled in Line~\ref{alg:refineBA:scanEq} of the algorithm.
The same line will also handle the other recursive call, and because the splitter is now open the construction site can be shrunk.

If there is a split of the construction site, the to be selected point looses potential $\CSP$ that pays for the call and the new boundaries. 
If there is no split, it remains one construction site and the additional potential of 1 of the closed splitter pays for the call.

\subsubsection{Scan for Equality point}\label{sec:scanEq}\label{sec:refineBAlast}
(See Alg.~\ref{alg:refineBA}, Line~\ref{alg:refineBA:scanEq})
This case is invoked if $p$ is inverted and $M$ is open.
We shrink the construction site on the left by scanning for an equality point.
This is a simultaneous scan on $A$ and $B$, implemented by removing elements from $M$ and $T$, analyzing their geometry given by their neighbors in $\UH(A)$ and $\UH(B)$ respectively, searching for the new equality point~$e$.
If we do not find such an equality point before reaching~$q$ (the right boundary of the construction site), there is no (part of a) streak of polarity $B$ over~$A$ within this construction site, and we are finished.
Otherwise we create the found new equality point~$e$ as part of the data structure for $\UH(A)$ and $\UH(B)$. 
We recurse with this~$e$ as the new left endpoint, i.e., with the construction site between~$e$ and~$q$.

The mirrored situation at~$q$ is handled symmetrically.

The scan is paid for by the potential of points in the split array and (overpaid by $\CSP$) by the potential of points in the splitter.

\subsubsection{Selecting another point}\label{sec:searchSplitBA}\label{sec:refineBAfirst}
We start the discussion with the (easy) case that resembles Section~\ref{sec:searchSplit}, namely that of selecting another point~$o$.
This~$o$ is determined by a call to \pc{split} on~$M$ with the colors defined as in Section~\ref{sec:searchSplit} that guarantee that~$o$ is not in the shadow of a selected point.
We determine $\RrayP{o}=\LrayP{o}$ using the split operation of the split-array~$T$.
The important difference to the simple case is that now the hypothesis that the whole construction site consists of a single streak of polarity $B$ above~$A$ can fail.
This can happen in Line~\ref{alg:refineBA:vertical} if the point~$o$ determined by the splitter for selection ($o$ is not in any shadow/color) is inverted, visible by $\RrayP{o}=\LrayP{o}$ being below~$o$.
Perhaps curiously, our formulation of the algorithm leaves reacting to this inversion to the following two recursive calls, leading to the otherwise implausible cases that are handled before that. 
\vote{too verbose?}


The potential of the two new construction sites is that of the original construction site plus two new boundaries, but with one more selected point. 
The potential of the so far unselected inside point~$o$ drops by $\CSP$ which can pay for this and the function call.

\subsection{Phase~I Initializing: Construction Sites $B$ above $A$ (Hypothesized Polarity)}\label{sec:initBA}
Initializing general construction sites of hypothesized polarity $B$ above~$A$ is the first step of Phase~I.
This operates on the shortcut version of~$B$, where up to four points are replaced by a list of points $\SC{L}$ that are identified in the data structure by the points $\SC x \in \SC B$ and~$\SC y\in \SC B$, so far without any auxiliary points or connections to~$A$. 
While the corresponding part of~$A$ could be arbitrarily complicated, the previously existing certificate has only constant complexity.
It is easy to find two neighboring connections between the hulls that are unaffected by the change,
and the constantly many lost connections on~$A$ form the boundaries of the construction sites.
This initialization is summarized in pseudocode (including calling convention) as Algorithm~\ref{alg:reestablishBA} and discussed in the  following.
\vote{Remove $e_x,e_y$ as argument? It could be scanned for}
See also Figure~\ref{fig:constructionsites}.

The amortized running time of this procedure is $O(|L|)$, taking into account the potential in the data structure before (mainly in the points and splitters of~$A$) and in the constantly many created construction sites.

\code{reestablishBA}

\vote{connection to lines / subsections as before?}
Here, the calling convention includes only points of~$\SC B$, and only in the case of previously existing equality points, they are provided as arguments $e_x, e_y$.
Lost equality points and previously selected points of~$A$ between $\SC x$ and ~$\SC y$ are still part of the data structure of~$A$.
After establishing the initial construction sites, for all of them the refinement algorithm of Section~\ref{sec:refineBAq}  is called.
Hence, when returning, all new equality points are found and established as part of the certificate of polarity~$B$ over $A$.
The returned points on the locally outer hull enclose all new (or changed) equality points.

To prepare the construction sites of polarity $B$ over~$A$, we create a split array~$T$ for points of~$\SC B$, essentially containing the points of~$\SC L$, but sometimes a constant number of extra points.

By their definition, $\SC x$ and $\SC y$ define the slab with possible changes to $\SC B$.
We consider the constantly many existing point certificates whose extent intersects with the slab of $\SC x$ and~$\SC y$.
The lost equality points of~$\UH(A)$ and previously selected points of~$A$ are the internal boundaries of the construction sites.\vote{keep this repetition?}
For the leftmost and rightmost certificate that contain $\SC x$ and respectively $\SC y$ (or perhaps both), we extend the construction site to an (unchanged) connection between the two hulls, and use it as the boundary.

If~$\SC x$ is outside of~$\UH(A)$ we scan to the left along~$\SC B$ (importantly the shortcut version) until we find the first
\begin{itemize}
\item right ray intersection $\RrayI{p}$, leading to $p$ as the selected point and $\RrayP{p}=\RrayI{p}$ as boundary; or 
\item equality point~$e$.  We choose $p:=e$, and $\RrayP{p}=e$ as boundary.
\end{itemize}
 
If~$\SC x$ is inside of~$\UH(A)$ necessarily $\SC x=x$ and $\overline{xr}$ contains precisely one lost equality point~$e_x$.
In this case the (bottom) boundary of the construction site is set to $p=e_x$.
By the geometry of the situation we know that $x$ is to the left of the vertical line through~$p=e_x$, and we search through $\SC L\subset \SC B$ for the intersection $\RrayP{p}$ with the vertical line through~$e_x=p$.

With the symmetric arguments we determine the rightmost point $q\in A$ and its corresponding point $\LrayP{q} \in \UH(\SC{B})$.
We place the points of the new $\SC B$ between~$\RrayP{p}$ and $\LrayP{q}$ into one initial split-array~$T$.
These are at most constantly more points than~$\SC{L}$.

To create initial construction sites, we split the situation at all lost equality points and selected points $p_i \in A$.
We split~$T$ to find the corresponding $\LrayP{p_i}=\RrayP{p_i}$ on the vertical line through~$p_i$.
On each of them we call the recursive \pcrefineBA{} algorithm of Section~\ref{sec:refineBAq}.

Observe that this creates only constantly many construction sites:
There can be at most 4 lost equality points, and by Lemma~\ref{lem:LimitedImpact} at most eleven affected selected points. 
Hence, the initial potential of the construction sites is justified in different ways:
The potential in the boundaries is only constant, together with the potential in the split-arrays, this is $O(L)$.
The potential in the splitters (unselected points) comes directly from the potential of the static data structure before the replace operation.

We return $\RrayP{p}$ and~$\LrayP{q}$ as the unchanged points that enclose all changes on the locally upper hull.

\subsection{Phase~II Initializing: Extending and Establishing Streaks of Known Polarity $A$ over $B$}\label{sec:initAB}

It remains to describe how to create (initial) construction sites for all affected streaks of polarity~$A$ over~$B$, initializing the adaptive refinement procedure of Section~\ref{sec:refineAB}.
By Phase~I, all equality points and streaks of polarity~$B$ over~$A$ of the new situation are established, together with a valid separation certificate.
Perhaps, as in the case of joining streaks, this phase consisted of just one scan over the new and freshly surfacing points, showing that no such streak or equality point exists.
What remains is to create or repair the separation certificate for all affected streaks of polarity~$A$ over~$B$.
The situation, calling convention, and procedure is summarized in pseudocode as Algorithm~\ref{alg:reestablishAB}.

\code{reestablishAB}



The initialization starts from the two points~$\SC x,\SC y\in \SC B$, enclosing the replaced points~$\SC r$ on~$\SC B$.
(Line~\ref{alg:reestablishAB:scan}) A scan along the points of~$\SC{L}$ takes $O(L)$ time and will find all new or changed equality points.
These are stored to be returned.

(Line~\ref{alg:reestablishAB:newCS}) It is easy to deal with pairs of equality points in this region that enclose a streak of polarity $A$ over~$B$.
For such a streak we place all points of~$\SC A$ between the equality points into a split array~$T$, and all points of~$B$ between the equality points into a new splitter~$M$.
With this, we call the recursive refinement of Section~\ref{sec:refineAB}.
Note that this includes the streaks called~$\gamma$ in Figure~\ref{fig:cases}, case~(d), (e) and~(f).
Aggressive shortcuts make sure that the part of~$\SC A$ that was already outside~$B$ has constant size. 
Hence, the potential of~$T$ is linear in the number of points of~$A$ that are freshly outside~$B$ as a result of the replace operation, which means that their potential drops at least by 1 from 3 to 2, allowing to pay for this.
All points freshly placed in a splitter are points of~$L$, and any edge on~$L$ can have at most two equality points.
Hence the potential in the splitters and initial boundaries of all these new construction sites is~$O(L)$.

If the leftmost equality point~$e$ has left of it a streak of polarity $A$ over~$B$, we need to extend or adjust the certificate of an already existing streak of that polarity.
This can only happen if there was one equality point on~$\overline{xr}$, and hence $\SC x = x$.

(Line~\ref{alg:reestablishAB:adjustRay}) If~$x$ is a selected point of~$B$, its left ray changes, and we create a construction sites that has a right boundary at~$x$.
Remember that by our definitions, even if~$x$ should be the leftmost point of the situation, there is always an implicit equality point at infinity. 
Scanning for the new ray intersections of~$x$ with the shortcut hull~$\SC{A}$ touches only constantly many points of~$A$ that are not freshly surfacing.
Because $A$ and~$B$ do not change to the left of~$x$, the splitter to the left of~$x$ is still completely inside~$\UCo(A)$.
Hence, Lemma~\ref{lem:MonotonicShadows} guarantees monotonicity of shadows. 
With this, we call the recursive refinement of Section~\ref{sec:refineAB}.
The symmetric procedure works for a rightmost equality point whose right streak is of polarity $A$ over~$B$.
Such a construction site has constant additional potential ($T$ has constant size, the splitter has its potential from the static situation, and the boundaries are constant).

(Line~\ref{alg:reestablishAB:adjustEq}) Even if $x$ is not selected, it is the rightmost point in an open splitter~$M$ of a potentially invalid (ray through the equality point changed) right boundary certificate of a selected point~$p\in B$ (it could be~$p=x$, $M=\{p\}$).
We extend~$M$ to the right by the points of~$L\subset B$ between~$x$ and~$e$.
We scan to the left on~$\SC A$ starting at~$e$ and place the points into a split array~$T$ up to the ray intersection~$\RrayI{p}$, which is the first ray intersection we meet.
The potential in~$M$ is at most $O(|L|)$ and $T$ has constant size such that this step is amortized $O(|L|)$.

(Line~\ref{alg:reestablishAB:join}) A last possibility is that no new equality point exists.
If the replace happened completely outside~$A$, there is no new or changed streak of polarity $A$ over~$B$, and we are done.
Otherwise, we are in the case of two streaks of polarity $A$ over~$B$ joining, as illustrated in Figure~\ref{fig:cases}, case~(c).
Necessarily, we have $\SC x=x$ and $\SC y=y$ and $\SC L=L$.
Further, $x$ is the rightmost point of the open splitter $M_1$ of the boundary certificate at the selected point~$p$ of the lost equality point on~$\overline{xr}$.
Similarly, $y$ is the leftmost point of the open splitter $M_2$ of the boundary certificate at~$q$ of the lost equality point on~$\overline{ry}$.
We create a new closed (crucial!) splitter~$M$ by a join of $(M_1,L,M_2)$.
We place the points of~$\SC A$ between $\RrayP{p} = \RrayI{p}$ and $\LrayP{q}=\LrayI{q}$ into the split array~$T$.
This yields a construction site as in Section~\ref{sec:refineAB}. 
If $x$ or~$y$ are selected, their rays change, which is already handled in Line~\ref{alg:reestablishAB:adjustRay}.

\vote{section Analysis?}
It remains to argue for the initial potential of this joining construction site.
The points in the interior of the splitters are either from~$L$ (new), or have sufficient potential from before the replace operation.
Because we work only with the shortcut version of~$A$, the potential in the split-arrays of these construction sites is paid by points of~$A$ being no longer inside~$B$ (plus a constant), which releases, as noted before, a potential of 1 per point.

All in all, after the construction sites are finished, we established complete certificates for all new or changed streaks of polarity $A$ over~$B$.
To be able to return the working slab in the form of the leftmost and rightmost point of change on the locally upper hull, we keep track of the extreme points ever put in a split-array.

\subsection{Representation of construction sites}\label{sec:reprCS}\vote{needed?}

We implement a construction site in the straight forward way:
Because a splitter always contains the two enclosing (selected) points as first and last element, but not an enclosing equality point.
The connection to the data structure of the certificate for~$A$ is clear.
Only the pointer from the selected point to the two splitter containing it as extreme point need to be maintained explicitly.
For~$B$ this is slightly more complicated because $\RrayP{p}$ and $\LrayP{q}$ can be on the segments of~$\UH(\SC{B})$ and~$T$ can be empty.
This complication is easily addressed by storing a pointer to a neighboring point of $\RrayP{p}$ and $\LrayP{q}$ on~$\UH(\SC{B})$ and start operations by exploring the geometry of the neighbors of the pointers. 
This approach also works when splitting a construction site by calling the split operation of~$T$:
If the resulting split-arrays are non-empty we can find appropriate points from the new leftmost or rightmost element of the split-array.
If one of the resulting split-arrays is empty, appropriate pointers are neighbors of the corresponding pointer of~$T$.
Hence we can assume that we have direct access to $\RrayP{p}$ and $\LrayP{q}$ on~$\UH(\SC{B})$.

\section{Dealing with Degeneracy}
\label{sec:degenerate}

\vote{out, refer to standard techniques instead?}
So far we used the assumption that input points are in general
position to avoid dealing with degenerate cases.
This is convenient, as it allows us to concentrate on the important
situations instead of getting drowned by special cases.
Here we summarize the situations that benefited from this assumption,
and suggest how to modify the algorithms to correctly deal
with the degenerate cases.

In the data structure (and algorithm) we have to be prepared to find
points that act in two roles simultaneously, for example an
input point that is an equality point or an input point that is also a
ray intersection.

If the same points can be both in $A$~and~$B$, we might see a stretch
of~$\UH(A)$ that coincides with a stretch of~$\UH(B)$.
We can easily handle this by treating such a stretch as an extended
equality point.  Introducing one more stage in the life-cycle of a
point, there is no problem with the accounting.
We have to allow the deletion of points that are in such an equality
stretch.

If two oppositely directed rays intersect the other hull at the same
position, both the defining points are allowed to be selected (the
alternative would be to consider them to be in the shadow of each
other).
\vote{do we need to identify all cases that need extra treatment?}


In the bridge-finding we have to break ties in a way that the
resulting hull does not have two collinear segments.

\vote{argue that this is all?}

\section{Lower Bounds}
\label{sec:LowerBounds}

In this section we prove Theorem~\ref{thm:CHLB}, implying that any convex hull data structure with
$O(n^{1-\varepsilon})$ query time, for any constant $\varepsilon>0$,
must have amortized $\Omega(\log n)$ query and insertion time, i.e., we derive
lower bounds on running times that asymptotically match the quality of
the data structures we presented in the previous sections.
\vote{do we need to check if this improves Chazelle?}
The amortized running time functions in the following
theorems are defined by the property that the insertion of~$n$ 
elements (starting with an empty data structure) 
takes total~$n\cdot I(n)$ time.
Note that we do not require the functions to be non-decreasing.

Our method is reduction based.  We use 
a semidynamic insertion-only convex hull data structure 
to solve a parametrized decision problem, arriving at the
lower bound.  The lower bound on the decision problem holds for
algebraic computation trees~\cite{ben-or83}.  A real-RAM algorithm 
for a decision problem can
be understood as generating a family of decision trees, the height of
the tree corresponds to the worst-case execution time of the
algorithm.  This is the model used in the seminal work by Ben-Or, from
where we take the main theorem \cite[Theorem~3]{ben-or83} that bounds
the depth of a computation tree in terms of the number of connected
components of the decided set.  We consider the following decision
problem, a variant of element-uniqueness.

\begin{definition}
  For a vector~$z=(x_1,\ldots,x_n,y_1,\ldots,y_k) \in\Rset^{n+k}$ we
  have $z\in \DiSeP\subset\Rset^{n+k}$ if and only if $y_1 <
  y_2 < \cdots < y_k$ and for all~$i$ and~$j$ it holds~$x_i\neq
  y_j$.
\end{definition}
\vspace{-1.3ex}
\begin{lemma}
\label{lem:DSLB}
  For any natural numbers
  $k$~and~$n$ with $10\leq k\leq n$, the depth~$h$ of an algebraic
  computation tree (the running time of a real-RAM algorithm) deciding
  the set \DiSeP\ is lower bounded by~$h \geq \frac{1}{57}\cdot n\cdot\log k$.
\end{lemma}

\begin{proof}
  Let $y_i=i$ for $i=1,\ldots,k$.
  There are~$(k+1)^n$ ways of distributing the values~$x_i$ into the
  $k+1$ intervals formed by $\Rset\setminus \{y_1,\ldots,y_k\}$.
  No two vectors that have different
  such distributions can be in the same connected component of
  \DiSeP.
  Using \cite[Theorem~3]{ben-or83} this implies~$2^h
  3^{n+k+h}\geq(k+1)^n$.
  Taking the  base~2 logarithm and rearranging terms yields
  $h(1+\log 3) \geq n \log (k+1) - (n+k) \log 3 \geq 
  ( 1- \frac{2\log 3}{\log k})n\log k \geq 
  ( 1- \frac{\log 9}{\log 10}) n\log k $,
  where we use $\log (k+1) > \log k$ and $10\leq k\leq n$.
  The claimed lower bound follows from
  $(1-\frac{\log 9}{\log 10})/(1+\log3)\geq\frac{1}{57}$.
\end{proof}
  
\begin{definition}
  The \pc{semidynamic (kinetic) membership} problem asks for a data structure that
  maintains a set~$S$ of real numbers under insertions, and allows for
  a value~$x$ the membership query~$x\in S$.
  In the kinetic case the membership query~$x$ is required to be not smaller than
  any previously performed query.
\end{definition}

\begin{theorem}
\label{thm:KinLB}
  Let~$\cal A$ be a data structure for the \pc{semidynamic
    kinetic membership} problem.  
  Assume that the time for $n$
  \pc{Insert} operations is bounded by~$n\cdot I(n)$ and the
  amortized running time for the \pc{kinetic-find-min} query is
  bounded by~$q(n)$.
  Then we have
  $ I(n) = \Omega ( \log \frac{n}{q(n)} )$.
\end{theorem}
\begin{proof}
  Fix an arbitrary~$n$ and choose the parameter $k= \floor{n/q(n)}$.
  If $k\leq 9$ we have $\log\frac{n}{q(n)}\leq\log 10$, hence the
  theorem is true for such~$n$.%
  \vote{discuss that always $I(n)=\Omega(1)$? Hence -2 does not hurt}
  Otherwise we have $k\geq 10$, and we consider the following reduction
  from \DiSeP.
  Let
  $z=(a_1,a_2\ldots,a_n,b_1,b_2\ldots,b_k) \in \Rset^{n+k}$
  be some input to \DiSeP.
  We check in time~$k$ whether we have~$b_1< b_2 < \cdots < b_k$.  
  If this is not the case, we reject.
  We insert all the~$n$ values~$a_i$ into~$\cal A$.
  Then we perform~$k$ queries for~$b_1,\ldots,b_k$.
  If one of the queries returns~$b_j\in S$, i.e., $b_j=a_i$ for
  some~$i$ and~$j$, we reject, otherwise we accept.
  This algorithm correctly solves the \DiSeP\ problem.

\pagebreak[2]  
  By Lemma~\ref{lem:DSLB} and the running time of this algorithm we
  get
  $ I(n)\cdot n + (q(n)+1) \cdot k \geq \frac{1}{57} \cdot n \cdot\log k $.
%
  Using our choice of~$k$ we get 
  $I(n)\cdot n+n+k\geq \frac{1}{57}\cdot n\cdot\log\floor{n/q(n)}$.  
  Dividing by~$n$ and rearranging terms
  yields $ I(n) \geq \frac{1}{57}\cdot\log( \floor{n/q(n)}) - 2$.
\end{proof}

\begin{theorem}
\label{thm:MembLB}
  Consider a data structure for the \pc{Semidynamic
    Membership} problem on the real-RAM supporting
  \pc{Member} queries in amortized $q(n)$ time, 
  for size parameter~$n$. 
  Then we have
$ q(n) = \Omega( \log n ) $.
\end{theorem}

\begin{proof}
  Let~$I_k$ be the worst-case time it takes to insert~$k$ elements into an
  initially empty data structure,  and let~$n= I_k + k$.
  We make a reduction from \DiSeP{} as follows.
  Let
  $z=(a_1,\ldots,a_n,b_1,\ldots,b_k) \in \Rset^{n+k}$
  be an input to \DiSeP.  We check in time~$k$ whether we
  have~$b_1< b_2 < \cdots < b_k$.  If this is not the case,
  we reject.  
  We insert the values~$b_1,\ldots,b_k$ into the data structure.
  Now we perform~$n$ queries for the values $a_1,\ldots,a_n$.
  This algorithms correctly solves~\DiSeP.
  By Lemma~\ref{lem:DSLB} we get for $k\geq 10$ 
  the inequality
  $ k + I_k + n\cdot q(k)  \geq \frac{1}{57}\cdot  n\cdot \log k$.
  Using~$k +I_k = n$ and dividing by~$n$ we get
  $ q(k) \geq \frac{1}{57} \log k - 1$.
\end{proof}

A data structure for the \pc{Semidynamic Predecessor Problem} maintains a
set~$S$ of real numbers under insertions, and allows queries for~$r$,
reporting the element~$s\in S$, such that $s\leq r$, and there is
no~$p\in S$ with~$s<p\leq r$.
From Theorem~\ref{thm:KinLB} and Theorem~\ref{thm:MembLB} follows the
next corollary. 
\begin{corollary}
  Consider a data structure implementing the \pc{Semidynamic
    Predecessor Problem} on the real-RAM that supports
  \pc{Predecessor} queries in amortized~$q(n)$ time, and \pc{Insert} in
  amortized~$I(n)$ time for size parameter~$n$.
 Then we have
$ q(n) = \Omega( \log n )$ and
  $I(n) = \Omega ( \log \frac{n}{q(n)} ).$
\end{corollary}

\begin{theorem}
\label{thm:KinHeapLB}
  Consider a kinetic heap data structure.  Assume that 
	the time for $n$
	\pc{Insert}
  operations is bounded by~$n\cdot I(n)$ and the amortized running time for
  the \pc{kinetic-find-min} query is bounded by~$q(n)$.  Then we have
  $ I(n) = \Omega ( \log \frac{n}{q(n)} )$.
\end{theorem}
\begin{proof}
  We use a kinetic heap to solve \pc{semidynamic kinetic
  membership}.   
For an insertion of~$a_i$ we insert the tangent on the curve $y=-x^2$ at the
point~$(a_i,-a_i^2)$ with slope $-2a_i$.
For a member query~$b_j$ we perform \pc{kinetic-find-min}$(b_j)$. 
See Figure~\ref{fig:lower-bound}, left. 
If and only if the query returns the tangent line through~$(b_j,-b_j^2)$, we
answer~``$b_j\in S$''.
The Theorem follows from Theorem~\ref{thm:KinLB}.
\end{proof}

\begin{figure}[htb]
  \begin{center}
    \input{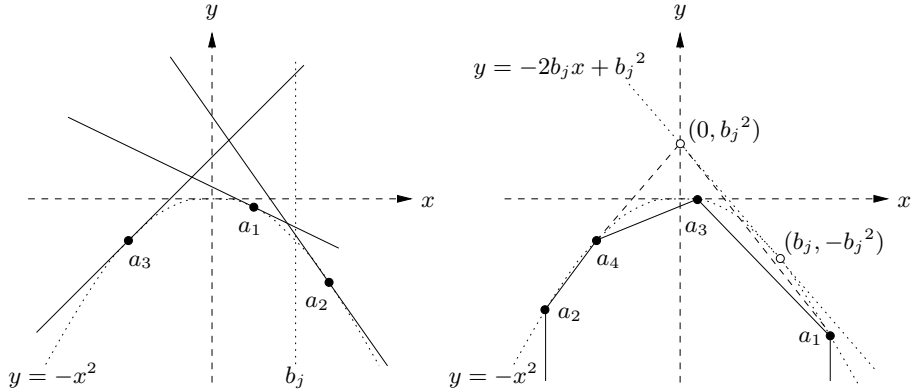}
    \caption[Lower bound reductions]
    {Lower bound reductions: Kinetic find min queries (left) and 
    tangent queries (right).}
    \label{fig:lower-bound}
  \end{center}
\end{figure}

Finally we conclude with the proof of the main lower bound theorem.

\begin{proof} (of Theorem~\ref{thm:CHLB})
  A semidynamic insertion-only convex-hull data structure can by duality be used
  as a kinetic heap (Section~\ref{sec:duality}).
  The lower bound on the insertions follows from  Thoerem~\ref{thm:KinHeapLB}.
  The lower bound on the queries follows from Theorem~\ref{thm:MembLB} 
  by reducing semidynamic membership queries to lower envelope queries using
  the same geometric reduction as in Theorem~\ref{thm:KinHeapLB}
  (Figure~\ref{fig:lower-bound}, left).
\end{proof}

Note that for~$q(n) = O(n^{1-\varepsilon})$, Theorem~\ref{thm:KinLB}
implies $I(n)=\Omega(\log n)$. 
Another example is that $I(n)=O(\log\log n)$ yields 
$q(n)=\Omega( n / \log^{O(1)} n)$. 
Theorem~\ref{thm:KinLB} shows that the amortized insertion times of
the data structure of Theorem~\ref{thm:KinHeap} and
Theorem~\ref{thm:fast} are asymptotically optimal.

If instead the convex-hull data structure provides only
tangent-queries, the same lower-bounds hold.
This follows because we can solve the semidynamic membership problem
by mapping a finite set $S\subset\Rset$ to the convex point set
$\{(a_i,-{a_i}^2)\mid a_i\in S\}$, and a query $b_j$ to the tangent
query for the point $(0,{b_j}^2)$.
The value $b_j$ is in $S$ if and only if one of the tangents is
the line~$y=-2b_ix+{b_i}^2$ (see Figure~\ref{fig:lower-bound}, right).


\subsection{Trade-Off}
\label{sec:tradeoff}

%
For query time $O(n^{1-\varepsilon})$, the above lower bounds state
that query time and update time $O(\log n)$ is the best possible. Our
dynamic planar convex hull data structure matches these bounds and for
set membership queries the bounds are matched by balanced binary
search trees. In the following we present data structures that match
the lower bound for combinations of insertion and query times, where
the insertion time is $o(\log n)$.

There is one simple idea achieving a trade-off between insertion and
query times: we maintain several (small) search structures and
insert into one of them.  In return the query operation has to query
all the search structures.
We will describe the predecessor problem and use balanced search trees
as the underlying data structure.  

We choose a parameter function~$s(n)$ that tells the data structure
how many elements might be stored in one search tree.  We assume
that~$s(n)$ is easy to evaluate (one evaluation in~$O(n)$ time
suffices) and non-decreasing.
The data structure keeps two lists of trees, one with the trees that
contain precisely~$s(n)$ elements and the other with the trees
containing less elements.  For an \pc{Insert}$(e)$ operation we
insert~$e$ into one of the search trees that contains less than~$s(n)$
elements.  If no such tree exists, we create a new one.  When~$s(n)$
increases, we join the two lists (all trees are now smaller
than~$s(n)$) and create an empty list of full search trees.  For a
query operation we query all the search trees and combine the result.

We achieve (amortized) insertion time $I(n)=O(\log s(n))$ and query time 
$q(n)=O(\frac{n}{s(n)} \log s(n) )$.  From Theorem~\ref{thm:KinLB} we
have that  given this query time the lower bound for the insertion time is
\vote{back to $-\log\log s(n)$?}
$$\Omega\left(\log \frac{n}{q(n)}\right) = \Omega\left(\log \frac{s(n)}{\log s(n)}\right)
= \Omega\big(\log s(n)^{1-\varepsilon}\big) = \Omega(\log s(n))\;,$$
i.e.,  we have achieved 
optimal (amortized) insertions.

If we are interested in data structures for the membership,
predecessor or convex hull problems that allow queries in~$q(n)$ time
for a smooth, easy to compute function~$q$,  then this technique
allows us to achieve data structures with asymptotically optimal
insertion times (because, in the terminology of
Overmars~\cite[Chapter~VII]{OvermarsPhD}, all these problems are decomposable search
problems).

If we want to accommodate deletions, a deletion should delete the
element from the search structure containing the point,
combined with standard global rebuilding whenever
$n$ has doubled or is halved. 

\section{Open Problems}\label{sec:openProblems}

It remains open whether a data structure achieving worst-case~$O(\log
n)$ update times and fast extreme-point queries exists.  
It is also unclear if other queries (like the segment of the convex
hull intersected by a line) can also be achieved in~$O(\log n)$ time,
or if it is possible to 
report $k$ consecutive points on the convex hull in $O(k+\log n)$ time.
Furthermore it would be desirable to come up with a simpler data
structure achieving matching running times.
\todo{last to do}
\vote{last vote}


\bibliographystyle{../sokalpha} 

%
\bibliography{ch}



\clearpage
\tableofcontents
\clearpage
\section{Notation}
\begin{tabular}{cl}
  Notation & Description \\
  \hline
  $n$ & number of points, $n=|S|$ \\
  $k$ & output size ($k$ consecutive points on hull), $k$-level \\
  $S$ & current set of points \\
  $S_1, S_2, \ldots $ & subset of points \\
  $A,B$ & subsets maintained by a merger \\
  $\SC{A},\SC{B}$ & shortcut versions of $A$ and $B$ \\
  $r$ & point $r\in B$ replaced by $L$ \\
  $L$ & list of points to replace $r$ on $B$ \\
  $L'$ & new visible points due to a replace \\
  $l,r$ & left and right (e.g.\ $\ell_l$, $\ell_r$) \\
  $x,y$ & left and right neighbors of $r$ \\
  $e,e_1,e_2$ & equality point \\
  $b_1,b_2,...$ & bridges \\
  $P$ & set of bypasses \\
  $p,q$ & point (selected) \\
  $x,y$ & points, typical neighbors \\
  $\overline{xy}$ & line segment \\
  $u_0,v_0$ & points, define enclosure of replace \\
  $w=(w_s,w_t)$ & bypass in $P$ \\
  $\ell,\ell',\ell_1,\ell_2$ & lines, lines defining shortcuts \\
  $h_\ell$ & halfspace below line $\ell$ \\
  $r,h,\ell$ & root,height,leaf in join-tree \\
  $H=(H_A,H_B)$ & shortcuts on $A$ and $B$ \\
  $Q=(Q_A,Q_B)$ & selected points \\
  $\alpha$ & inverser Ackerman, slope of line \\
  $b(\cdot)$ & number of barrier levels \\
  $q(n)$ & query time \\
  $d$ & query direction \\
  $I(n)$ & insertion time \\
  $t,t'$ & time ($x$-value) in kinetic/parametric data structure \\
  $\IT$ & interval tree \\
  $\varepsilon$ & constant $\varepsilon>0$ \\
  $\UV(X)$ & vertices on upper hull of $X$ \\
  $\UC(X)$ & upper closure of $X$ \\
  $\UCo(X)$ & interior of upper closure of $X$ \\  
  $\UH(X)$ & upper hull of $X$ \\
  \hline
\end{tabular}


\end{document}